\documentclass[conference]{IEEEtran}
\usepackage{amsmath,amssymb,amsthm}
\usepackage{mathtools}
\usepackage{bussproofs}
\usepackage{xspace}
\usepackage{xcolor,colortbl}
\usepackage{xifthen}
\usepackage{hyperref}
\usepackage[all]{xy}

\theoremstyle{definition}
\newtheorem{theorem}{Theorem}
\newtheorem{definition}[theorem]{Definition}
\newtheorem{proposition}[theorem]{Proposition}
\newtheorem{lemma}[theorem]{Lemma}

\newtheorem{example}[theorem]{Example}

\newtheorem{claim2}{\sc Claim}
\newenvironment{claim}{\begin{claim2}\rm}{\end{claim2}\rm}
\newenvironment{claimfirst}{\setcounter{claim2}{0}
               \begin{claim2}\rm}{\end{claim2}\rm}

\newenvironment{pfclaim}{\begin{trivlist}\item[]{\sc Proof of
Claim}}{\hfill {\mbox{$\blacktriangleleft$}}\end{trivlist}}

\newcommand{\margincomment}[1]{}

\newenvironment{tbs}{%
   \small\tt
   \begin{itemize}}{\end{itemize}}
\newcommand{\btbs}{\begin{tbs}}                                                                      
\newcommand{\etbs}{\end{tbs}}

\newcommand{\hide}[1]{}

\newcommand{\rem}[1]{}

\usepackage{enumerate}
\usepackage{stmaryrd}

\newcommand{\comp}{\mathop{;}}
\newcommand{\cross}{{\times}}
\newcommand{\diam}[1]{\langle {#1} \rangle}
\newcommand{\gcha}{\sqcup}    
\newcommand{\gchd}{\sqcap}    

\newcommand{\iter}[1]{#1^*}
\newcommand{\diter}[1]{#1^\cross}
\newcommand{\test}[1]{#1?}
\newcommand{\dtest}[1]{#1!}

\newcommand{\lra}{\leftrightarrow}

\newcommand{\Ra}[1][]{%
  \ifthenelse{\isempty{#1}}{\Longrightarrow}{\overset{\small #1}{\Rightarrow}}
}
\newcommand{\LRa}[1][]{%
  \ifthenelse{\isempty{#1}}{\Leftrightarrow}{\overset{\small #1}{\Leftrightarrow}}
}

\newcommand{\tr}[1]{#1^\sharp}  
\newcommand{\trg}[1]{\tau_{#1}} 
\newcommand{\itrg}[2]{\trg{#2}^{#1}}
\newcommand{\nnf}[1]{\overline{#1}}
\newcommand{\startr}[2]{{#2}^{{#1}\bullet}}
\newcommand{\startrg}[3]{\beta({#2},{#1},{#3})}
\newcommand{\dd}[1]{\langle\!\langle#1\rangle\!\rangle}

\newcommand{\eps}{\varepsilon}

\newcommand{\Var}{\mathit{Var}}
\newcommand{\FVar}{\mathit{FVar}}

\newcommand{\Cl}[1]{\mathit{Cl}(#1)}

\newcommand{\lessAL}[1][]{
  \ifthenelse{\isempty{#1}}{<^-}{<^-_{#1}}
}
\newcommand{\lesstrAL}[1][]{
  \ifthenelse{\isempty{#1}}{<}{<_{#1}}
  }
\newcommand{\leqAL}[1][]{
  \ifthenelse{\isempty{#1}}{\leq}{\leq_{#1}}
  }

\newcommand{\leqfp}[1][]{%
  \preccurlyeq
}
\newcommand{\lessfp}[1][]{%
  \prec
}

\newcommand{\restr}[2]{{#1}{\upharpoonright}_{#2}}
\newcommand{\ula}[1]{\underleftarrow{#1{!}}}
\newcommand{\ulag}[2]{\underleftarrow{#1{!} \cdot #2}}

\newcommand{\nam}[1]{\ensuremath{\mathsf{#1}}}
\newcommand{\var}[1]{\ensuremath{\mathit{#1}}}
\newcommand{\names}[1][]{
  \ifthenelse{\isempty{#1}}{\ensuremath{N}}{\ensuremath{N_{#1}}}
}

\newcommand{\Clo}[1][]{  
  \ifthenelse{\isempty{#1}}{\ensuremath{\mathsf{Clo}}\xspace}{\ensuremath{\mathsf{Clo}(#1)}\xspace}
} 
\newcommand{\CloM}[1][]{  
  \ifthenelse{\isempty{#1}}{\ensuremath{\mathsf{CloM}}}{\ensuremath{\mathsf{CloM}_{#1}}}
} 
\newcommand{\CloG}[1][]{  
  \ensuremath{\mathsf{CloG}}\xspace
} 

\newcommand{\monML}{\ensuremath{\mathsf{mon\text{-}ML}}\xspace}

\newcommand{\GLms}{\ensuremath{\mathsf{G}}\xspace}

\newcommand{\GameOpRules}{\ensuremath{\mathsf{GameOp}}\xspace}
\newcommand{\DeepGRules}{\ensuremath{\mathsf{DeepG}}\xspace}

\newcommand{\HilbP}{\ensuremath{\mathsf{Par}}\xspace} 
\newcommand{\HilbFull}{\ensuremath{\mathsf{Par}_{\mathsf{Full}}}\xspace} 

\newcommand{\AtProp}{\mathsf{P}_0}
\newcommand{\AtGame}{\mathsf{G}_0}

\newcommand{\langPar}{\mathcal{L}_{\mathsf{Par}}}
\newcommand{\formPar}{\mathcal{L}_{\mathsf{Par}}}
\newcommand{\gamePar}{\mathcal{G}_{\mathsf{Par}}}

\newcommand{\langNorm}{\mathcal{L}_{\mathsf{NF}}}
\newcommand{\formNorm}{\mathcal{L}_{\mathsf{NF}}}
\newcommand{\gameNorm}{\mathcal{G}_{\mathsf{NF}}}
\newcommand{\langNormAnn}{(\langNorm)^{\Ann}}

\newcommand{\langFull}{\mathcal{L}_{\mathsf{Full}}}
\newcommand{\formFull}{\mathcal{L}_{\mathsf{Full}}}
\newcommand{\gameFull}{\mathcal{G}_{\mathsf{Full}}}

\newcommand{\Ann}{\mathsf{Ann}}
\newcommand{\langM}{{\mathcal{L}^{\mu}_\mathsf{NF}}}
\newcommand{\langMAnn}{(\langM)^{\Ann}}

\newcommand{\langBimod}{{\mathcal{L}^{2\mu}_\mathsf{NF}}}
\newcommand{\langBimodAnn}{(\langBimod)^{\Ann}}

\newcommand{\dnnf}[1]{\mathsf{nf}(#1)}
\newcommand{\nf}[1]{\dnnf{#1}}
\newcommand{\trParNorm}[1]{\dnnf{#1}}

\newcommand{\pari}[1]{\mathsf{pa}(#1)}

\newcommand{\trNormPar}[1]{\pari{#1}}

\newcommand{\leftprooftransformarrow}{\longleftarrow}

\newcommand{\RuAnd}{\ensuremath{\land}\xspace}
\newcommand{\RuClo}{\ensuremath{\mathsf{clo}}\xspace}
\newcommand{\RuD}{\ensuremath{\mathsf{D}}\xspace}

\newcommand{\RuLFP}{\ensuremath{\mu}\xspace}
\newcommand{\RuGFP}{\ensuremath{\nu}\xspace}
\newcommand{\RuWeak}{\ensuremath{\mathsf{weak}}\xspace}
\newcommand{\RuCompD}{\ensuremath{\comp_d}\xspace}

\newcommand{\RuInd}{\ensuremath{\mathsf{ind}}\xspace}
\newcommand{\RuIndS}{\ensuremath{\mathsf{ind}_s}\xspace}

\newcommand{\RuNuClo}{\ensuremath{\nu\text{-}\mathsf{clo}}\xspace}
\newcommand{\RuExp}{\ensuremath{\mathsf{exp}}\xspace}

\newcommand{\RuCross}{\ensuremath{\cross}\xspace}
\newcommand{\RuStar}{\ensuremath{*}\xspace}
\newcommand{\RuQuestion}{\ensuremath{?}\xspace}
\newcommand{\RuBang}{\ensuremath{!}\xspace}

\newcommand{\RuMonDgame}{\ensuremath{\mathsf{Mon^g_d}}\xspace}
\newcommand{\RuMonDfmla}{\ensuremath{\mathsf{Mon^f_d}}\xspace}

\newcommand{\RuMonMod}{\ensuremath{\mathsf{mod}_m}\xspace}

\newcommand{\RuOr}{\ensuremath{\lor}\xspace}

\newcommand{\RuChoiceA}{\ensuremath{\gcha}\xspace}
\newcommand{\RuChoiceD}{\ensuremath{\gchd}\xspace}
\newcommand{\RuTestA}{\ensuremath{?}\xspace}
\newcommand{\RuTestD}{\ensuremath{!}\xspace}
\newcommand{\RuComp}{\ensuremath{\comp}\xspace}

\newcommand{\RuMP}{\ensuremath{\mathsf{MP}}\xspace}
\newcommand{\RuMon}{\ensuremath{\mathsf{Mon}}\xspace}
\newcommand{\RuBarInd}{\ensuremath{\mathsf{BarInd}}\xspace}

\newcommand{\RuBdi}{\ensuremath{\mathsf{BarInd}^\cross}\xspace}

\newcommand{\AxOne}{\ensuremath{\mathsf{Ax1}}\xspace}

\newcommand{\Ax}{\ensuremath{\mathsf{Ax}}\xspace}

\renewcommand{\phi}{\varphi}

\newcommand{\isdef}{\coloneqq}   

\newcommand{\bbS}{\mathbb{S}}
\newcommand{\mng}[1]{\llbracket{#1}\rrbracket}
\newcommand{\lfp}{\mathsf{lfp}}
\newcommand{\gfp}{\mathsf{gfp}}

\newcommand{\mm}[1]{\mathcal{M}(#1)}

\newcommand{\ol}[1]{\overline{#1}}
\newcommand{\dg}[1]{\widetilde{#1}}
\newcommand{\wh}[1]{#1}

\definecolor{greeen}{HTML}{008833}

\newcommand{\Fcross}{F^\cross}
\newcommand{\Scross}{S^\cross}
\newcommand{\subterm}{\unlhd}
\newcommand{\ssubterm}{\lhd}

\usepackage{ifthen}

\newtheorem{definition_th}{Definition}[section]  

\newtheorem{algo_th}[definition_th]{Algorithm}

\newenvironment{proofof}[1]{\bigskip\noindent{\bf Proof of #1}}{\vspace{0.5cm}}

\ifCLASSINFOpdf
\else
\fi
\begin{document}
%
\title{Completeness for Game Logic\footnote{TEsting}}


%
\author{\IEEEauthorblockN{Sebastian Enqvist\IEEEauthorrefmark{1},
Helle Hvid Hansen\IEEEauthorrefmark{2},
Clemens Kupke\IEEEauthorrefmark{3},
Johannes Marti\IEEEauthorrefmark{4} and
Yde Venema\IEEEauthorrefmark{5}}
\IEEEauthorblockA{\IEEEauthorrefmark{1}Stockholm University, SE-10691
Stockholm. Email: thesebastianenqvist@gmail.com}
\IEEEauthorblockA{\IEEEauthorrefmark{2}Delft University of Technology, P.O. Box 5015, NL-2600 GA Delft. Email: h.h.hansen@tudelft.nl}
\IEEEauthorblockA{\IEEEauthorrefmark{3}University of Strathclyde, 16 Richmond St, Glasgow G1 1XQ, UK. Email: clemens.kupke@strath.ac.uk}
\IEEEauthorblockA{\IEEEauthorrefmark{4}Universit\"at Bremen, Postfach
330440, 28334 Bremen. Email: johannes.marti@gmail.com }
\IEEEauthorblockA{\IEEEauthorrefmark{5}University of Amsterdam, P.O. Box 94242, NL-1090 GE Amsterdam. Email: y.venema@uva.nl}}



\IEEEoverridecommandlockouts
%
\IEEEpubid{\makebox[\columnwidth]{\copyright2019 IEEE.\hfill} \hspace{\columnsep}\makebox[\columnwidth]{ }}
%
%

\maketitle
\pagestyle{plain}

\begin{abstract}
Game logic was introduced by Rohit Parikh in the 1980s as a generalisation of 
propositional dynamic logic (PDL) for reasoning about 
outcomes that players can force in determined 2-player games.
Semantically, the generalisation from programs to games is mirrored by moving from Kripke
models to monotone neighbourhood models.
Parikh proposed a natural PDL-style Hilbert system which was
easily proved to be sound, but its completeness has thus
far remained an open problem.

In this paper, we introduce a cut-free sequent calculus for game logic,
and two cut-free sequent calculi that manipulate annotated formulas, one
for game logic and one for the monotone $\mu$-calculus, the variant of
the polymodal $\mu$-calculus where the semantics is given by monotone
neighbourhood models instead of Kripke structures.
We show these systems are sound and complete, and that
completeness of Parikh's axiomatization follows.
Our approach builds on recent ideas and results by Afshari \& Leigh (LICS 2017)
in that we obtain completeness via a sequence of proof transformations
between the systems.   
A crucial ingredient is a validity-preserving translation from game logic to the 
monotone $\mu$-calculus.
\end{abstract}


%
\IEEEpeerreviewmaketitle

\section{Introduction}
\label{sec:intro}

\subsection{Game logic, background and motivations}
Game logic was introduced by Parikh in the 1980s \cite{Parikh85}
as a modal logic for reasoning about the outcomes that players can
force in determined 2-player games.
We refer to the two players as \emph{Angel} and \emph{Demon},
following \cite{PaulyParikh:GL}. 
A modal formula $\diam{\gamma}\phi$ should be read as, \emph{``Angel has a 
strategy in the game $\gamma$ to ensure an outcome in which $\phi$ holds''}. 

Syntactically, Parikh's game logic is an extension of propositional dynamic logic 
(PDL) \cite{FL1979:PDL} as games are composed from atomic games and 
constructors that denote sequential composition of games,
as well as choice, iteration and test for Angel,
and finally the dual operator which denotes swapping the roles of the two players.
In Parikh's original language,
the strategic ability of Demon is thus only implicitly expressed
through the dual operator,
and PDL programs can be viewed as 1-player games (played by Angel).
Semantically, the generalisation from 1-player games to 2-player games is 
obtained by moving from Kripke structures to monotone neighbourhood structures.
Game logic is thus a non-normal, monotone modal logic.

Just as PDL can be translated into the (normal) modal $\mu$-calculus 
\cite{carr:pdl2014},
game logic can be naturally translated into the monotone modal $\mu$-calculus
\cite{Pau:phd}, and from there into the normal modal $\mu$-calculus for the
language that has two normal modalities for each monotone 
modality \cite{KrachtWolter99}. This was already sketched by Parikh 
in \cite{Parikh85}, and later improved in \cite{PaulyParikh:GL,Pau:phd} to show 
that the satisfiability of game logic is in \textsc{exptime}. 
We refer to \cite{Parikh85} and the survey \cite{PaulyParikh:GL} for applications
of game logic and further results.

\subsection{A landscape of logics for games} 
Parikh's game logic is probably the first of a family of logics 
designed to reason about different aspects of games.
\rem{On the one hand there are modal logics for multi-player games that can express strategic powers
of groups of agents such as ATL~\cite{Alur:2002:ATL} and Coalition Logic~\cite{Pauly2002:CL}. 
On the other hand, there are logics that do focus on 2-player games but go beyond game logic by treating
strategies as first-order objects such as strategy logics~\cite{krhepi:stra10, Mogavero2010:SL}.
}
Since then,
modal logics for multi-player games that can express strategic powers
of groups of agents have appeared
such as ATL~\cite{Alur:2002:ATL} and Coalition Logic~\cite{Pauly2002:CL}. 
There are also logics that focus on 2-player games but go beyond game logic
such as strategy logics~\cite{krhepi:stra10, Mogavero2010:SL},
which treat strategies as first-order objects,
and dGL~\cite{pla:diff15} which combines game operations and first-order quantification for hybrid games.

\subsection{The challenge of completeness for game logic}


It is a long-standing open question whether a complete proof system for
game logic exists.
The completeness result for dGL in~\cite{pla:diff15} is of a rather different nature,
since it concerns the completeness of a non-recursively enumerable logic relative to some oracle logic.
Parikh proposed in \cite{Parikh85} a natural-looking
PDL-like Hilbert system $\HilbP$,
but a proof of its completeness has thus far remained an open problem.
Only (relatively easy) partial results were known: completeness for the
dual-free fragment \cite{Parikh85}, and for the iteration-free fragment
\cite{PaulyParikh:GL,Pau:phd}. Giving a completeness proof similar to
the one for PDL from \cite{koze:elem81} using canonical models seems
impossible for the full language of game logic as such a proof
essentially involves a filtration argument. It is not difficult to see,
however, that game logic is not well-behaved with respect to
filtrations. 

The difficulty of showing completeness for the entire language of game
logic can perhaps be 
explained by the fact that in the presence of both angelic iteration and dual,
game logic (when interpreted over Kripke frames) spans all levels of the 
alternation hierarchy of the (normal) modal $\mu$-calculus \cite{Berwanger03}.
This is in stark contrast with PDL, LTL and CTL$^*$  which are all contained in 
low levels of the alternation hierarchy. 
Over Kripke models, game logic is thus a highly expressive fragment of the modal
$\mu$-calculus for which completeness is highly involved. The classical automata-based approach to the
completeness of the $\mu$-calculus from \cite{koze:resu83,walu:comp93} relies on the existence of
``disjunctive'' normal forms in the language of the $\mu$-calculus. 
It is unlikely that a similar normal form
can be defined for the more rigid game logic syntax, as occurrences of
the $\cross$-operator introduce greatest fixpoint operators that are
invariably tied to conjunctions.

\subsection{Main results and approach} 
In this paper, we introduce three cut-free sequent calculi, two for game logic and one for the monotone $\mu$-calculus, that we show all to be sound and complete.
The first of these is the system for game logic $\GLms$ which is a cut-free sequent calculus with deep inference rules. We show that $\GLms$ is complete, and that this implies completeness of Parikh's Hilbert system.
One of the rules in $\GLms$ is a so-called \emph{strengthened induction rule},
which is inspired by the strengthened induction rule 
in~\cite{AfshariLeigh:LICS17}, and somewhat similar to
Kozen's context rule~\cite[Proposition~5.7(vi)]{koze:resu83}. 
Our approach relies on game logic being able to express this rule.
Just as it is convenient to work with $\mu$-calculus formulas in negation normal
form, the system $\GLms$ works on game logic formulas in \emph{normal form}, 
where negation may only be applied to atomic propositions, and the dual game
operator only to atomic games.
Consequently, the system $\GLms$ is defined for the normal form language
$\langNorm$ which contains demonic game constructors as primitives.
Given a game logic formula $\phi$, $\nf{\phi}$ is the formula obtained by bringing $\phi$ into dual and negation normal form.

The second system for game logic, called $\CloG$, is a cut-free sequent calculus with a closure rule.
In $\CloG$, game logic formulas from $\langNorm$ are annotated with \emph{names} for formulas of the form $\diam{\gamma^\cross}\phi$.
These names keep track of unfoldings of these greatest fixpoint formulas, and
together with the closure rule they facilitate the detection of 
repeated unfolding of greatest fixpoints formulas in the same context
(which closes the proof tree branch).
Technically, this is achieved by imposing side conditions on
the closure rule in $\CloG$.
These side conditions involve an order $\leqfp$ on the set $F$ consisting of game logic formulas of the form $\diam{\gamma^*}\phi$ or $\diam{\gamma^\cross}\phi$.
These {game logic fixpoint formulas} will be in 1-1 correspondence with fixpoint variables when we translate into the monotone $\mu$-calculus.

The third system, $\CloM$, is a cut-free sequent calculus for the monotone 
$\mu$-calculus, also with a closure rule and name annotations.
This system is a monotone variant of the system $\Clo$ for the normal modal $\mu$-calculus introduced in \cite{AfshariLeigh:LICS17}.
In $\CloM$, the side conditions are expressed with the usual (priority/subsumption)
order $\leqAL$ on fixpoint variables where $x \leqAL y$ means that $x$ is of 
higher priority than $y$.
\medskip

%
Our approach to proving soundness and completeness builds on recent work
by Afshari \& Leigh. In \cite{AfshariLeigh:LICS17} they presented a
cut-free sequent calculus for the normal modal $\mu$-calculus, and
proved its completeness via a series of transformations through other
proof systems, including the system $\Clo$, and ending at the complete
tableaux systems with names developed in \cite{Stirling14tableaux} and
\cite{NJ10thesis}.
We prove completeness of the systems $\GLms$, $\CloG$ and $\CloM$  by showing
that we can transform derivations as follows:
\[ \HilbP
\stackrel{1)}{\leftprooftransformarrow} 
\GLms
\stackrel{2)}{\leftprooftransformarrow} \CloG 
\stackrel{3)}{\leftprooftransformarrow} \CloM
\stackrel{4)}{\leftprooftransformarrow} \Clo 
\] 

1) First, the transformation of
$\GLms$-derivations to $\HilbP$-derivations goes via an intermediate
 Hilbert system $\HilbFull$, which is an extension of $\HilbP$ to the
full language which has angelic as well as explicit demonic operations
and freely-placed duals and negations. These transformations are
relatively straightforward using that $\HilbP$ essentially has cut via
modus ponens.

2) The transformation of $\CloG$-derivations to $\GLms$-derivations requires non-trivial adaptations of the analogous result in \cite[Theorem VI.1]{AfshariLeigh:LICS17}. It uses a translation $\startr{}{(-)}$ that replaces annotations on game logic formulas with certain ``deep insertions of demonic tests'', which are the game logic analogues of the ``deep disjunctions'' of \cite{AfshariLeigh:LICS17}. 

3) The transformation of $\CloM$-derivations into $\CloG$-derivations relies on a novel translation $\tr{(-)}$ from game logic into the monotone $\mu$-calculus.
This translation is truth- and validity-preserving, it commutes with fixpoint unfolding, and crucially, it reflects the order on fixpoint variables in $\tr{\phi}$ into the order on fixpoint formulas in $F$.
Note that the translation of game logic from \cite{PaulyParikh:GL} goes into the two-variable 
fragment of modal $\mu$-calculus, and it is therefore not useful for the proof 
transformations in this paper.
Indeed, we see the translation $\tr{(-)}$ as one of our main technical 
contributions.

4) Finally, we obtain completeness of $\CloM$ from the completeness of 
$\Clo$ \cite{AfshariLeigh:LICS17} by transforming $\Clo$-derivations into 
$\CloM$-derivations via
a validity-preserving translation $(-)^t$ which is the fixpoint extension of a well-known translation of monotone modal logic into normal modal logic \cite{KrachtWolter99}.
\medskip

To summarise, completeness of Parikh's system $\HilbP$ is obtained
by the following argument.
Assume that $\phi$ is a game logic formula that is valid over monotone neighbourhood models.
As the above mentioned translations are validity-preserving,
the normal modal $\mu$-calculus formula $(\tr{(\nf{\phi})})^t$
is valid over Kripke models. 
By completeness of $\Clo$, there is a $\Clo$-derivation of 
$(\tr{(\nf{\phi})})^t$. 
By the above sequence of transformations, we obtain a $\HilbP$-derivation of 
$\phi$.

\subsection{Outline}
%
The paper is organised as follows. In Section~\ref{sec:two-sys} we recall the 
basic definitions of game logic, we introduce Parikh's Hilbert-style axiomatisation $\HilbP$, we present the cut-free Gentzen style system \GLms and
show that \GLms-derivations can be transformed into \HilbP-derivations (Thm~\ref{thm:GLms-to-HilbP}).
In Section~\ref{sec:annot-sys}, we introduce the annotated proof system
$\CloG$ for game logic and show
how $\CloG$-derivations can be translated into 
\GLms-derivations (Thm.~\ref{thm:CloG-to-KozenG}).
In Section~\ref{sec:mon-mu}
we define the annotated system $\CloM$ for the monotone $\mu$-calculus 
and prove its soundness and completeness by connecting it to the $\Clo$-system from~\cite{AfshariLeigh:LICS17} using the standard simulation
of monotone modal logic with a binormal modal logic. 
In Section~\ref{sec:translation}, we show how $\CloM$-derivations can be transformed into corresponding
$\CloG$-derivations using the translation $\tr{(-)}$
of game logic into the monotone $\mu$-calculus (Thm.~\ref{mu to gl}).
In Section~\ref{sec:s-c}, we apply the transformation results to prove
soundness and completeness of $\CloG$, $\GLms$ and $\HilbP$.
Finally, in Section~\ref{sec:conclusion} we conclude and discuss related and future work.
Due to space limitations we only provide proofs of key results. All proofs  
that have been omitted from the main text can be found in the appendix.

\section{Two derivation systems}
\label{sec:two-sys}

\subsection{Game logic: basics}
\label{subsec:game logic basics}

Throughout, we assume fixed countable sets $\AtProp$ and $\AtGame$ of atomic
propositions and atomic games, respectively.
Over these sets we shall define three distinct languages of game logic.
Parikh's original language $\langPar$ only allows the angelic version of game 
constructors, while dual and negation may occur freely, 
The \emph{normal form language} $\langNorm$ allows both angelic and demonic game
constructors, while negation and duals may only occur in front of atoms.
The \emph{full} language $\langFull$ allows all connectives and game constructors from the
other two languages, and freely placed duals and negations.

\begin{definition}
    \label{def:langPar}
    \label{def:langNorm}
The languages $\langPar$ and $\langNorm$ consist of the formulas and games
generated by the following grammars:
\[\begin{array}{ccl}
  \formPar \ni \phi
  & ::= &
  p 
  \mid \lnot \phi
  \mid \phi \lor \phi
  \mid \diam{\gamma}\phi, \;\gamma\in\gamePar
  \\
  \gamePar \ni \gamma
  & ::= &
  g 
  \mid \gamma \comp \gamma
  \mid \gamma \gcha \gamma
  \mid \gamma^*
  \mid \gamma^d
  \mid \phi?,\;\phi \in \formPar
\\
\formNorm \ni \phi
  & ::= &
  p 
  \mid \neg p
  \mid \phi \lor \phi
  \mid \phi \land \phi
  \mid \diam{\gamma}\phi, \;\gamma\in\gameNorm
\\
\gameNorm \ni \gamma
  & ::= &
  g  
  \mid g^d 
  \mid \gamma \comp \gamma
  \mid \gamma \gcha \gamma
  \mid \gamma \gchd \gamma  
  \mid \gamma^*
  \mid \gamma^\cross
   \\ && \phantom{g} 
  \mid \phi?
  \mid \phi!,\;\phi \in \formNorm
  \\
 \formFull \ni \phi
  & ::= &
  p 
  \mid \lnot \phi
  \mid \phi \lor \phi
  \mid \diam{\gamma}\phi, \;\gamma\in\gameFull
  \\
  \gameFull \ni \gamma
  & ::= &
  g 
  \mid \gamma \comp \gamma
  \mid \gamma \gcha \gamma
  \mid \gamma \gchd \gamma  
  \mid \gamma^*
  \mid \gamma^\cross
  \mid \gamma^d
   \\ && \phantom{g} 
  \mid \phi?,\;\phi \in \formFull
\end{array}
\]
where $p \in \AtProp$ and $g \in \AtGame$.
In $\langPar$ and $\langFull$ we admit the connectives $\to, \land, \lra$ as 
the usual abbreviations.
\end{definition}

The game operations should be read as follows. 
The \emph{composition} $\gamma\comp\delta$ means first play $\gamma$, then 
play $\delta$.
The \emph{angelic choice} $\gamma \gcha \delta$ is the game where Angel decides whether 
to play $\gamma$ or $\delta$. 
The \emph{angelic iteration} $\gamma^*$ is the game in which $\gamma$ is played a finite,
possibly zero, number of times, with Angel at the start and after each round
deciding whether to stop or play one more round of $\gamma$.
The \emph{angelic test} $\psi?$ is the game in which $\phi$ is evaluated, and
Angel immediately ``loses'' if  $\psi$ is false, and otherwise play continues.
The \emph{dual game} $\gamma^d$ is the game in which the roles of the
two players are interchanged, i.e., the strategies of Angel in $\gamma^d$ are 
exactly the strategies of Demon in $\gamma$, and vice versa.
The definitions of the demonic operations are such that (cf.~\cite{PaulyParikh:GL}):
\begin{equation}
    \gamma \gchd \delta = ( \gamma^d \gcha \delta^d )^d, \quad
    \gamma^\cross = ( (\gamma^d)^* )^d, \quad
    \psi! = ((\lnot\psi)?)^d\;
\rem{  \begin{array}{rcl}
    \gamma \gchd \delta &=& ( \gamma^d \gcha \delta^d )^d\\
    \gamma^\cross &=& ( \gamma^d )^d\\
    \psi! &=& ((\lnot\psi)?)^d
  \end{array}
  }
\end{equation}
The interpretation of the demonic operations is obtained by replacing ``Angel'' with ``Demon''
in the above. However, since a modal formula $\diam{\gamma}\phi$ expresses
the strategic ability of Angel in $\gamma$, 
$\diam{\gamma\gchd\delta}\phi$ means that Angel has a strategy to achieve $\phi$
in both $\gamma$ and $\delta$, and
$\diam{\gamma^\cross}\phi$ means that Angel has a strategy for maintaining $\phi$
indefinitely when playing $\gamma$ repeatedly, and not knowing when the iteration terminates.
Finally, $\psi!$ is the game in which Angel immediately ``wins'' if
$\psi$ is true. Hence, $\diam{\psi!}
\varphi$ is true if at least one of $\psi$ and $\varphi$ is true.

We will often refer to formulas and games jointly as \emph{terms}. 
We denote the \emph{subterm relation} by $\subterm$, using
$\ssubterm$ for the strict version.
For example, $g^\cross \ssubterm \diam{g^\cross \comp h}p$ and $h
\ssubterm (\diam{h}p?)\comp g$.

Formulas of the form $\diam{\gamma^*}\phi$ or $\diam{\gamma^\cross}\phi$ will play the role of fixpoint variables on the game logic side. In particular, we need to define an order $\lessfp$ on them, but it is not immediately clear how to do that. 
For example, a naive approach based on the subformula-relation will not work, 
since we need that, e.g., $\diam{(g^\cross \gcha h)^\cross}p \lessfp 
\diam{g^\cross}p$.
Our solution is to use the converse subterm relation on the game terms
that label the modalities.

\begin{definition}\label{def:F-order}
We define the set of least, greatest, respectively all, \emph{fixpoint formulas} in
$\langNorm$ as follows:
\[\begin{array}{lcl}
   F^*      & \isdef & \{\diam{\gamma^*} \varphi 
      \mid \gamma \in \gameNorm, \varphi \in \formNorm \},
\\ F^\cross & \isdef & \{\diam{\gamma^\cross} \varphi 
      \mid \gamma \in \gameNorm, \varphi\in \formNorm \},\\
   F        & \isdef & F^* \cup F^\cross.
\end{array}\]
We define an order $\lessfp$ on $F$ by setting
$\diam{\gamma^\circ} \varphi \lessfp \diam{\delta^\dagger} \psi$
for $\circ,\dagger \in \{*,\cross\}$
if $\delta^\dagger \ssubterm \gamma^\circ$. We write $\diam{\gamma^\circ} \varphi \leqfp \diam{\delta^\dagger} \psi$ if $\delta^\dagger \lessfp \gamma^\circ$ or $\delta^\dagger = \gamma^\circ$.
\end{definition}
It should be clear that $\lessfp$ is transitive and irreflexive.

\rem{
\begin{example}
  Applying the definition of $\lessfp$,
  we find that, for example: 
  $\diam{g^\cross \gcha h^\cross}p \lessfp \diam{g^\cross}p$
  and
  $\diam{(a^* \comp (\diam{b^\cross}p)?)^\cross}p \lessfp \diam{a^*}p$,
  but it is not the case that
  $\diam{(a^* \comp (\diam{b^\cross}p)?)^\cross}p \leqfp \diam{((\diam{b^\cross}p)?)^\cross}p$.
\end{example}
}

We need the following notion of (Fischer-Ladner) \emph{closure}.

\begin{definition}\label{def:clos}
The \emph{closure} $\Cl{\xi}$ of a formula $\xi \in \langNorm$ is the smallest
subset of $\langNorm$ that contains $\xi$ and is closed under subformulas as well as 
the following rules:
If $\diam{\gamma^{*}}\phi\in\Cl{\xi}$ then $\phi\lor\diam{\gamma}\diam{\gamma^{*}}\phi \in \Cl{\xi}$.
If $\diam{\gamma^{\cross}}\phi \in\Cl{\xi}$ then $\phi\land\diam{\gamma}\diam{\gamma^{\cross}}\phi\in\Cl{\xi}$.
If $\diam{\psi?}\phi \in\Cl{\xi}$ then $\psi \in\Cl{\xi}$.
If $\diam{\psi!}\phi \in\Cl{\xi}$ then $\psi \in\Cl{\xi}$.
\rem{ 
$\neg p/p$, $\phi_{0}\land\phi_{1}/\phi_{i}$, $\phi_{0}\lor\phi_{1}/\phi_{i}$,
$\diam{\gamma}\phi/\phi$,

$\diam{\gamma \comp \delta}\phi/\diam{\gamma}\diam{\delta}\phi$,
$\diam{\gamma_{0} \gcha \gamma_{1}}\phi/\diam{\gamma_{i}}\phi$,
$\diam{\gamma_{0} \gchd \gamma_{1}}\phi/\diam{\gamma_{i}}\phi$,
}
\rem{ 
$\diam{\gamma^{*}}\phi/\phi\lor\diam{\gamma}\diam{\gamma^{*}}\phi$,
$\diam{\gamma^{\cross}}\phi/\phi\land\diam{\gamma}\diam{\gamma^{\cross}}\phi$,

$\diam{\psi?}\phi/\psi$,
$\diam{\psi!}\phi/\psi$.
}
The sets $F(\xi), F^{*}(\xi), F^{\cross}(\xi)$ of all/least/greatest
\emph{fixpoint formulas of a formula} $\xi \in \langNorm$ are given as 
$F(\xi) \isdef F \cap \Cl{\xi}$, etc.
\end{definition}

The simplest way to define the semantics of these languages is as follows
\cite{PaulyParikh:GL}.
\rem{An \emph{effectivity function for a game $\gamma$} on a set $S$ is a
monotone map $E_\gamma: \wp(S) \to \wp(S)$, and $s \in E_\gamma(Y)$
means that at position $s$, Angel is effective for $Y$ in $\gamma$,
i.e., Angel has a strategy in $\gamma$ that ensures that the outcome of
$\gamma$ is in $Y$.}
We denote by $\mm{S}$ the set of all monotone maps $f\colon\wp(S) \to \wp(S)$. An \emph{effectivity function for a game $\gamma$ on a set $S$} is then a
$E_\gamma\in \mm{S}$, and $s \in E_\gamma(Y)$
means that at position $s$, Angel is effective for $Y$ in $\gamma$,
i.e., Angel has a strategy in $\gamma$ that ensures that the outcome of
$\gamma$ is in $Y$.

\begin{definition}
A \emph{game model} is a triple $\bbS = (S,E,V)$ such that
$V: \AtProp \to \wp(S)$ is a \emph{valuation} and $E\colon \AtGame \to \mm{S}$
assigns an effectivity function on $S$ to every atomic $g \in \AtGame$.
By a mutual induction on formulas and games, we define the
\emph{meaning} $\mng{\phi}^{\bbS}$ of a formula $\phi$ in a model $\bbS$,
and the effectivity function $E_\gamma$ in $\bbS$   
for complex games $\gamma$ as follows:
\[
\begin{array}{lll}
   \mng{p}^{\bbS}      & \isdef & V(p)
\\ \mng{\neg p}^{\bbS} & \isdef & S \setminus p
\\ \mng{\phi\lor\psi}^{\bbS}  & \isdef & \mng{\phi}^{\bbS} \cup \mng{\psi}^{\bbS}
\\ \mng{\phi\land\psi}^{\bbS} & \isdef & \mng{\phi}^{\bbS} \cap \mng{\psi}^{\bbS}
\\ \mng{\diam{\gamma}\phi}^{\bbS} & \isdef & \wh{E}_{\gamma}(\mng{\phi}^{\bbS})
\\[.3em] \wh{E}_g(X) & \isdef & E(g)(X)
\\ \wh{E}_{(\gamma^{d})}(X) & \isdef & S \setminus \wh{E}_{\gamma}(S \setminus X)
\\ \wh{E}_{\gamma\comp\delta}(X)  & \isdef & \wh{E}_{\gamma}(\wh{E}_{\delta}(X))
\\ \wh{E}_{\gamma\gcha\delta}(X)  & \isdef 
    & \wh{E}_{\gamma}(X) \cup \wh{E}_{\delta}(X)
\\ \wh{E}_{\gamma\gchd\delta}(X)  & \isdef  
    & \wh{E}_{\gamma}(X) \cap \wh{E}_{\delta}(X)
\\ \wh{E}_{(\gamma^{*})}(X)       & \isdef & \lfp\, Y.\, X \cup \wh{E}_{\gamma}(Y) 
\\ \wh{E}_{(\gamma^{\cross})}(X)  & \isdef & \gfp\, Y.\, X \cap \wh{E}_{\gamma}(Y) 
\\ \wh{E}_{(\phi?)}(X)  & \isdef & \mng{\phi}^{\bbS} \cap X
\\ \wh{E}_{(\phi!)}(X)  & \isdef & \mng{\phi}^{\bbS} \cup X
\end{array}
\]
Notions like satisfiability, equivalence, etc., are defined in the standard way.
In particular, a game formula $\phi$ is \emph{valid}, notation: $\models
\phi$, if $\mng{\phi}^{\bbS} = S$, for every game model $\bbS = (S,E,V)$.
\end{definition}

\begin{proposition}
\label{p:trl}
There are recursively defined, truth-preserving translations
\[\begin{array}{lll}
   \trParNorm{-} \colon & \langFull \to \langNorm
\\ \trNormPar{-} \colon & \langFull \to \langPar
\end{array}\]
\end{proposition}

As a corollary of this, negation is definable in $\langNorm$.
We shall need the following explicit definition in the sequel.

\begin{definition}
By a mutual induction we define the \emph{complementation} $\ol{\phi} \isdef
\trParNorm{\neg\phi}$ of an $\langNorm$-formula $\phi$, and the \emph{dual game}
$\dg{\gamma}$ of an $\langNorm$-game $\gamma$:
\[
\begin{array}[t]{lll}
   \ol{p}      & \isdef & \neg p
\\ \ol{\neg p} & \isdef & p
\\ \ol{\phi\lor\psi}  & \isdef & \ol{\phi} \land \ol{\psi}
\\ \ol{\phi\land\psi}  & \isdef & \ol{\phi} \lor \ol{\psi}
\\ \ol{\diam{\gamma}\phi} & \isdef & \diam{\dg{\gamma}}\ol{\phi}
\end{array}
\hspace*{10mm}
\begin{array}[t]{lll}
   \dg{g}      & \isdef & g^{d}
\\ \dg{(g^{d})} & \isdef & g
\\ \dg{\gamma\comp\delta}  & \isdef & \dg{\gamma} \comp \dg{\delta}
\\ \dg{\gamma\gcha\delta}  & \isdef & \dg{\gamma} \gchd \dg{\delta}
\\ \dg{\gamma\gchd\delta}  & \isdef & \dg{\gamma} \gcha \dg{\delta}
\\ \dg{(\gamma^{*})}  & \isdef & (\dg{\gamma})^{\cross}
\\ \dg{(\gamma^{\cross})}  & \isdef & (\dg{\gamma})^{*}
\\ \dg{(\phi?)}  & \isdef & \ol{\phi}!
\\ \dg{(\phi!)}  & \isdef & \ol{\phi}?
\end{array}
\]
\end{definition}
The following proposition is proved by a straightforward induction. We leave the details to the reader.
\begin{proposition}
In any game model $\bbS = (S,E,V)$ we have
\[
      \mng{\ol{\phi}}^{\bbS} = S \setminus \mng{\phi}^{\bbS}
\text{ and } \wh{E}_{\dg{\gamma}} = \wh{E}_{\gamma^{d}}
\]
for any formula $\phi \in \formNorm$ and game $\gamma \in \gameNorm$.
\end{proposition}

\subsection{Parikh's Hilbert-style system}

The first axiom system for game logic was proposed and conjectured to be 
complete by Parikh \cite{Parikh85}. This is a Hilbert-style system
for the language $\langPar$ that axiomatises the angelic iteration with what Parikh calls \emph{Bar Induction}.
We will refer to this system as $\HilbP$, and it is shown in
Figure~\ref{fig:pfParikhRu}  below. 
For $\phi \in \langPar$, we write $\HilbP \vdash \phi$
if there is a $\HilbP$-derivation of $\phi$.

\begin{figure}[htb]
\fbox{
\begin{minipage}[t]{.47\textwidth}

\begin{minipage}[t]{.49\textwidth}
\quad\textbf{\HilbP Axioms:}
  \begin{enumerate}
  \item All propositional tautologies.
  \item $\diam{\gamma \comp \delta}\phi \lra \diam{\gamma}\diam{\delta}\phi$
  \item $\diam{\gamma \gcha \delta}\phi \lra \diam{\gamma}\phi \lor \diam{\delta}\phi$
  \item $\diam{\gamma^*}\phi \lra \phi \lor \diam{\gamma}\diam{\gamma^*}\phi$
  \item $\diam{{\psi}?}\phi \lra \psi \land \phi$
  \item $\diam{\gamma^d}\phi \lra \lnot\diam{\gamma}\lnot\phi$
  \end{enumerate}
 \end{minipage}
%
\begin{minipage}[t]{.5\textwidth}
\quad\quad\textbf{\HilbP Rules:}
    \begin{prooftree}
      \AxiomC{$\phi$}
      \AxiomC{$\phi \to \psi$}
      \RightLabel{\RuMP}
      \BinaryInfC{$\psi$}
    \end{prooftree}

 \begin{prooftree}
      \AxiomC{$\phi \to \psi$}
      \RightLabel{\RuMon}
      \UnaryInfC{$\diam{\gamma} \phi \to \diam{\gamma}\psi$}
    \end{prooftree}

    \begin{prooftree}
      \AxiomC{$\diam{\gamma}\phi \to\phi$}
      \RightLabel{\RuBarInd}
      \UnaryInfC{$\diam{\gamma^*}\phi \to \phi$}
    \end{prooftree}
  \end{minipage}
 \medskip
  \end{minipage}
}
\caption{Axioms and rules of $\HilbP$.}
\label{fig:pfParikhRu}
\end{figure}

The system $\HilbP$ is easily seen to be sound.
A main contribution of our paper is that we confirm Parikh's
completeness conjecture. 
We prove Theorem~\ref{t:par-completeness} in section~\ref{sec:s-c} below.

\begin{theorem}[Soundness and Completeness of \HilbP]
\label{t:par-completeness}
For every formula $\phi \in\formPar$, we have:
$
\HilbP \vdash \phi \text{ iff } \models \phi.
$
\end{theorem}

\rem{
For the sake of clarity: here and throughout the paper, by ``completeness'' we 
mean ``weak completeness'' in the sense that every valid formula is provable. 
As with PDL and the modal $\mu$-calculus, strong completeness of a finitary 
proof system is not possible, since game logic is not compact.
}

\subsection{The cut-free sequent system $\GLms$ for game logic}

We now introduce a cut-free (Tait-style) sequent system \GLms for
game logic formulas in normal form.
A \emph{sequent} is thus defined as a finite set of $\formNorm$-formulas (to be read disjunctively).
For a finite set $\Phi \subseteq \formNorm$, we define $\nnf{\Phi}\in
\formNorm$ as the normal form $\nnf{\bigvee\Phi}$ of $\neg(\bigvee \Phi)$.

The system $\GLms$ consists of several parts.
Its core is the sequent calculus version of monotone modal logic as shown in
Figure~\ref{fig:pfGL1}.
\begin{figure}[thb]
\fbox{
\begin{minipage}[t]{0.47\textwidth}
\begin{minipage}[t]{0.2\textwidth}
\begin{prooftree}
 \AxiomC{$\phantom{\Phi,}$}
 \RightLabel{\Ax}
 \UnaryInfC{$\Phi, \nnf{\Phi}$}
\end{prooftree}
\end{minipage}
\begin{minipage}[t]{0.2\textwidth}
\begin{prooftree}
  \AxiomC{$\Phi$}
  \RightLabel{\RuWeak}
  \UnaryInfC{$\Phi,\phi $}
 \end{prooftree}
\end{minipage}
\begin{minipage}[t]{0.35\textwidth}
 \begin{prooftree}
  \AxiomC{$\phantom{\Phi,} \phi , \psi $}
  \RightLabel{\RuMonMod}
  \UnaryInfC{$\diam{g} \phi , \diam{g^d} \psi $}
 \end{prooftree}
 \end{minipage}
 \medskip 
 \begin{minipage}[t]{0.45\textwidth}
 \begin{prooftree}
 \AxiomC{$\Phi, \phi , \psi $}
 \RightLabel{\RuOr}
 \UnaryInfC{$\Phi, \phi  \lor \psi $}
\end{prooftree}
\end{minipage}
\begin{minipage}[t]{0.45\textwidth}
    \begin{prooftree}
 \AxiomC{$\Phi, \phi $}
 \AxiomC{$\Phi, \psi $}
 \RightLabel{\RuAnd}
 \BinaryInfC{$\Phi, \phi  \land \psi $}
\end{prooftree}
\end{minipage}
\end{minipage} 
 }
 \caption{The basic rules of the sequent calculus \monML for Game Logic.}
 \label{fig:pfGL1}
 \end{figure}
In order to reason about game operators, in Fig.~\ref{fig:game-op-rules} we list
some rules, each of which directly mirrors the semantic meaning of one game
constructor.
\begin{figure}[thb]
\fbox{
\begin{minipage}{.47\textwidth}
\begin{minipage}[t]{.36\textwidth}
\begin{prooftree}
 \AxiomC{$\Phi, \phi  \lor \diam{\gamma} \diam{\gamma^*} \phi $}
 \RightLabel{\RuStar}
 \UnaryInfC{$\Phi, \diam{\gamma^*} \phi $}
\end{prooftree}
\begin{prooftree}
 \AxiomC{$\Phi, \phi  \land \diam{\gamma} \diam{\gamma^\cross} \phi $}
 \RightLabel{\RuCross}
 \UnaryInfC{$\Phi, \diam{\gamma^\cross} \phi $}
\end{prooftree}
\end{minipage}
\begin{minipage}[t]{.36\textwidth}
\begin{prooftree}
 \AxiomC{$\Phi, \diam{\gamma} \phi  \lor \diam{\delta} \phi $}
 \RightLabel{\RuChoiceA}
 \UnaryInfC{$\Phi, \diam{\gamma \gcha \delta} \phi $}
\end{prooftree}
\begin{prooftree}
 \AxiomC{$\Phi, \diam{\gamma} \phi  \land \diam{\delta} \phi $}
 \RightLabel{\RuChoiceD}
 \UnaryInfC{$\Phi, \diam{\gamma \gchd \delta} \phi $}
\end{prooftree}
\end{minipage}
\begin{minipage}[t]{.25\textwidth}
\begin{prooftree}
 \AxiomC{$\Phi, \psi  \land \phi $}
 \RightLabel{\RuQuestion}
 \UnaryInfC{$\Phi, \diam{ \psi ?} \phi $}
\end{prooftree}
\begin{prooftree}
 \AxiomC{$\Phi, \psi  \lor \phi $}
 \RightLabel{\RuBang}
 \UnaryInfC{$\Phi, \diam{\psi !} \phi $}
\end{prooftree}
\end{minipage}
\end{minipage}
}
\caption{The sequent calculus rules \GameOpRules for game operations.}
\label{fig:game-op-rules}
\end{figure}

In the third part of the \GLms proof system we have the three ``deep'' derivation
rules given in Figure~\ref{fig:deeprules}. 
These rules are needed for technical reasons, as will become clear in some of
the proofs further on. 
\begin{figure}[thb]
\fbox{
\begin{minipage}{0.48\textwidth}
\hspace{-3mm}
\begin{minipage}{0.33\textwidth}
\begin{prooftree}
 \AxiomC{$\Phi,\psi(\gamma)$}
 \RightLabel{\small \RuMonDgame}
 \UnaryInfC{$\Phi, \psi(\chi ! \comp \gamma)$}
\end{prooftree}
\end{minipage}
\begin{minipage}{0.33\textwidth}
\begin{prooftree}
 \AxiomC{$\Phi, \psi(\phi)$}
 \RightLabel{\small \RuMonDfmla}
 \UnaryInfC{$\Phi, \psi(\diam{\chi !} \phi)$}
\end{prooftree}
\end{minipage}
\hspace{1mm} 
\begin{minipage}{0.32\textwidth}
\begin{prooftree}
 \AxiomC{$\Phi, \psi(\diam{\gamma} \diam{\delta} \phi)$}
 \RightLabel{\small \RuCompD}
 \UnaryInfC{$\Phi, \psi(\diam{\gamma\comp\delta} \phi)$}
\end{prooftree}
\end{minipage}
\end{minipage}
}
\caption{Deep rules for Game Logic: \DeepGRules. The notation $\psi(\phi)$ should be read as follows: $\psi$ is a context, i.e., a formula with a unique occurrence of a proposition letter $p$, and  $\psi(\phi)$ is the formula obtained by substituting $p$ for $\phi$ in $\psi$.}
\label{fig:deeprules}
\end{figure}

The final ingredient of \GLms is the strengthened induction rule $\RuIndS$ in
Figure~\ref{fig:KozenG}.
This rule, just like the homonymous rule in~\cite{AfshariLeigh:LICS17} on which
it is based, detects unfoldings of greatest fixpoints in the same context.
This may become clearer when we show in Theorem~\ref{thm:CloG-to-KozenG} how \RuIndS is used to translate the closure rule of the system $\CloG$.
In this sense, $\RuIndS$ plays a role similar to the \emph{context rule}~\cite[Proposition~5.7(vi)]{koze:resu83} in Kozen's completeness proof for the aconjunctive fragment of the modal $\mu$-calculus. Only, Kozen's proof is based on satisfiability and the context rule therefore deals with least fixpoint unfoldings.
Our approach is based on validity, and $\RuIndS$ therefore detects greatest fixpoint unfoldings. 

\begin{figure}[thb]
\fbox{
\begin{minipage}{0.47\textwidth}
\begin{prooftree}
 \AxiomC{$\Phi,
   \phi \land \diam{\gamma} 
   \diam{( \nnf{\Phi}! \comp \gamma)^\cross}
   \diam{ \nnf{\Phi}!}\phi$}
 \RightLabel{\RuIndS}
 \UnaryInfC{$\Phi, \diam{\gamma^\cross} \phi$}
\end{prooftree}
\end{minipage}
}
\caption{Strengthened induction rule for Game Logic.}
\label{fig:KozenG}
\end{figure}

To obtain a more concrete understanding of the $\RuIndS$ rule, think of the formula
$\diam{\gamma^{\cross}}\phi$ as a greatest fixpoint formula 
$\nu x. \phi \land \diam{\gamma}x$.
The ``standard'' fixpoint rule for $\gamma^{\cross}$ would read as
follows: ``from $\psi \to \phi \land \diam{\gamma}\psi$ infer $\psi \to 
\diam{\gamma^{\cross}}\phi$'', or, formulated as a Tait-style sequent rule:
\begin{prooftree}
 \AxiomC{$\Phi, \phi \land \diam{\gamma}\nnf{\Phi}$}
 \RightLabel{\RuInd}
 \UnaryInfC{$\Phi, \diam{\gamma^\cross} \phi$}
\end{prooftree}
Now, observing that $\diam{ \nnf{\Phi}!}\phi \equiv \nnf{\Phi} \lor \phi$,
one may see that $\RuIndS$ is indeed a variation of $\RuInd$.

Some further understanding of the rule $\RuIndS$ may be gained by establishing
its \emph{soundness}.
For this purpose we may reason by contraposition, showing that the refutability 
of the conclusion of $\RuIndS$ implies the refutability of its premise.
It is not hard to see that this boils down to proving the following statement,
which is formulated using the dual formulas and games.

\begin{proposition}\label{prop:strind}
If $\chi \land \diam{\gamma^{*}}\phi$ is satisfiable, then so is either
$\chi \land \phi$ or 
$\chi \land 
\diam{\gamma}\diam{(\nnf{\chi}?\comp\gamma)^{*}}\diam{\nnf{\chi}?}\phi$.
\end{proposition}
In words, this Proposition states the following.
Suppose that there is a situation where $\chi$ holds and where Angel has a 
strategy in the game $\gamma^{*}$ ensuring the outcome $\phi$.
Suppose furthermore that $\chi$ and $\phi$ cannot be true simultaneously.
Then there is a situation where $\chi$ holds, and where Angel has a strategy 
in $\gamma^{*}$ which not only ensures that $\phi$ holds afterwards, but 
also guarantees that while playing $\gamma^{*}$, after each round of playing 
$\gamma$, the formula $\nnf{\chi}$ holds.

The completeness of $\GLms$ will follow from the completeness
of the system $\CloG$, which we introduce in the next section.
The proof of Theorem~\ref{thm:compGLms}
will be outlined in Section~\ref{sec:s-c}.

\begin{theorem}[Soundness and Completeness of \GLms]
\label{thm:compGLms}
For all $\xi \in\formNorm$, we have:
$
\GLms \vdash \xi \text{ iff } \models \xi.
$
\end{theorem}

The following theorem states the transformation results
between $\GLms$ and $\HilbP$ that are needed for transferring
soundness from $\HilbP$ to $\GLms$,
and completeness from $\GLms$ to $\HilbP$.

\begin{theorem}
  \label{thm:GLms-to-HilbP}
  We have:
  \begin{enumerate}
    \item 
      For all $\phi \in \langPar$,
      if $\GLms \vdash \dnnf{\phi}$ then $\HilbP \vdash \phi$.  
    \item For all $\xi \in \langNorm$,
      if $\GLms \vdash \xi$ then $\HilbP \vdash \pari{\xi}$.  
  \end{enumerate}
\end{theorem}

\section{An annotated proof system}
\label{sec:annot-sys}



The completeness of \GLms will follow from the completeness of the annotated tableau system $\CloG$ which we introduce now.

\subsection{The $\CloG$ system for Game Logic}
\label{subsec:clog}

In $\CloG$, formulas are annotated with names that are used to
detect repeated unfoldings of greatest fixpoint formulas in the same context.
With each greatest fixpoint formula $\phi \in F^\cross$
we associate a countable set $\names[\phi]$ of \emph{names for $\phi$}.
We assume that $\names[\phi] \cap \names[\psi] =\emptyset$ if $\phi\neq\psi$.
The set of all names is $\names = \bigcup_{\phi\in {F^\cross}}\names[\phi]$.
Names will typically be denoted by $\nam{x}, \nam{y}, \ldots$ or with subscripts $\nam{x}_0, \nam{x}_1, \ldots$.
Names inherit the order $\leqfp$ on the set $F$ of fixpoint formulas:
For all $\nam{x} \in \names[\phi]$, $\nam{y} \in \names[\psi]$,
we define $\nam{x} \leqfp \nam{y}$ iff $\phi \leqfp \psi$. 
For a sequence of names $\nam{a} = \nam{x}_0,\nam{x}_1,\ldots, \nam{x}_{n-1} \in
N^*$ and a fixpoint formula $\phi \in F$, we will write  $\nam{a} \leqfp
\phi$ if for all $\nam{x}_i$ occurring in $\nam{a}$, $\nam{x}_i \in
\names[\psi]$ such that $\psi \leqfp \phi$. The empty sequence is
denoted by $\nam{\eps}$.
An \emph{annotation} is a sequence $\nam{a} = \nam{x}_0,\nam{x}_1,\ldots, \nam{x}_{n-1} \in N^*$ that is
non-repeating and monotone w.r.t. $\leqfp$, i.e., for all $i < n-1$, $\nam{x}_i \leqfp \nam{x}_{i+1}$.
An \emph{annotated game logic formula $\phi^\nam{a}$} consists of a
formula $\phi \in \langNorm$ and an annotation $\nam{a} \in N^*$.
\emph{Annotated $\CloG$-sequents} are finite sets of annotated game
logic formulas, and will be denoted by $\Phi, \Psi,$ etc.

The system $\CloG[\xi]$ derives $\CloG$-sequents using the axiom and rules in
Figure~\ref{fig:CloG}.
The \emph{closure rule} $\RuClo_{\nam{x}}$ discharges \emph{all} occurrences of
the sequent $\Phi, \diam{\gamma^\cross} \varphi^{\nam{{ax}}}$ appearing as an 
assumption above the proof node where the rule is applied.
The side conditions ensure that no fixpoint formula of
higher priority than $\diam{\gamma^\cross}\phi$ is unfolded between
the application of 
$\RuClo_{\nam{x}}$ and its discharged assumption.

A $\CloG$-proof is a finite tree of $\CloG$-inferences in which each leaf is
labelled by an axiom or a discharged assumption.
Intuitively, a $\CloG$-proof can be understood as a finitary representation of a non-wellfounded/circular proof. The discharged assumptions are the nodes where the circularity is detected.
For a formula $\xi\in \langNorm$, we write $\CloG \vdash \xi$ to mean that
there is a $\CloG$-proof of $\xi^{\eps}$.
Note that $\CloG$ is \emph{analytic} in the sense that 
any $\CloG$-proof of $\xi^\eps$ will contain only formulas from $\Cl{\xi}$, and names for fixpoint formulas in $F^\cross(\xi)$.

\rem{For the sake of readability, we will occasionally be sloppy about the distinction between plain and annotated formulas, for instance, by omitting explicit annotations, \margincomment{HH: Do we need this?}
or by saying that $\varphi^\nam{a}$ belongs to the closure of
$\psi^\nam{b}$ to mean that $\varphi$ belongs to the closure of $\psi$.
}

\begin{figure}
\fbox{
\begin{minipage}[t]{.47\textwidth}
\begin{minipage}[t]{.49\textwidth}
\begin{prooftree}
 \AxiomC{\phantom{X}}
 \RightLabel{\AxOne}
 \UnaryInfC{$p^{\nam{\eps}}, (\neg p)^{\nam{\eps}}$}
\end{prooftree}
\begin{prooftree}
 \AxiomC{$\Phi, \varphi^{\nam{a}}, \psi^{\nam{a}}$}
 \RightLabel{\RuOr}
 \UnaryInfC{$\Phi, (\varphi \lor \psi)^{\nam{a}}$}
\end{prooftree}
\end{minipage}
\begin{minipage}[t]{.49\textwidth}
\begin{prooftree}
 \AxiomC{$\varphi^{\nam{a}}, \psi^{\nam{b}}$}
 \RightLabel{\RuMonMod}
 \UnaryInfC{$(\diam{g} \varphi)^{\nam{a}}, (\diam{g^d} \psi)^{\nam{b}}$}
\end{prooftree}
\begin{prooftree}
 \AxiomC{$\Phi, \varphi^{\nam{a}}$}
 \AxiomC{$\Phi, \psi^{\nam{a}}$}
 \RightLabel{\RuAnd}
 \BinaryInfC{$\Phi, (\varphi \land \psi)^{\nam{a}}$}
\end{prooftree}
\end{minipage}

\medskip
\begin{minipage}[t]{.49\textwidth}
\begin{prooftree}
 \AxiomC{$\Phi, (\diam{\gamma} \phi \lor \diam{\delta} \phi)^{\nam{a}}$}
 \RightLabel{\RuChoiceA}
 \UnaryInfC{$\Phi, (\diam{\gamma \gcha \delta} \phi)^{\nam{a}}$}
\end{prooftree}
\end{minipage}
\begin{minipage}[t]{.48\textwidth}
\begin{prooftree}
 \AxiomC{$\Phi, (\diam{\gamma} \phi \land \diam{\delta} \phi)^{\nam{a}}$}
 \RightLabel{\RuChoiceD}
 \UnaryInfC{$\Phi, (\diam{\gamma \gchd \delta} \phi)^{\nam{a}}$}
\end{prooftree} 
\end{minipage}

\medskip
\begin{minipage}[b]{.32\textwidth}
\begin{prooftree}
 \AxiomC{$\Phi$}
 \RightLabel{\RuWeak}
 \UnaryInfC{$\Phi,\varphi^{\nam{a}}$}
\end{prooftree}
\end{minipage}
\begin{minipage}[b]{.32\textwidth}
\begin{prooftree}
 \AxiomC{$\Phi, \varphi^{\nam{{ab}}},$}
 \RightLabel{\RuExp}
 \UnaryInfC{$\Phi, \varphi^{\nam{{axb}}}$}
\end{prooftree}
\end{minipage}
\begin{minipage}[b]{.33\textwidth}
\begin{prooftree}
 \AxiomC{$\Phi, (\diam{\gamma} \diam{\delta} \phi)^{\nam{a}}$}
 \RightLabel{\RuComp}
 \UnaryInfC{$\Phi, (\diam{\gamma \comp \delta} \phi)^{\nam{a}}$}
\end{prooftree}
\end{minipage}

\medskip
\begin{minipage}[t]{.66\textwidth}
\begin{prooftree}
 \AxiomC{$\Phi, (\phi  \lor \diam{\gamma} \diam{\gamma^*} \phi)^{\nam{a}}$}
 \LeftLabel{$(\nam{a} \leqfp[\xi] \diam{\gamma^*} \phi)$}
 \RightLabel{\RuStar}
 \UnaryInfC{$\Phi, (\diam{\gamma^*}\phi)^{\nam{a}}$}
\end{prooftree}

\begin{prooftree}
 \AxiomC{$\Phi, (\phi  \land \diam{\gamma} \diam{\gamma^\cross} \phi)^{\nam{a}}$}
 \LeftLabel{$(\nam{a} \leqfp[\xi] \diam{\gamma^\cross} \phi)$}
 \RightLabel{\RuCross}
 \UnaryInfC{$\Phi, (\diam{\gamma^\cross} \phi)^{\nam{a}}$}
\end{prooftree}
\end{minipage}
\begin{minipage}[t]{.3\textwidth}
\begin{prooftree}
 \AxiomC{$\Phi,(\psi  \land \phi)^{\nam{a}}$}
 \RightLabel{\RuTestA}
 \UnaryInfC{$\Phi, (\diam{\psi ?} \phi)^{\nam{a}}$}
\end{prooftree}

\begin{prooftree}
 \AxiomC{$\Phi,(\psi  \lor \phi)^{\nam{a}}$}
 \RightLabel{\RuTestD}
 \UnaryInfC{$\Phi, (\diam{\psi !} \phi)^{\nam{a}}$}
\end{prooftree}
\end{minipage}

\medskip
\begin{minipage}[b]{.99\textwidth}
\begin{prooftree}
 \AxiomC{$[\Phi, \diam{\gamma^\cross} \varphi^{\nam{{ax}}}]^{\nam{x}}$}
 \noLine
 \UnaryInfC{$\vdots$}
 \noLine
 \UnaryInfC{$\Phi, (\varphi \land \diam{\gamma} \diam{\gamma^\cross}
\varphi)^{\nam{{ax}}}$}
 \LeftLabel{($\nam{a} \leqfp[\xi] \nam{x} \in N_{\diam{\gamma^\cross} \varphi}, \nam{x} \notin \Phi,\nam{a}$)}
 \RightLabel{$\RuClo_\nam{x}$}
 \UnaryInfC{$\Phi, (\diam{\gamma^\cross} \varphi)^{\nam{a}}$}
\end{prooftree}
\end{minipage}

\end{minipage}
}
\caption{The axiom and rules of the system $\CloG$. In the side condition of $\RuClo_x$, ``$\nam{x} \notin \Phi,\nam{a}$'' means that $\nam{x}$ does not 
occur in $\Phi$ or $\nam{a}$.
}
\label{fig:CloG}
\end{figure}

Completeness of $\CloG$ will follow from the completeness of the system $\CloM$, which we introduce in Section~\ref{sec:CloM}.
We prove Theorem~\ref{thm:sound-complete-CloG}
in Section~\ref{sec:s-c}.

\begin{theorem}[Soundness and Completeness of \CloG]
\label{thm:sound-complete-CloG}
For all $\xi \in\formNorm$, we have
$\CloG \vdash \xi \text{ iff } \models \xi$.
\end{theorem}

\subsection{Removing annotations with the bullet translation}

In order to translate $\CloG$-proofs into $\GLms$-proofs,
we must remove annotations. 
First, we introduce some notation.
We let
\[
\dd{\gamma}{\phi} \isdef 
   \left\{\begin{array}{ll}
      \dd{\gamma_{1}}{\dd{\gamma_{2}}{\phi}} 
      & \text{ if } \gamma = \gamma_{1}\comp\gamma_{2}
   \\ \diam{\gamma}\phi & \text{ otherwise }
   \end{array}\right.
\]
That is, if $\gamma = \gamma_{1}\comp \cdots \comp \gamma_{k}$, and none of the
game terms $\gamma_{i}$ is itself a composition, then
$\dd{\gamma}{\phi} = \diam{\gamma_{1}}\cdots \diam{\gamma_{k}}\phi$.

Given a sequence $\vec{\phi} = \phi_1,\ldots,\phi_{n}$ of formulas 
and a game term $\gamma$, we define
\[\begin{array}{lll}
   \ula{\phi}          & \isdef 
   & \phi_{n}{!} \comp ( \dots ( \phi_{1}{!} ) \dots)
\\ \ulag{\phi}{\gamma} & \isdef 
   & \phi_{n}{!} \comp ( \dots ( \phi_{1}{!} \comp \gamma ) \dots).
\end{array}\]


We can now define the translation $\startr{}{(-)}$ which removes annotations. Intuitively, what this translation does is to weaken fixpoint formulas by adding dual tests corresponding to formulas associated with names in the annotation of a fixpoint formula. This will be used to ``remember'' contexts in which greatest fixpoint formulas have been unfolded. The translation needs to be set up carefully, so that it can be used to transform $\CloG$-proofs to $\GLms$-proofs. In particular, it is tailored to fit with the strengthened induction rule in $\GLms$.

\begin{definition}\label{def:bullet-tr}
Assume that we have an assignment 
$\{\chi_{\nam{x}} \mid \nam{x} \in N\}$ of
a game logic formula $\chi_{\nam{x}} \in \formNorm$ to each name $\nam{x}$.
We define the \emph{bullet translation} $\startr{}{(-)}$ from annotated game logic formulas to $\langNorm$ by
$\startr{\eps}{\phi} = \phi$, and for non-empty annotations $\nam{a}$ as follows:
\[
\begin{aligned}
   \startr{\nam{a}}{p} &\isdef   p, 
\\ \startr{\nam{a}}{(\neg p)} &\isdef   \neg p, \quad
\\ \phantom{x}
   \startr{\nam{a}}{(\varphi \lor \psi)} & \isdef   
     \startr{\nam{a}}{\varphi} \lor \startr{\nam{a}}{\psi}, 
\\ \startr{\nam{a}}{(\varphi \land \psi)} & \isdef  
     \startr{\nam{a}}{\varphi} \land \startr{\nam{a}}{\psi}, 
\\ \startr{\nam{a}}{(\diam{\gamma} \varphi)} & \isdef  
     \dd{\startrg{\nam{a}}{\gamma}{\varphi}}{\startr{\nam{a}}{\varphi}},\quad
\end{aligned}
\begin{aligned}
   \startrg{\nam{a}}{g}{\varphi}   &\isdef   g,
\\ \startrg{\nam{a}}{g^d}{\varphi} &\isdef   g^d,
\\ \startrg{\nam{a}}{\psi?}{\varphi} &\isdef   (\startr{\nam{a}}{\psi})?, 
\\ \startrg{\nam{a}}{\psi!}{\varphi} &\isdef   (\startr{\nam{a}}{\psi})!,
\\ \startrg{\nam{a}}{\gamma^*}{\varphi} &\isdef   \gamma^*, 
\\
\end{aligned}
\]
\[
\begin{aligned}
 \startrg{\nam{a}}{\gamma \comp \delta}{\varphi} &{}\isdef  
     \startrg{\nam{a}}{\gamma}{\diam{\delta} \varphi} \comp
     \startrg{\nam{a}}{\delta}{\varphi}, 
\\ \startrg{\nam{a}}{\gamma \gcha \delta}{\varphi} &{}\isdef  
 \startrg{\nam{a}}{\gamma}{\varphi} \gcha \startrg{\nam{a}}{\delta}{\varphi}, 
\\ \startrg{\nam{a}}{\gamma \gchd \delta}{\varphi} &{}\isdef  
   \startrg{\nam{a}}{\gamma}{\varphi} \gchd \startrg{\nam{a}}{\delta}{\varphi}. 
\\
\end{aligned}
\]
The crucial clause of the translation is the case for the demonic iteration. 
If $\nam{a} = \nam{b} \nam{x}_1 \cdots \nam{x}_n \nam{c}$,
where $\nam{x}_1,\dots,\nam{x}_n$ are all the names for
$\diam{\gamma^\cross} \varphi$ in $\nam{a}$, then we
define
\[
\startrg{\nam{b} \nam{x}_1 \cdots \nam{x}_n \nam{c}}{\gamma^\cross}{\varphi} 
  \isdef (\ulag{\chi}{\gamma})^\cross;\ula{\chi}
\]
where $\vec{\chi} \isdef \chi_{\nam{x}_{1}},\ldots,\chi_{\nam{x}_{n}}$.
Note that as a special case we have
$\startrg{\nam{a}}{\gamma^\cross}{\varphi} = \gamma^\cross$ if there are
no names for $\diam{\gamma^\cross} \varphi$ in $\nam{a}$.
\end{definition}

The bullet translation only affects the outermost
fixpoint operators of a game term. This does, however, not mean that
there is only ever one fixpoint affected in a formula. For instance
when following the trace of the formula $\diam{(g \comp
(h^\cross))^\cross} p$ in some $\CloG$-proof the
fixpoint might unravel such that we obtain the formula $\diam{h^\cross}
\diam{(g \comp (h^\cross))^\cross} p$. Applying the bullet
translation to this formula might affect the outermost fixpoints of
both modalities.

The following lemma shows how the bullet translation applies to annotated 
fixpoint formulas. 
It is needed in the proof of Theorem~\ref{thm:CloG-to-KozenG} below.

\def\stateLemmaBulletTrCross{
  Let $\nam{a} = \nam{b} \nam{x}_1 \dots \nam{x}_n$ where $\nam{x}_1,\dots,\nam{x}_n$ are all the names in $\nam{a}$ for $\diam{\gamma^\cross} \phi \in F^\cross$.
Then we have:
\begin{eqnarray} 
\label{bulltr of cross}
\startr{\nam{a}}{(\diam{\gamma^\cross} \phi)} &=&
\diam{(\ulag{\chi}{\gamma})^\cross}\dd{\ula{\chi}}\startr{\nam{b}}{\phi} \text{ and }\\
\label{bulltr of unfolded cross}
\startr{\nam{a}}{(\phi \land \diam{\gamma} \diam{\gamma^\cross} \phi)}
&=&
\startr{\nam{b}}{\phi} \land \dd{\gamma}
\diam{(\ulag{\chi}{\gamma})^\cross}
\dd{\ula{\chi}}\startr{\nam{b}}{\phi}\quad
\end{eqnarray}
}
\begin{lemma} 
\label{lem:bulltr-of-cross} \label{lem:startr-of-cross}
\stateLemmaBulletTrCross
\end{lemma}


\subsection{Embedding $\CloG$ into $\GLms$}

We are now ready to show how $\CloG$-derivations
can be transformed to $\GLms$-derivations.
This will be used in Section~\ref{sec:s-c}
to transfer completeness from $\CloG$ to $\GLms$
and soundness from $\GLms$ to $\CloG$.

\begin{theorem}\label{thm:CloG-to-KozenG}
  For all $\xi \in \langNorm$,
  if $\CloG \vdash \xi$ then $\GLms \vdash \xi$.
\end{theorem}
%
%
\rem{
The theorem is proven by transforming a $\CloG$-proof 
of $\xi^\eps$ into  a $\GLms$-proof of $\xi$.
We will only provide a proof sketch 
here and refer the reader to the appendix for more details.
}
\begin{proof}
Consider a game logic formula $\xi$ and assume that $\pi$ is a proof of 
$\xi^\eps$ in $\CloG$. 
We assume that each application of the $\RuClo$-rule in $\pi$ introduces a
distinct name, i.e., for any distinct pair of rule applications 
$\RuClo_{\nam{x_1}}$ and $\RuClo_{\nam{x_2}}$ in $\pi$ we have $\nam{x_1} \not= 
\nam{x_2}$. 
This assumption is w.l.o.g. as we can rename the variable names occurring in 
$\pi$ appropriately if needed. 
The shape of the rules of $\CloG$ also imply that for each variable name 
$\nam{x}$ occurring in $\pi$, there is a corresponding occurrence of the 
$\RuClo_{\nam{x}}$-rule.
        
We now assign a formula $\chi_{\nam{x}}$ to each variable name $\nam{x}$ 
occurring in $\pi$.
This assignment is defined by induction on the distance of the (unique)
$\RuClo_{\nam{x}}$ instance in $\pi$ from the root of $\pi$. 
Concretely, for a variable name $\nam{x}$ we consider the sequent $\Phi$ 
consisting of the side formulas of the application of $\RuClo_{\nam{x}}$ in 
$\pi$ and set $\chi_{\nam{x}} \isdef \nnf{\startr{}{\Phi}}$.
Here the bullet translation of $\overline{\Phi}$ is well-defined as any variable 
name $\nam{y}$ occurring in $\Phi$ must have been introduced by an instance of
$\RuClo_{\nam{y}}$ that is closer to the root of the proof tree than 
$\RuClo_{\nam{x}}$, so that the formula $\chi_{\nam{y}}$ is already defined by 
the induction hypothesis.
\medskip

We now show how to transform the $\CloG$-proof $\pi$ of $\xi^\eps$ into a 
$\GLms$-proof of $\xi$ by demonstrating that 
(i) for all (discharged) assumptions $\Phi$ of $\pi$  there is a \GLms-derivation 
of $\startr{}{\Phi}$, and 
(ii) for all \CloG-rule applications $\Phi_1/\Phi_2$ in $\pi$ there is a 
corresponding \GLms derivation of $\startr{}{\Phi_2}$ from assumptions in 
$\startr{}{\Phi_1}$.

Consider first the bullet translation of an arbitrary discharged assumption of 
an application of $\RuClo_{\nam{x}}$ in $\pi$.
Such a translation is of the form
$\startr{}{\Phi},\startr{\nam{ax}}{(\diam{\gamma^\cross} \varphi)}$
for some annotated sequent $\Phi$ and a game logic formula
$(\diam{\gamma^\cross} \varphi)$.
Furthermore, by definition we have $\chi_{\nam{x}} = \nnf{\startr{}{\Phi}}$.
Now consider the following \GLms proof:  
%
%

{\small 
\begin{prooftree}
 \AxiomC{$\startr{}{\Phi}, \chi_{\nam{x}}$}
 \RightLabel{\RuWeak}
 \UnaryInfC{$\startr{}{\Phi}, \chi_{\nam{x}}, \theta$}
 \RightLabel{\RuOr}
 \UnaryInfC{$\startr{}{\Phi}, \chi_{\nam{x}} \lor \theta$}
 \RightLabel{\RuBang}
 \UnaryInfC{$\startr{}{\Phi}, \diam{\chi_{\nam{x}}!}\theta$}
 \AxiomC{$\startr{}{\Phi}, \chi_{\nam{x}}$}
 \RightLabel{\RuWeak}
 \UnaryInfC{$\startr{}{\Phi}, 
   \chi_{\nam{x}}, \diam{\ulag{\chi}{\gamma}} 
   \diam{(\chi_{\nam{x}}!\comp \ulag{\chi}{\gamma})^\cross} 
   \diam{\chi_{\nam{x}}!}\theta
  $}
 \RightLabel{\RuOr}
 \UnaryInfC{$\startr{}{\Phi}, 
   \chi_{\nam{x}} \lor \diam{\ulag{\chi}{\gamma}} 
   \diam{(\chi_{\nam{x}}! \comp \ulag{\chi}{\gamma})^\cross} 
   \diam{\chi_{\nam{x}}!}\theta
  $}
\RightLabel{\RuBang}
 \UnaryInfC{$\startr{}{\Phi}, 
   \diam{\chi_{\nam{x}}!}
   \diam{\ulag{\chi}{\gamma}} 
   \diam{(\chi_{\nam{x}}! \comp \ulag{\chi}{\gamma})^\cross} 
   \diam{\chi_{\nam{x}}!}\theta
  $}
\RightLabel{\RuCompD}
 \UnaryInfC{$\startr{}{\Phi}, 
   \diam{\chi_{\nam{x}}! \comp \ulag{\chi}{\gamma}} 
   \diam{(\chi_{\nam{x}}! \comp \ulag{\chi}{\gamma})^\cross} 
   \diam{\chi_{\nam{x}}!}\theta
  $}
\RightLabel{\RuAnd}
\BinaryInfC{$\startr{}{\Phi}, 
   \diam{\chi_{\nam{x}}!}\theta \land 
   \diam{\chi_{\nam{x}}! \comp \ulag{\chi}{\gamma}} 
   \diam{(\chi_{\nam{x}}! \comp \ulag{\chi}{\gamma})^\cross} 
   \diam{\chi_{\nam{x}}!}\theta
  $}
 \RightLabel{\RuCross}
\UnaryInfC{$\startr{}{\Phi},
   \diam{(\chi_{\nam{x}}! \comp \ulag{\chi}{\gamma})^\cross} 
   \diam{\chi_{\nam{x}}!}\theta
   $}
\end{prooftree}}
\noindent
where $\theta \isdef \dd{\ula{\chi}}\startr{\nam{a}}{\phi}$
and $\vec{\chi} = \chi_{\nam{x}_{1}},\ldots,\chi_{\nam{x}_{n}}$
with $\nam{x}_1,\dots,\nam{x}_n$ being all names of
$\diam{\gamma^\cross}\phi$ in $\nam{a}$.
The remaining assumption in this \GLms proof is the sequent 
$\startr{}{\Phi}, \chi_{\nam{x}} = \startr{}{\Phi}, \nnf{\startr{}{\Phi}}$.
But in fact for \emph{any} finite set $\Psi = \{ \psi_{1}, \ldots, \psi_{n}\}$
we can easily derive the sequent $\Psi, \nnf{\Psi} = \Psi, 
\nnf{\psi_{1}}\land \cdots \land \nnf{\psi_{n}}$ in \GLms using $n$
instances of \Ax and \RuWeak followed by an application of \RuAnd.
Using Lemma~\ref{lem:startr-of-cross} one can verify that 
\[ \startr{}{\Phi},
   \diam{(\chi_{\nam{x}}! \comp \ulag{\chi}{\gamma})^\cross} 
   \diam{\chi_{\nam{x}}!}\theta
   =  \startr{}{\Phi},\startr{\nam{ax}}{(\diam{\gamma^\cross} \varphi)} \]
which shows that we have  constructed the required \GLms derivation of the
translated assumption. 

We show claim (ii) above, ie., that for each rule application in $\pi$
there is a corresponding \GLms derivation. We only consider the rules
\RuExp and \RuClo.
For the other rules the reasoning is either trivial or it follows from
reasoning that is similar but simpler as the one for \RuClo.

Suppose that an instance of the \RuExp-rule is applied in $\pi$ to obtain
$\Phi,\phi^{\nam{axb}}$ from  $\Phi,\phi^{\nam{ab}}$. 
Let $\theta = \diam{\gamma^\cross}\phi'$ be the fixpoint formula
corresponding to $\nam{x}$ and suppose w.l.o.g.
that $\theta \in F^\cross(\phi)$ and that the bullet translation
$\startr{\nam{ab}}{\phi}$ is of the form
$\psi(\dd{(\ulag{\chi} \gamma)^\cross; \ula{\chi}}\psi')$ where $\vec{\chi} = \chi_{\nam{x}_1} \dots \chi_{\nam{x}_n}$
are the context formulas corresponding to the names $\nam{x}_1,\dots,\nam{x}_n$ of $\theta$ that
occur in $\nam{ab}$. Let $\vec{\chi'}= \nam{x}_1 \dots \nam{x} \dots \nam{x}_n$ be the list of names
of $\theta$ in $\nam{axb}$. Then 
$\startr{\nam{axb}}{\phi} = \psi(\dd{(\ulag{\chi'} \gamma)^\cross; \ula{\chi'}}\psi')$ and it is now easy to see
that this formula is derivable from $\startr{\nam{ab}}{\phi}$ in $\GLms$ by applying the \RuMonDgame-rule
twice for each occurrence of $\theta$ that got expanded by the bullet translation. 

Lastly, consider an application of the \RuClo-rule in $\pi$ that derives
from ${\Phi},
{(\varphi\land\diam{\gamma}\diam{\gamma^\cross}\varphi)}^\nam{ax}$ the
conclusion $\Phi, (\diam{\gamma^\cross}\varphi)^{\nam{a}}$. We need to
construct a corresponding \GLms derivation. First observe that by
Lemma~\ref{lem:startr-of-cross} we have 
\[
\startr{\nam{a}}{(\diam{\gamma^\cross} \phi)} =
\diam{(\ulag{\chi}{\gamma})^\cross}\dd{\ula{\chi}}\startr{\nam{b}}{\phi},
\]
where $\nam{b}$ and $\vec{\chi}$ are chosen as in the previous case.
Furthermore 
{\small
\begin{align*}
\startr{\nam{ax}}{(\varphi \land \diam{\gamma} \diam{\gamma^\cross} \varphi)} 
   & = \startr{\nam{b}}{\varphi} \land
        \dd{\gamma}
        \diam{(\chi_{\nam{x}}!\comp\ulag{\chi}{\gamma})^\cross} 
        \diam{\chi_{\nam{x}}!}\dd{\ula{\chi})}
        \startr{\nam{b}}{\phi},
\\ & = \startr{\nam{b}}{\varphi} \land
        \dd{\gamma}
        \diam{(\nnf{\startr{}{\Phi}}!\comp\ulag{\chi}{\gamma})^\cross} 
        \diam{\nnf{\startr{}{\Phi}}!}\dd{\ula{\chi})}
        \startr{\nam{b}}{\phi}
\end{align*}}
where we again used Lemma~\ref{lem:startr-of-cross} and the fact that
$\chi_{\nam{x}} = \nnf{\startr{}{\Phi}}$. Now we build the following
\GLms derivation:
\begin{prooftree}
\AxiomC{$\startr{}{\Phi},
   \startr{\nam{b}}{\varphi} \land
   \dd{\gamma}
   \diam{(\nnf{\startr{}{\Phi}!} ; \ulag{\chi}{\gamma})^\cross} 
   \diam{\nnf{\startr{}{\Phi}}!} 
   \dd{\ula{\chi}}\startr{\nam{b}}{\varphi}$}
\doubleLine
\RightLabel{\RuCompD}
\UnaryInfC{$\startr{}{\Phi}, 
   \startr{\nam{b}}{\phi} \land 
   \diam{\gamma}
   \diam{(\nnf{\startr{}{\Phi}!} ; \ulag{\chi}{\gamma})^\cross} 
   \diam{\nnf{\startr{}{\Phi}}!} 
   \dd{\ula{\chi}}\startr{\nam{b}}{\varphi}$}
\doubleLine
\RightLabel{\RuMonDgame}
\UnaryInfC{$\startr{}{\Phi}, 
   \startr{\nam{b}}{\phi} \land 
   \diam{\ulag{\chi}{\gamma}}
   \diam{(\nnf{\startr{}{\Phi}}! ; \ulag{\chi}{\gamma})^\cross}
   \diam{\nnf{\startr{}{\Phi}}!} 
   \dd{\ula{\chi}}\startr{\nam{b}}{\varphi}$}
\doubleLine
\RightLabel{\RuMonDfmla}
\UnaryInfC{$\startr{}{\Phi}, 
   \dd{\ula{\chi}}\startr{\nam{b}}{\phi} \land 
   \diam{\ulag{\chi}{\gamma}}
   \diam{(\nnf{\startr{}{\Phi}!} ; \ulag{\chi}{\gamma})^\cross}
   \diam{\nnf{\startr{}{\Phi}}!} 
   \dd{\ula{\chi}}\startr{\nam{b}}{\varphi}$}
\RightLabel{\RuIndS}
\UnaryInfC{$\startr{}{\Phi}, 
   \diam{(\ulag{\chi}{\gamma})^\cross}
   \dd{\ula{\chi}}\startr{\nam{b}}{\phi},
   $}
\end{prooftree}
Here, the double lines indicate that multiple applications
of the specified rule could be required to reach the next sequent. 
Using the equations given above the proof tree, we have 
given a \GLms-derivation of $\startr{}{\Phi}, \startr{\nam{a}}{(\diam{\gamma^\cross} \phi)}$
from assumption
$\startr{}{\Phi}, \startr{\nam{ax}}{(\varphi\land\diam{\gamma}\diam{\gamma^\cross}\varphi)}$.
This shows that for each instance of \RuClo there is a corresponding \GLms-derivation as required.
\end{proof}


\section{The monotone $\mu$-calculus}
\label{sec:mon-mu}


In this section we give the basic definitions of the monotone $\mu$-calculus,
and we introduce an annotated proof system for it.

\subsection{The monotone $\mu$-calculus: syntax and semantics}

Additionally to the sets $\AtProp$ and $\AtGame$ from
Section~\ref{subsec:game logic basics} we now also fix a countable set
$\Var$ of \emph{fixpoint variables}. 
We shall only consider $\mu$-calculus formulas in negation normal form.

\begin{definition} \label{def:langmonmml}
The language $\langM$ of the \emph{monotone $\mu$-calculus} consists of the
formulas:
\[
\begin{array}{ccl}
  \langM \ni A, B
  & ::= &
  p \mid \lnot p \mid x
  \mid A \lor B
  \mid A \land B \\
 & &
 \phantom{p}
  \mid \diam{g} A
  \mid \diam{g^d} A \mid \mu x . A \mid \nu x . A
\end{array}
\]
where $p \in \AtProp$, $g \in \AtGame$ and $x \in \Var$.

We apply the usual notions concerning variable binding, writing 
$\Var(A)/\FVar(A)$ for the sets of all/all free variables in $A$.
A formula $A$ is a \emph{sentence} if $\FVar(A) = \emptyset$.
\end{definition}

This is essentially the language of a multi-modal $\mu$-calculus,
except we write $\diam{g^d} \varphi$ instead of $[g]\varphi$
in order to stay closer to game logic syntax.

Given a $\langM$-formula $A$, we define its (Fischer-Ladner) closure $\Cl{A}$
in the usual way (via subformulas and unfoldings).
Our definition of the system $\CloM$ below crucially involves the following
priority order $\leqAL[A]$ on $\Var(A)$, and a notion of well-namedness 
from \cite{AfshariLeigh:LICS17}.

\begin{definition} \label{def:AL-order}
 Let $A$ be a $\mu$-calculus formula and $x, y \in \Var(A)$. We write $x
\lessAL[A] y$ if for some subformula of $A$ of the form
$\sigma y . B$ where $\sigma \in \{\mu,\nu\}$,
the variable $x$ occurs freely in $\sigma y . B$. We denote by
$\lesstrAL[A]$ the transitive closure of $\lessAL[A]$ on $\Var(A)$. We
denote by $\leqAL[A]$ the reflexive, transitive closure of $\lessAL[A]$
on $\Var(A)$. We say that $A$ is \emph{locally well-named} if
$\lesstrAL[A]$ is irreflexive.
\end{definition}
We show examples of $\lessAL[A]$ in Examples~\ref{ex:tr-ord-1}~and~\ref{ex:tr-ord-2} below. 

The semantics of the monotone $\mu$-calculus over the game models from 
Section~\ref{subsec:game logic basics} is standard.

\begin{definition} \label{def:monML semantics}
We define the \emph{meaning} $\mng{A}^{\bbS}_h$ of a formula $A \in \langM$ in 
the game model $\bbS = (S,E,V)$, relative to an assignment $h : \Var \to \wp(S)$,
by a standard induction, where, e.g., $\mng{\diam{g^d} A}^{\bbS}_h \isdef 
S \setminus E_{g}(S \setminus \mng{A}^{\bbS}_h)$, and 
\[
\begin{array}{lll}
   \mng{\mu x . A}^{\bbS}_h & \isdef & \lfp\, X . \,
\mng{A}^\bbS_{h[x \mapsto X]}
\\ \mng{\nu x . A}^{\bbS}_h & \isdef & \gfp\, X . \,
\mng{A}^\bbS_{h[x \mapsto X]}.
\end{array}
\]
Here $h[x \mapsto X]$ is the assignment $h'$ given by
$h'(x) \isdef X$ and $h'(y) \isdef h(y)$ for $y \neq x$.

The meaning of a sentence $A$ in $\bbS$ does not depend on the  assignment, and
so we denote this set as $\mng{A}^\bbS$.
Notions like satisfiability, validity, etc., are all defined in the standard 
way.
\end{definition}

\subsection{The $\CloM$-system for the monotone $\mu$-calculus}
\label{sec:CloM}

The proof system $\CloM$ is the monotone analogue of the annotated sequent
system $\Clo$ for the $\mu$-calculus from \cite{AfshariLeigh:LICS17}, with one 
further difference. 
In order to prove our key Proposition~\ref{p:translation} below, we need more
control over the order on fixpoint variables than allowed by the global order 
of~\cite{AfshariLeigh:LICS17}.
Instead, we will use the fixed order $\leqAL[C]$, for some ambient formula $C$, 
and define $\CloM$ to be parametric in such a $C$.

So fix a well-named $\mu$-calculus formula $C$.
In the derivation system $\CloM[C]$, formulas are annotated with sequences of
names, as in the system $\CloG$.  
To each variable $\var{x} \in \Var(C)$ we link a set $\names[\var{x}]$ of
\emph{names for $\var{x}$}, in such a way that $\names[\var{x}]\cap\names[\var{y}]=
\emptyset$ if $\var{x}\neq\var{y}$. 
As in Section~\ref{subsec:clog}, we let the names inherit the order $\leqAL[C]$
over variables, introduce annotations, and extend the relation $\leqAL[C]$
to hold between annotations and variables.

The proof rules of $\CloM[C]$ are in Figure~\ref{fig:CloM}.
A $\CloM[C]$-proof is a finite tree of inferences in which each leaf is 
labelled by an axiom or a discharged assumption.
For a formula $A\in \langM$, 
we write $\CloM \vdash A$ if there is a $\CloM[A]$-proof of $A^{\eps}$.

Analogous to $\CloG$, the system $\CloM$ is analytic in the sense that any 
$\CloM[C]$-proof of $C^\eps$ contains only formulas from $\Cl{C}$ and
names for variables in $\Var(C)$. 
Hence, the order $\leqAL[C]$ is defined for these names and variables.

\begin{figure}
\fbox{
\begin{minipage}[t]{.47\textwidth}
\begin{minipage}[t]{.49\textwidth}
\begin{prooftree}
 \AxiomC{\phantom{X}}
 \RightLabel{\AxOne}
 \UnaryInfC{$p^{\nam{\eps}}, (\neg p)^{\nam{\eps}}$}
\end{prooftree}
\begin{prooftree}
 \AxiomC{$\Gamma, A^{\nam{a}} , B^{\nam{a}} $}
 \RightLabel{\RuOr}
 \UnaryInfC{$\Gamma, (A  \lor B)^{\nam{a}} $}
\end{prooftree}
\end{minipage}
\begin{minipage}[t]{.49\textwidth}
\begin{prooftree}
 \AxiomC{$A^{\nam{a}} , B^{\nam{b}} $}
 \RightLabel{\RuMonMod}
 \UnaryInfC{$\diam{g}A^{\nam{a}} , \diam{g^d}B^{\nam{b}}$}
\end{prooftree}
\begin{prooftree}
 \AxiomC{$\Gamma, A^{\nam{a}} $}
 \AxiomC{$\Gamma, B^{\nam{a}} $}
 \RightLabel{\RuAnd}
 \BinaryInfC{$\Gamma, (A \land B)^{\nam{a}} $}
\end{prooftree}
\end{minipage}

\medskip
\begin{minipage}[b]{.32\textwidth}
\begin{prooftree}
 \AxiomC{$\Gamma$}
 \RightLabel{\RuWeak}
 \UnaryInfC{$\Gamma,A^{\nam{a}} $}
\end{prooftree}
\end{minipage}
\begin{minipage}[b]{.66\textwidth}
\begin{prooftree}
 \AxiomC{$\Gamma, A(\mu x . A(x))^{\nam{a}}$}
 \LeftLabel{$(\nam{a} \leqAL[C] \var{x})$} 
 \RightLabel{\RuLFP}
 \UnaryInfC{$\Gamma, (\mu x . A(x))^{\nam{a}}$}
\end{prooftree}
\end{minipage}

\medskip

\begin{minipage}[b]{.32\textwidth}
\begin{prooftree}
  \AxiomC{$\Gamma, A^{\nam{{ab}}},$}
 \RightLabel{\RuExp}
 \UnaryInfC{$\Gamma, A^{\nam{{axb}}}$}
\end{prooftree}
\end{minipage}
\begin{minipage}[b]{.66\textwidth}
\begin{prooftree}
 \AxiomC{$\Gamma, A(\nu x . A(x))^{\nam{a}}$}
 \LeftLabel{$(\nam{a} \leqAL[C] \var{x})$} 
 \RightLabel{\RuGFP}
 \UnaryInfC{$\Gamma, (\nu x . A(x))^{\nam{a}}$}
\end{prooftree}
\end{minipage}

\medskip

\begin{minipage}[b]{.99\textwidth}
\begin{prooftree}
 \AxiomC{$[\Gamma, \nu x . A(x)^{\nam{{ax}}}]^{\nam{x}}$}
 \noLine
 \UnaryInfC{$\vdots$}
 \noLine
 \UnaryInfC{$\Gamma, A(\nu x . A(x))^{\nam{{ax}}}$}
 \LeftLabel{($\nam{a} \leqAL[C] \nam{x} \in N_{x}, \nam{x} \notin
\Gamma, \nam{a}$)}
 \RightLabel{$\RuNuClo_{\nam{x}}$}
 \UnaryInfC{$\Gamma, (\nu x . A(x))^{\nam{a}}$}
\end{prooftree}
\end{minipage}

\end{minipage}
}
\caption{The axiom and rules of the system $\CloM[C]$.}
\label{fig:CloM}
\end{figure}

We show that
$\CloM$ is sound and complete for the semantics
of the monotone $\mu$-calculus
given in Definition~\ref{def:monML semantics} 
by reduction to the soundness and completeness of the system $\Clo$
with respect to Kripke models, which has been proven in
\cite{AfshariLeigh:LICS17}.

The reduction uses a translation $(-)^t \colon \langM \to \langBimod$
from monotone $\mu$-calculus into normal $\mu$-calculus 
that is based on well-known ideas, going back to
\cite{KrachtWolter99},
though our approach is closer to the one from \cite[Ch.10]{Helle03},
for simulating monotone modal logics with normal modal logics.
The basic idea is that an effectivity function
$E \colon \wp(S) \to \wp(S)$
corresponds to a monotone neighbourhood function
$N \colon S \to \wp(\wp(S))$
which can be encoded with two Kripke relations on state space $S \cup \wp(S)$:
A relation $R_N \subseteq S \times \wp(S)$ that relates states to their neighbourhoods,
and a relation $R_\ni \subseteq \wp(S) \times S$ that relates neighbourhoods to their elements. 
Conversely, from two Kripke relations on a state space $S$
one can define
a monotone neighbourhood function $N \colon S \to \wp(\wp(S))$ 
using the same idea.
The language $\langBimod$ is the modal $\mu$-calculus language
that has two (normal) modalities $\diam{g_N}$ and $\diam{g_\ni}$ for each $g \in \AtGame$.
The translation of atomic formulas, Boolean connectives and fixpoints
is defined by straightforward recursion. For modalities, we take
$(\diam{g}A)^t = \diam{g_{N}}[g_{\ni}] (A^t)$,
and
$(\diam{g^d}A)^t = [g_{N}]\diam{g_{\ni}}(A^t)$.
\rem{
  \begin{tabular}{lcl}
$(\mu x. A)^t = \mu x. (A^t)$, $(\diam{g}A)^t = \diam{g_{N}}[g_{\ni}] (A^t)$,\\ 
$(\nu x. A)^t = \nu x. (A^t)$, $(\diam{g^d}A)^t = [g_{N}]\diam{g_{\ni}}(A^t)$.
\end{tabular}\\[.3em]
}
\noindent Using the model translatons described above, we can show that 
$(-)^t$ preserves satisfiability and validity.

\begin{lemma}\label{lem:t-sat-val}
For all formulas $C \in \langM$,
\begin{enumerate}
\item if $C$ is satisfiable in a game model,
then $C^t$ is satisfiable in a Kripke model for $\langBimod$.
\item if $C$ is valid over game models, 
then $C^t$ is valid over Kripke models for $\langBimod$.
\end{enumerate}
\end{lemma}

Lemma~\ref{lem:t-sat-val} gives the semantic part needed in the proof of
the following theorem.

\begin{theorem}[Soundness and Completeness of ${\CloM}$]
\label{thm:CloM-complete}
For all 
$C \in \langM$,
$C$ is valid (on game models) iff ${\CloM} \vdash C$.
\end{theorem}

A detailed proof of Theorem~\ref{thm:CloM-complete}
is found in the appendix. Here we only give a sketch.
What remains is to show that we can translate between proofs in $\CloM$
and proofs in the similar annotated proof system
$\Clo$ \cite{AfshariLeigh:LICS17} for the normal $\mu$-calculus.

First, we note that the translation $(-)^t$ extends to annotated formulas and sequents in the obvious manner.
For both directions, we transform proof trees starting from the root going up.
The most interesting case is in 
the construction of a $\CloM$-proof for $A \in \langM$ from a $\Clo$-proof of $A^t \in \langBimod$
when the modal rule from $\Clo$ is applied.
So suppose some node $v$ in a $\Clo$-proof $\pi$ is obtained from an application of the (normal) modal rule,
and $v$ is labelled with a sequent $\Gamma^t$
where $\Gamma$ is a sequent of annotated $\langM$-formulas.
Then $\Gamma^t$ must have the form
$\diam{g_{N}} [g_\ni] A_1 ,..., \diam{g_{N}} [g_\ni] A_n, [g_{N}]\diam{g_{\ni}} B$, and 
hence $\Gamma$ must have the form
$\diam{g} C_1,...,\diam{g} C_n, \diam{g^d} D$ where $C_1^t = A_1,...,C_n^t = A_n$, $D^t = B$. 
By inspection of the rules of $\Clo$, and assuming that $n\geq 2$ (since the other case is easier), 
we see that the subtree of the $\Clo$-proof $\pi$ rooted at $v$
must have the following shape: 
\begin{prooftree}
\AxiomC{}
 \noLine
 \UnaryInfC{$\vdots$}
 \noLine
\UnaryInfC{$A_i,B$}
\RightLabel{$\mathsf{mod}$}
\UnaryInfC{$[g_\ni] A_i,\diam{g_{\ni}} B$}
\RightLabel{$\RuWeak$}
 \UnaryInfC{$ [g_\ni] A_1 ,...,  [g_\ni] A_n,\diam{g_{\ni}} B$}
 \RightLabel{$\mathsf{mod}$}
 \UnaryInfC{$\diam{g_{N}} [g_\ni] A_1 ,..., \diam{g_{N}} [g_\ni] A_n, [g_{N}]\diam{g_{\ni}} B$}
\end{prooftree}
We mimic this by the following $\CloM$-derivation steps: 
\begin{prooftree}
\AxiomC{$C_i,D$}
\RightLabel{$\mathsf{mod}$}
 \UnaryInfC{$\diam{g} C_i, \diam{g^d} D$}
 \RightLabel{$\RuWeak$}
 \UnaryInfC{$\diam{g} C_1,...,\diam{g} C_n, \diam{g^d} D$}
\end{prooftree} 
The label of the top node of this derivation translates to $A_i,B$, and so we can inductively continue the construction using the corresponding smaller subtree of $\pi$. 

\section{Game Logic and the monotone $\mu$-calculus}
\label{sec:translation}
\label{sec:complete}


In this section we define a novel translation from formulas in game
logic to formulas in the monotone $\mu$-calculus, and 
prove that if the translation of a formula is provable in \CloM then the
formula is provable in \CloG.

\subsection{Translating Game Logic to the monotone $\mu$-calculus}
\label{sec:sharp-trans}

It is shown in \cite[sec.~6.4.2]{Pau:phd} that game logic can be translated 
into the two-variable fragment of the monotone $\mu$-calculus.
However, we use more than two variables because we need to keep track
of the nesting of fixpoints.
Before we give the formal definition of our translation,
we first explain informally how we achieve this.
Consider the translation of a game logic formula $\xi\in \langNorm$.
Formulas $\diam{\gamma^\circ}\phi \in F(\xi)$ translate
to fixpoint formulas of the form $\sigma x. A(x)$ on the $\mu$-calculus side.
In order to synchronise the translation across unfolding of fixpoint formulas,
we syntactically encode $\diam{\gamma^\circ}\phi$ into the fixpoint variable
that it gives rise to in the translation of $\xi$.

\begin{definition}\label{def:tr}
We define the translation $\tr{(-)}\colon \langNorm \to \langM$
by a mutual induction on formulas and games as follows:
\[
\begin{array}{rcl}
   \tr{p} & \isdef  & p 
\\ \tr{(\neg p)} & \isdef  & \neg p 
\\ \tr{(\varphi \land \psi)} & \isdef  & \tr{\varphi} \land \tr{\psi} 
\\ \tr{(\varphi \lor \psi)} & \isdef  & \tr{\varphi} \lor \tr{\psi} 
\\ \tr{(\diam{\gamma} \varphi)} & \isdef  & \itrg{\varphi}{\gamma}(\tr{\varphi})
\\[.3em]
 \itrg{\varphi}{g}(A) & \isdef  & \diam{g} A \\
 \itrg{\varphi}{g^d}(A) & \isdef  & \diam{g^d} A \\
 \itrg{\varphi}{\gamma \gchd \delta}(A) & \isdef  &
\itrg{\varphi}{\gamma}(A) \land \itrg{\varphi}{\delta}(A) \\
 \itrg{\varphi}{\gamma \gcha \delta}(A) & \isdef  &
\itrg{\varphi}{\gamma}(A) \lor \itrg{\varphi}{\delta}(A) \\
 \itrg{\varphi}{\gamma^*}(A) & \isdef  & \mu x^{\diam{\gamma^*} \varphi} . A
\lor \itrg{\diam{\gamma^*}\varphi}{\gamma}(x^{\diam{\gamma^*} \varphi})
\\
 \itrg{\varphi}{\gamma^\cross}(A) & \isdef  & \nu x^{\diam{\gamma^\cross}
\varphi} . A \land
\itrg{\diam{\gamma^\cross}\varphi}{\gamma}(x^{\diam{\gamma^\cross}
\varphi}) \\
 \itrg{\varphi}{\gamma \comp \delta}(A) & \isdef  &
\itrg{\diam{\delta}\varphi}{\gamma}(\itrg{\varphi}{\delta}(A)) \\
 \itrg{\varphi}{\psi ?}(A) & \isdef  & \tr{\psi} \land A \\
 \itrg{\varphi}{\psi !}(A) & \isdef  & \tr{\psi} \lor A \\
\end{array}
\]
\end{definition}
\rem{In the translation of a game term $\gamma^\circ$, whose top-level
operator is an iterator $\circ \in \{*,\cross\}$, we tag the introduced
fixpoint variable with the formula in the closure of the translated game
logic formula that this fixpoint is coming from.
}
\begin{example}\label{ex:tr-ord-1}
For $\varphi = \diam{(a^* \comp (b^\cross \gcha c))^\cross} p$, the translation is 
\[\begin{array}{rcl}
\tr{\varphi} &=&
 \nu x^{\varphi} . p \land  \\ & & \quad \mu x^{\psi} .
         ( (\nu x^\theta . x^\phi \land \diam{b}x^\theta) \lor \diam{c} x^\phi )  \lor \diam{a}x^\psi 
\\
\text{ with }
\psi   &=& \diam{a^*}\diam{b^\cross \gcha c}\phi\\
\theta &=& \diam{b^\cross}\phi
\end{array}\]
Applying the definitions of order on game logic fixpoint formulas (Def.~\ref{def:F-order}) and $\mu$-calculus fixpoint variables (Def.~\ref{def:AL-order}), we find that:
\[
\phi \lessfp \psi,\;
\phi \lessfp\theta
\quad\text{ and } \quad
x^\phi \lessAL[\tr{\phi}] x^\psi,\;
x^\phi \lessAL[\tr{\phi}] x^\theta
\]
\end{example}

\begin{example}\label{ex:tr-ord-2}
For $\phi = \diam{ (a^* ; (\diam{b^\cross}p)?)^\cross}\diam{c^*}q$,
\[\begin{array}{rcl}
\tr{\phi} &=&
    \nu x^\phi . \, ( \mu x^\psi . q \lor \diam{c}x^\psi)
                \land\\
          & & \qquad\;
                (\mu x^\zeta . ( (\nu x^\theta . p \land \diam{b}x^\theta) \land x^\phi )  \lor \diam{a}x^\zeta)
\\
\text{ with } \psi &=& \diam{c^*}q\\
 \zeta &=& \diam{a^*}\diam{ \diam{b^\cross}p)? }\phi, \quad\text{and}\\
 \theta &=& \diam{b^\cross}p
\end{array}\]
Applying the definitions of order on game logic fixpoint formulas (Def.~\ref{def:F-order}) and $\mu$-calculus fixpoint variables (Def.~\ref{def:AL-order}), we find that:
\[
\phi \lessfp \zeta,\; \phi \lessfp\theta
\quad\text{ and } \quad
x^\phi \lessAL[\tr{\phi}] x^\zeta
\]
\end{example}

\rem{For example, the translation $\tr{\varphi}$ of the formula
$\varphi = \diam{(a^* \comp b)^\cross} p$ is $\nu x^{\varphi} . p \land \mu
x^{\psi} . \diam{b} x^{\varphi} \lor \diam{a} x^{\psi}$, where $\psi =
\diam{a^*} \diam{b} \diam{(a^* \comp b)^\cross} p$.
}

The above examples illustrate how
the order on fixpoint variables in $\mu$-calculus
is reflected in game logic fixpoints along the translation,
and that translations are always locally well-named.
These are the syntactic properties of $\tr{(-)}$ that are crucial to our proofs.

\begin{proposition} \label{p:translation}
For all $\xi \in \formNorm$ the translation $\tr{\xi}$ is locally
well-named, and for all $\varphi,\psi \in F(\xi)$ we have
$x^\varphi,x^\psi \in \Var(\tr{\xi})$, and that $x^\varphi
\leqAL[\tr{\xi}] x^\psi$ implies $\phi \leqfp \psi$.
\end{proposition}



On the semantic side, our translation is adequate in the sense that
it is truth- and validity preserving.
Recall that $\langNorm$ and $\langM$ are both interpreted over game models,
i.e., monotone neighbourhood models.

\begin{proposition} 
    \label{p:adequacy translation}
For every  $\xi \in \langNorm$ and every game model
$\bbS$ it holds that $\mng{\xi}^\bbS = \mng{\tr{\xi}}^\bbS$.
\end{proposition}
\begin{proof}
By a straightforward induction.
\end{proof}

\subsection{From \CloM to \CloG}
\label{sec:CloM-to-CloG}

We now show how to construct a $\CloG$-derivation of a game logic
formula $\xi$ from a $\CloM$-derivation of $\tr{\xi}$.

For this purpose, we identify the set $N_\varphi$ of names for
$\varphi \in F(\xi)$ with the set $N_{x^\varphi}$ of names for the
variable $x^\varphi \in \Var(\tr{\xi})$. This is possible since both
sets are defined to be arbitrary countable sets.
We then extend the translation $\tr{(-)}$ to annotated formulas and sequents by taking
\begin{equation*}\label{eq:named-tr}
  \tr{(\varphi^{\nam{a}})} \isdef (\tr{\varphi})^{\nam{a}}
  \quad\text{ and }\quad
  \tr{\Phi} \isdef \{ \tr{(\phi^{\nam{a}})} \mid \phi^{\nam{a}} \in \Phi \}.  
\end{equation*}
That is, the translation leaves annotations unchanged.

\begin{theorem}
  \label{thm:CloM-to-CloG}
  \label{mu to gl}
For all $\xi\in\langNorm$,
if $\CloM \vdash\tr{\xi}$ then  $\CloG \vdash \xi$.
\end{theorem}

\begin{proof}
We will prove the theorem by induction on the complexity of proof trees, and
for a proper development of the induction we need to take care of derivations 
with open branches because the \RuClo-rule allows to discharge assumptions.
We shall write $\pi: \mathcal{A} \vdash_{\CloM[C]} \Gamma$ to say that $\pi$
is a $\CloM[C]$-derivation of $\Gamma$ from assumptions in $\mathcal{A}$, and 
similarly for $\CloG$-derivations with open assumptions.

More precisely, we shall prove, by induction on the complexity of 
$\CloM$-derivations, that every $\CloM[\tr{\xi}]$-proof $\pi$ satisfies the 
following property:
\begin{center}
\fbox{\parbox{8cm}{
 for every game logic sequent $\Phi$:
    if $\pi: \mathcal{A} \vdash_{\CloM[\tr{\xi}]} \tr{\Phi}$
\\ then there is a $\CloG[\xi]$-proof $\pi': \mathcal{G} \vdash_{\CloG} \Phi$
where $\tr{\mathcal{G}} = \mathcal{A}$.
}}\hspace*{0.3cm}($\ast$)
\end{center}

Two preliminary remarks are in order before we dive into the proof details.
First, in the sequel we will often omit the annotation of formulas, for the sake 
of readability.
And second, without loss of generality we may adopt the \emph{injectivity 
assumption} stating that for each formula $A$ in $\tr{\Phi}$ there is precisely 
one formula $\varphi$ in $\Phi$ with $\tr{\varphi} = A$.
\medskip

In the base case of our proof, the derivation $\pi$ is either an application of 
the axiom \AxOne or a one-node derivation of a sequent $\tr{\Phi}$, where the 
set of assumptions of $\pi$ is the singleton set $\{ \tr{\Phi}\}$.
In both cases it is straightforward to see that the derivation $\pi'$, 
consisting of a single node labelled $\Phi$, meets the requirements stated 
in ($\ast$).
%
\medskip

For the inductive step, first observe that we may assume that none of the
formulas in $\Phi$ is of the form $\diam{\gamma \comp \delta} \psi$. 
Should $\varphi \in \Phi$ be of this form then we could apply the rule \RuComp 
and subsequently work with the formula $\diam{\gamma} \diam{\delta} \psi$, for
which it holds that $\tr{(\diam{\gamma} \diam{\delta} \psi)} =
\tr{(\diam{\gamma \comp \delta} \psi)}$. 
This can be repeated until the resulting formula is of the required shape.
\smallskip

For the proof of the inductive step, we make a case distinction as to the last 
applied rule in the $\CloM$-derivation $\pi$.

In case the last applied rule is the rule \RuAnd, then $\tr{\Phi}$ must be of 
the form $\tr{\Phi} = \Gamma, A_{0} \land A_{1}$ and the rule \RuAnd is applied
to the premises $\Gamma,A_{0}$ and $\Gamma,A_{1}$. 
By our injectivity assumption there is precisely one formula $\varphi$ in $\Phi$
such that $\Phi = \Psi, \varphi$, $\tr{\Psi} = \Gamma$ and $\tr{\varphi} = A_{0}
\land A_{1}$.
But then it follows by the definition of the translation $\tr{(-)}$ and our 
assumption on the shape of the formulas in $\Phi$ that there are three 
possibilities: either 
(i) $\varphi = \phi_{0} \land \phi_{1}$ such that $\tr{\phi_{0}} = A_{0}$ and
$\tr{\phi_{1}} = A_{1}$, or 
(ii) $\varphi = \diam{\gamma_{0} \gchd \gamma_{1}} \psi$ such that
$\tr{(\diam{\gamma_{0}}\psi)} = A_{0}$ and $\tr{(\diam{\gamma_{1}}\psi)} = A_{1}$,
or (iii)
$\phi = \diam{\psi ?}\chi$ such that $\tr{\psi} = A_{0}$ and $\tr{\chi} = A_{1}$.


The other cases being similar, we only consider case (ii).
Here we have $\CloM[\tr{\xi}]$-proofs $\pi_{0}, \pi_{1}$ of the sequents
$\tr{\Psi},\tr{\diam{\gamma_{0}} \psi}$ and 
$\tr{\Psi},\tr{\diam{\gamma_{1}} \psi}$, 
from two respective sets of assumptions $\mathcal{A}_{0}$ and $\mathcal{A}_{1}$
such that $\mathcal{A}_{0} \cup \mathcal{A}_{1} = \mathcal{A}$.
Use the induction hypothesis to obtain, for $i = 0,1$, a set $\mathcal{G}_{i}$ 
of game logic sequents such that $\tr{\mathcal{G}_{i}} = \mathcal{A}_{i}$, as 
well as a $\CloG[\xi]$-proof $\pi_{i}': \mathcal{A}_{i} \vdash \Psi,
\diam{\gamma_{i}} \psi$.
We then apply the rule \RuAnd to get a proof of the sequent 
$\Psi, \diam{\gamma_{0}} \psi \land \diam{\gamma_{1}} \psi$, followed by the 
rule \RuChoiceD to derive the sequent 
$\Phi = \Psi, \diam{\gamma_{0} \gchd \gamma_{1}} \psi$.
Finally, the set of assumptions of the resulting derivation $\pi'$ is the set
$\mathcal{G}_{0} \cup \mathcal{G}_{1}$, which clearly satisfies the condition 
that $\tr{(\mathcal{G}_{0} \cup \mathcal{G}_{1})} = \mathcal{A}$.
\smallskip

The cases where the last rule applied in $\pi$ is one of \RuOr, \RuMonMod, or
\RuWeak, are similarly easy to deal with; we omit the details.

%
%

Now consider the case where $\pi$ ends with an application of the rule
\RuLFP for a least fixpoint. We then have that $\tr{\Phi} = \Gamma, \mu
x. A(x)^{\nam{a}}$, the premise of this application of $\RuLFP$ is the sequent
$\Gamma, A(\mu x. A(x))^{\nam{a}}$, and the side condition $\nam{a} 
\leqAL[\tr{\xi}] \var{x}$ is fulfilled. 
As explained above we can assume that there is a single formula $\varphi$ in 
$\Phi$ such that $\Phi = \{\Psi, \varphi\}$,
$\tr{\phi}{\Psi} = \Gamma$ and $\tr{\varphi} = \mu x. A(x)$.
As we have already excluded the possibility that $\varphi$ is a modality whose
main operator is the composition it follows from the definition of the 
translation $\tr{(-)}$ that $\varphi = \diam{\gamma^*} \psi$ such that 
$A(x) = \tr{\psi} \lor \itrg{\varphi}{\gamma}(x)$.
Note that $x = x^{\diam{\gamma^*}\psi}$ by definition of the translation
$\tr{(-)}$.
Some further calculations show that
\[
A(\mu x . A(x)) 
  = \tr{\psi} \lor \itrg{\varphi}{\gamma}(\mu x . A(x)) 
  = \tr{(\psi \lor \diam{\gamma} \diam{\gamma^*} \psi)}.
\]
We can thus apply the induction hypothesis to obtain a $\CloG[\xi]$-proof of the
sequent $\Psi, \psi \lor \diam{\gamma} \diam{\gamma^*} \psi$, from a 
set of assumptions $\mathcal{G}$ satisfying $\tr{\mathcal{G}} = \mathcal{A}$.
We then want to use the rule \RuStar to obtain a proof of
$\Phi = \Psi,\varphi= \Psi, \diam{\gamma^*} \psi$ from the same set $\mathcal{G}$ of assumptions.
To do so we need to ensure that the side condition $\nam{a} \leqfp[\xi]
\diam{\gamma^*} \psi$ is satisfied.
Hence consider any name $\nam{y}$ that
occurs in $\nam{a}$ and let $\chi$ be the fixpoint formula such that $\nam{y} 
\in N_\chi$.
From the side condition $\nam{a} \leqAL[\tr{\xi}] x$ it follows that $\nam{y}
\leqAL[\tr{\xi}] x$, and then from 
Proposition~\ref{p:translation}
that $\chi \leqfp[\xi] \diam{\gamma^*} \psi$, and hence we obtain the required
$\nam{y} \leqfp[\xi] \diam{\gamma^*} \psi$.
\smallskip

If the last rule applied in $\pi$ is the fixpoint rule \RuGFP for the
greatest fixpoint then we can use a similar argument as in the paragraph
using \RuCross instead of \RuStar.
\medskip

Finally, consider the case where the last rule applied in $\pi$ is 
$\RuNuClo_{\nam{x}}$ 
for some name $\nam{x}$, discharging the assumption $\Omega = \Gamma, \nu x. A(x)^{\nam{{ax}}}$. 
We then may observe that $\tr{\Phi} = \Gamma, \nu x. A(x)^{\nam{a}}$, that the premise 
of this application of $\RuNuClo_{\nam{x}}$ is the sequent $\Gamma, A(\nu x. A(x))^{\nam{{ax}}}$, 
and that the side conditions $\nam{a} \leqAL[\tr{\xi}] \nam{x}$ and
$\nam{x} \not\in \Gamma,\nam{a}$ are fulfilled. 
As explained above we can assume that
there is a single formula $\varphi$ in $\Phi$ such that $\Phi = \{\Psi,
\varphi\}$, $\tr{\Psi} = \Gamma$ and $\tr{\varphi} = \nu x. A(x)$.
And, similar to the case of the rule \RuLFP discussed above,
we may assume that 
$\varphi = \diam{\gamma^\cross} \psi$ for $\gamma$ and $\psi$ such that $A(x) = 
\tr{\psi} \land \itrg{\varphi}{\gamma}(x)$, and that  
\[
A(\nu x . A(x)) 
  = \tr{\psi} \land \itrg{\varphi}{\gamma}(\nu x . A(x)) 
  = \tr{(\psi \land \diam{\gamma} \diam{\gamma^\cross} \psi)}.
\]
We then apply the induction hypothesis and obtain a $\CloG$-derivation of the 
premise of the $\RuNuClo_{\nam{x}}$-rule from assumptions $\mathcal{G} \cup 
\mathcal{G'}$, where each sequent in $\mathcal{G'}$ translates to $\Omega$ (the
assumption discharged by the application of the $\RuNuClo_{\nam{x}}$-rule with
conclusion $\tr{\Phi}$).
It follows that every sequent in $\mathcal{G'}$ must be of the form $\Theta,
\phi_0^{\nam{ax}}$ with $\tr{\Theta} = \Gamma$ and $\tr{\phi_0} = \nu x. A(x)$.
\rem{In addition it follows from the shape of its annotation (viz., $\nam{ax}$), that 
$\phi_0$ belongs to the closure of the formula $\phi = \diam{\gamma^\cross}\psi$.
But since we encode the formula $\diam{\gamma^\cross}\psi$ into the fixpoint
variable of its translation, this can only be the case if $\phi_0$ is actually
syntactically identical to $\phi$.}
From $\tr{\phi_0} = \nu x. A(x) = \tr{\phi}$ it follows that $\phi_0 = \phi$ (syntactically),
since we encode the formula $\phi = \diam{\gamma^\cross}\psi$ into the fixpoint
variable $x$ of its translation.
That is, we may take $\mathcal{G'} =  \{ \Theta, \diam{\gamma^{\cross}}
\psi^{\nam{{ax}}} \mid \Theta \in \mathcal{L} \}$ for some set $\mathcal{L}$ 
with $\tr{\mathcal{L}_{\Sigma}} = \{ \Gamma \}$.
Note, however, that the sequents in $\mathcal{L}$ will generally not be 
identical to $\Psi$, which means that we cannot simply finish our proof with an
application of the $\RuNuClo_{\nam{x}}$-rule of the $\CloG$-system here. We need 
a more elaborate construction.

In fact we need to generalise the statement about $\Psi$ and $\nam{a}$ to the
observation below, where we let $\mathcal{S}$ be the (finite!) set of game logic
sequents $\Sigma$ such that $\tr{\Sigma} = \Gamma$.

\begin{claimfirst}
\label{cl:dag}
For every $\Sigma \in \mathcal{S}$ there are game logic sequents 
$\mathcal{G}_{\Sigma}$ and $\mathcal{L}_{\Sigma}$ such that 
(\dag 1) 
$\tr{\mathcal{G}_{\Sigma}} = \mathcal{A}$, 
   $\tr{\mathcal{L}_{\Sigma}} = \{ \Gamma \}$
   and
(\dag 2) 
{for every $\nam{b} = \nam{ax}_{1}\cdots \nam{x}_{k}$, 
   with $\nam{x}_{1},\ldots, \nam{x}_{k}$ names for $\diam{\gamma^{\cross}}\psi$, 
   there is a $\CloG[\xi]$-proof }
\[\rho^{\nam{{b}}}_{\Sigma}: \, 
   \mathcal{G}_{\Sigma} \cup \{ \Theta, \diam{\gamma^{\cross}} 
   \psi^{\nam{b}} 
      \mid \Theta \in \mathcal{L}_{\Sigma} \}
   \;\vdash\; 
   \Sigma, (\psi \land \diam{\gamma} \diam{\gamma^{\cross}} 
   \psi)^{\nam{b}}
\]
\end{claimfirst}
\begin{pfclaim}
Fix a sequent $\Sigma \in \mathcal{S}$.
Repeating the argument that we just gave and that is directly based on 
the inductive hypothesis, we obtain sets of game logic sequents 
$\mathcal{G}_{\Sigma}$ and $\mathcal{L}_{\Sigma}$ satisfying condition (\dag1),
together with a $\CloG[\xi]$-derivation $\rho^{\nam{ax}}_{\Sigma}$ of the
sequent 
$\Sigma, (\psi \land \diam{\gamma} \diam{\gamma^{\cross}} \psi)^{\nam{{ax}}}$
from the assumptions 
$\mathcal{G}_{\Sigma} \cup \{ \Theta, \diam{\gamma^{\cross}} \psi^{\nam{{ax}}}
\mid \Theta \in \mathcal{L}_{\Sigma} \}$.

Now consider an annotation $\nam{b} = \nam{a}\ol{\nam{x}}$, where 
$\ol{\nam{x}} = \nam{x}_{1}\cdots \nam{x}_{k}$, with $k \geq 1$.
We will transform the derivation $\rho^{\nam{{ax}}}_{\Sigma}$ into the desired
derivation $\rho^{\nam{{b}}}_{\Sigma}$ in two stages.
First, we simply replace every occurrence of $\nam{x}$ as (part of) an annotation 
in $\rho^{\nam{{ax}}}_{\Sigma}$ with $\ol{\nam{x}}$.
This transforms $\rho^{\nam{{ax}}}_{\Sigma}$ into a structure $\rho'$ which is 
\emph{almost} a proper $\CloG[\xi]$-proof.
The only problem concerns applications in $\rho^{\nam{{ax}}}_{\Sigma}$ of the 
\emph{expansion} rule \RuExp of the form $\Delta, \phi^{\nam{{cd}}} / 
\Delta,\phi^{\nam{{cxd}}}$.
In $\rho'$ we may not be allowed to derive $\Delta,\phi^{\nam{{c\ol{\nam{x}}d}}}$ from
$\Delta,\phi^{\nam{{cd}}}$ by one application of the expansion rule,
but we can easily take care of this problem in the second state of the 
construction, namely by deriving $\Delta,\phi^{\nam{{c\ol{\nam{x}}d}}}$ from 
$\Delta, \phi^{\nam{{cd}}}$ by a \emph{series} of applications of the expansion
rule.
This finishes the proof of the claim.
\end{pfclaim}

We will use derivations of the form $\rho^{\nam{{b}}}_{\Sigma}$ as 
\emph{building blocks} for our $\CloG[\xi]$-derivation of the sequent $\Psi, 
\diam{\gamma^{\cross}} \psi^{\nam{{a}}}$.
The idea is to first build up, step by step, a 
\emph{pseudo-derivation} of $\Psi, \diam{\gamma^{\cross}} \psi^{\nam{{a}}}$ which 
differs from a proper $\CloG[\xi]$-proof in that not all assumptions of prospective 
applications of the $\RuClo$-rule are discharged.
Once we have completed the construction of this pseudo-derivation, we transform
it into a proper $\CloG[\xi]$-proof by taking care of these undischarged assumptions.
To do this in a proper way we need to be precise about the annotations, and we
need to introduce some auxiliary definitions.

Most importantly, we define a \emph{pseudo-derivation} to be a proof in the 
derivation system $\CloG[\xi]$ extended with the derivation rule \RuD:
\begin{prooftree}
 \AxiomC{$\Phi, (\varphi \land \diam{\gamma} 
    \diam{\gamma^\cross}\varphi)^{\nam{{a}\nam{x}}}$}
 \LeftLabel{$\nam{a} \leqfp[\xi] 
    \nam{x} \in N_{\diam{\gamma^\cross} \varphi}, 
    \nam{x} \notin \Phi,\nam{a}$}
 \RightLabel{$\RuD_{\nam{x}}$}
 \UnaryInfC{$\Phi, (\diam{\gamma^\cross} \varphi)^{\nam{a}}$}
\end{prooftree}
Clearly, \RuD is identical to the rule \RuClo, apart from the fact that it does
not require that the assumptions of the form $\Psi, \diam{\gamma^\cross} 
\varphi^{\nam{{ax}}}$ in the proof tree leading up to the premise of \RuD are 
discharged.
We shall call a node $t$ in a proof tree \emph{dangling} if the rule applied at
$t$ is \RuD.
Observe that a pseudo-derivation is a proper $\CloG[\xi]$-derivation just in case it
has no dangling nodes. 

We now construct a pseudo-derivation for the sequent $\Psi, \diam{\gamma^{\cross}}
\psi^{\nam{{a}}}$.
We shall make use of a set $\{ \nam{x}_{\Sigma} \mid \Sigma \in \mathcal{S} \}$ of
special, fresh names, all associated with the fixpoint formula 
$\diam{\gamma^\cross}\psi$.
Our starting point of the construction is the one-node derivation consisting 
of the sequent $\Psi, \diam{\gamma^{\cross}} \psi^{\nam{{a}}}$.

Now suppose that the current approximation $\sigma$ of the pseudo-derivation 
contains an assumption of the form $\Sigma, \diam{\gamma^{\cross}} \psi^{\nam{{b}}}$, 
where $\Sigma \in \mathcal{S}$ and the annotation $\nam{b}$ is of the form
$\nam{b} = \nam{a} \overline{\nam{x}}$ with $\nam{x}_{\Sigma}$ \emph{not} 
occurring in the sequence 
$\overline{\nam{x}} = \nam{x}_{\Sigma_{1}}\cdots \nam{x}_{\Sigma_{k}}$.
By our Claim~\ref{cl:dag}, we may take a $\CloG[\xi]$-proof
$\rho_{\Sigma}^{\nam{b}\nam{x}_{\Sigma}}$ of the sequent
$\Sigma, (\psi \land \diam{\gamma} \diam{\gamma^{\cross}} \psi)^{\nam{{bx}}_{\Sigma}}$
from the assumptions $\mathcal{G}_{\Sigma} \cup \{ \Theta, \diam{\gamma^{\cross}}
\psi^{\nam{{bx}}_{\Sigma}} \mid \Theta \in \mathcal{L}_{\Sigma} \}$.
We adjoin copies of the derivation $\rho_{\Sigma}^{\nam{b}\nam{x}_{\Sigma}}$ 
to the derivation tree, linking each leaf in the current approximation $\sigma$ 
which is labelled as indicated, to the root of a copy of 
$\rho_{\Sigma}^{\nam{{b}\nam{x}_{\Sigma}}}$
through an application of the rule $\RuD_{\nam{x}_\Sigma}$.

The above construction must terminate after finitely many steps, basically as
a consequence of the fact that the set $\mathcal{S}$ is finite.
Let $\rho$ denote the pseudo-derivation that we arrive at in this way, and let
$\mathcal{G}$ be the set of assumptions of $\rho$ that belong to the set
$\bigcup \{ \mathcal{G}_{\Sigma} \mid \Sigma \in \mathcal{S}\}$; 
clearly then we have that $\tr{\mathcal{G}} = \mathcal{A}$.

It is not difficult to verify that the pseudo-derivation $\rho$ satisfies the
following conditions:
\begin{enumerate}
\item 
All leaves of $\rho$ are labelled with an axiom, a sequent from $\mathcal{G}$, 
or else a sequent of the form 
$\Sigma, \diam{\gamma^{\cross}} \psi^{\nam{{b}}}$, where $\Sigma \in \mathcal{S}$
and the annotation $\nam{b}$ is of the form $\nam{b} = \nam{a} 
\nam{x}_{\Sigma_{1}}\cdots \nam{x}_{\Sigma_{k}}$,
with $\Sigma_{1} = \Psi$, $\Sigma \in \{ \Sigma_{1}, \ldots, \Sigma_{k} \}$,
and the $\Sigma_{i}$ are all distinct.
\item
If a leaf $l$ is labelled $\Sigma, \diam{\gamma^{\cross}} \psi^{\nam{{b}}}$, where
$\nam{b} = \nam{a} \nam{x}_{\Sigma_{1}}\cdots \nam{x}_{\Sigma_{k}}$, then the path from the root $r$ of 
$\rho$ to $l$ passes through nodes $r = t_{1},\ldots, t_{k}$, in that order, 
such that 
(a) every $t_{j}$ is either dangling or the conclusion of an application of 
the \RuClo-rule, and
(b) the name $\nam{x}_{\Sigma_{i}}$ was introduced at the successor of $t_{i}$.
\item 
If $t$ is a dangling node of $\rho$, labelled, say, with the sequent $\Sigma, 
\diam{\gamma^{\cross}} \psi^{\nam{{b}}}$, and $l$ is a leaf above $t$ labelled
with $\Sigma, \diam{\gamma^{\cross}} \psi^{\nam{c}}$, then $\nam{bx}_{\Sigma}$ is an initial 
segment of $\nam{c}$.
\end{enumerate}

Step by step we will now transform this pseudo-derivation into a proper 
$\CloG[\xi]$-derivation. 
Clearly it suffices to prove that we can turn any pseudo-derivation 
satisfying the conditions 1) -- 3) into a pseudo-derivation that still satisfies 
mentioned conditions, but has a smaller number of dangling nodes.

So let $\sigma$ be such a pseudo-derivation, and pick a dangling node, say, $t$,
that has maximal distance to the root; this means in particular that there are 
no dangling nodes above $t$. 
Let $t$ and its successor be labelled with, respectively, the sequents
$\Sigma, \diam{\gamma^{\cross}} \psi^{\nam{{a}} \overline{\nam{x}}}$
and 
$\Sigma, (\psi \land \diam{\gamma} \diam{\gamma^\cross} 
  \psi)^{\nam{{a}} \overline{\nam{x}}\nam{x}_{\Sigma}}$, 
and let $L_{t}$ be the set of leaves above $t$ that are labelled with a sequent
of the form $\Sigma, \diam{\gamma^{\cross}} \psi^{\nam{{b}}}$.
Now make a case distinction.

If $L_{t}$ is empty, the pseudo-derivation does not record a proper circular
dependency at $t$, so to speak. 
This is in fact the simplest case: we obtain a pseudo-derivation $\sigma'$ from 
$\sigma$ by (a) replacing $\RuD_{\nam{x}}$ with $\RuCross$ as the rule applied
at $t$, and (b) simply erasing all occurrences of the name $\nam{x}_{\Sigma}$ in the
pseudo-derivation above $t$.

If $L_{t}$ is non-empty, consider an arbitrary leaf $l$ in $L_{t}$, and
let $\Sigma, \diam{\gamma^{\cross}} \psi^{\nam{{b}}_{l}}$ be the sequent labelling $l$.
It follows from condition 3 that $\nam{b}_{l}$ is of the form
$\nam{a} \overline{\nam{x}} \nam{x}_{\Sigma} \nam{c}_{l}$ for some sequence
$\nam{c}_{l}$.
Now, extend $\sigma$ to $\sigma'$ by attaching a successor $l'$ to each $l\in L_{t}$
(so $l'$ is a leaf in $\sigma'$, but $l$ is not),
and label each such $l'$ with the sequent 
$\Sigma, \diam{\gamma^\cross} \psi^{\nam{{a}} \overline{\nam{x}}\nam{x}_{\Sigma}}$,
so that we may obtain the sequent of $l$ from that of $l'$ by applications of 
the expansion rule.
We then obtain the desired pseudo-derivation from $\sigma'$ 
by discharging the assumption
$\Sigma, \diam{\gamma^\cross} \psi^{\nam{a} \overline{\nam{x}}\nam{x}_{\Sigma}}$
at every leaf $l'$ with $ l \in L_t$, 
and simultaneously changing the proof rule applied at node $t$ into a 
(now legitimate) application of the $\RuClo_{\nam{x}_\Sigma}$-rule.

In both cases it is not hard to verify that the structure $\sigma'$ is in fact a
(pseudo-)derivation satisfying the clauses 1) -- 3), that the node $t$ is not a 
dangling node of $\sigma'$, and that the transformation of $\sigma$ into 
$\sigma'$ has not created any new dangling node. 

Finally, as a result of these transformations we obtain, as required, a 
$\CloG[\xi]$-derivation of the sequent $\Phi, \diam{\gamma^{\cross}}\psi$ from
the collection of assumptions $\mathcal{G}$ for which we already saw that 
$\tr{\mathcal{G}} = \mathcal{A}$.

This finishes the proof of Theorem~\ref{mu to gl}, since the \RuClo-rule 
was the last rule to be considered in the induction step.
\end{proof}

\section{Soundness and Completeness via Transformations}
\label{sec:s-c}

We now prove the soundness and completeness
of the proof systems
$\GLms$ (Theorem~\ref{thm:compGLms}) and
$\CloG$ (Theorem~\ref{thm:sound-complete-CloG})
as well as the completeness of $\HilbP$ (Theorem~\ref{t:par-completeness}).
We do this
using the translations and transformations we introduced earlier.
An overview is given by the following diagram.
Here, $\Clo$ is the system from \cite{AfshariLeigh:LICS17}, and 
$\langNormAnn$, $\langMAnn$ and $\langBimodAnn$ denote, respectively, the sets of 
annotated formulas of $\langNorm$, $\langM$ and $\langBimod$.
 
\[\xymatrix@R=2mm@C=6mm{
  \HilbP
  & \ar[l]_-{\text{Thm~\ref{thm:GLms-to-HilbP}}} \GLms
  & \ar[l]_-{\text{Thm.~\ref{thm:CloG-to-KozenG}}} \CloG
  & \ar[l]_-{\text{Thm.~\ref{thm:CloM-to-CloG}}} \CloM
  & \ar[l] \Clo
  \\
  \langPar \ar@<2pt>[r]^-{\nf{-}}
  & \ar@<2pt>[l]^-{\pari{-}} \langNorm
  & \ar[l]_-{\startr{}{(-)}} \langNormAnn \ar[r]^-{\tr{(-)}}
  & \langMAnn \ar[r]^-{(-)^t}
  & \langBimodAnn
}\]

The completeness of $\CloG$ and $\GLms$ is obtained from the completeness of $\CloM$,
and the fact that Proposition~\ref{p:adequacy translation} implies that
the translation $\tr{(-)}$ preserves validity over game models.
Hence for all $\xi \in \langNorm$ we find (\dag)
\[
\models \xi
\Ra[\text{Prop.~\ref{p:adequacy translation}}]
\models \tr{\xi}
\Ra[\text{Thm.~\ref{thm:CloM-complete}}]
\CloM \vdash \tr{\xi}
\Ra[\text{Thm.~\ref{thm:CloM-to-CloG}}]
\CloG \vdash \xi
\Ra[\text{Thm.~\ref{thm:CloG-to-KozenG}}]
\GLms \vdash \xi
\]

From the completeness of $\GLms$, we obtain the completeness
of $\HilbP$ as follows.
For all $\phi \in \langPar$,  we have
  \[
  \models \phi
  \quad\LRa[\text{Prop.~\ref{p:trl}}]\quad \models \nf{\phi}
  \quad\Ra[\text{(\dag)}]\quad
  \GLms \vdash \nf{\phi}
  \Ra[\text{Thm~\ref{thm:GLms-to-HilbP}(1)}]
  \HilbP \vdash \phi
  \]
  
To prove the soundness of $\GLms$  and $\CloG$,
let $\xi\in\langNorm$. We then have,
  \[
  \CloG \vdash \xi
  \Ra[\text{Thm~\ref{thm:CloG-to-KozenG}}]
  \GLms \vdash \xi
  \Ra[\text{Thm~\ref{thm:GLms-to-HilbP}(2)}]
  \HilbP \vdash \pari{\xi}
  \]
By the soundness of $\HilbP$, it follows that $\pari{\xi}$ is valid over game
models, and since $\pari{\xi}$ is equivalent with $\xi$ by 
Proposition~\ref{p:trl}, also $\xi$ is valid over game models.

\section{Conclusion}
\label{sec:conclusion}

In this paper we introduced two cut-free sequent calculi for Parikh's
game logic and established their soundness and completeness. From this
result, we also obtained completeness of the original Hilbert-style
proof system for game logic. This confirms a conjecture made by Parikh
in \cite{Parikh85}.
The completeness of these two systems was obtained
by translating game logic into the monotone $\mu$-calculus,
for which we also gave a cut-free sequent calculus that
we showed to be sound and complete.
\rem{For instance, we showed that their annotated proof system for the modal
$\mu$-calculus can be adapted to yield a complete system for the monotone 
$\mu$-calculus as well.}

\subsection{Discussion}

Our proof makes essential use of ideas and results
from Afshari and Leigh's  paper \cite{AfshariLeigh:LICS17}. In
particular, the idea of using the proof systems $\CloG$ and $\CloM$ to obtain cut-free
completeness is central here.
An important reason that our approach is possible is that
these annotated proof systems allow good control over the
structure of proofs.
In particular, formal proofs in $\CloG$ and $\CloM$ only
contain formulas that are in the Fischer-Ladner closure of the formula
at the root of the proof. This means that if the root formula of an
annotated proof is the translation of a game logic formula, then indeed
the entire proof can in a sense be carried out within game logic, modulo
the translation.
Also, the annotations provide a powerful machinery for
keeping track of unfoldings of fixpoint formulas along traces in a proof
tree. This is crucial in order to decide where to apply the
strengthened induction rule when we construct cut-free sequent proofs
from annotated ones.
\rem{For this step, we make use of our observation that
the strengthened induction rule can be formulated in game logic by
making use of the (dual) test operator. 
}

\subsection{Future research}

Completeness for fixpoint logics is generally considered to be difficult
as witnessed by the long wait for a completeness proof for the
modal $\mu$-calculus \cite{koze:resu83,walu:comp93}
and game logic.
Our work demonstrates that the techniques from  Afshari \&
Leigh~\cite{AfshariLeigh:LICS17}
can be transferred to other fixpoint logics, and we expect that it is the
beginning of a fruitful line of research into cut-free complete proof systems
for fixpoint logics.

\rem{
We see this paper both as an application and a continuation of Afshari and
Leigh's work, contributing to it as well as building upon it. 
In particular, we hope to have shown how their techniques fit in a wider setting,
pointing towards more applications in the future and towards a uniform approach 
to axiomatization of fixpoint logics. 
}

More generally, we believe this approach can be used to provide cut-free 
complete proof systems for
coalgebraic $\mu$-calculi~\cite{cirs:expt11, enqvist2018completeness}, and
for coalgebraic dynamic logics~\cite{HK:FICS15}.
Also, there are many fragments of the modal $\mu$-calculus that could be studied
by similar techniques. 
As one example, it would be interesting to develop annotated proof systems for
CTL$^*$, and see if this could help to simplify Reynold's axiomatization  of
CTL$^*$ \cite{reyn:axio01}.
It should also be checked whether our proof can be adapted to provide a cut-free
complete proof system for PDL. 
An indication that this is possible is that the deep rules in our system \GLms are reminiscent of display 
calculi. The latter have been successfully applied to obtain a complete proof system for PDL~\cite{Fritella2016:DisplayPDL}.

Going the opposite direction, similar techniques could potentially be applied to extensions of the $\mu$-calculus, such as the two-way $\mu$-calculus \cite{vardi1998reasoning}, hybrid $\mu$-calculus \cite{sattler2001hybrid}, and alternating $\mu$-calculus \cite{Alur:2002:ATL}.

Finally, we would like to investigate applications of our cut-free proof systems
for game logic to prove interpolation. 

\section*{Acknowledgements}

We would like to thank the anonymous LICS 2019 reviewers for their helpful comments and suggestions. 
We acknowledge financial support from Swedish Research
Council grant 2015-01774, EPSRC grant EP/N015843/1,
and ERC consolidator grant 647289 CODA.

\bibliographystyle{IEEEtran}


\newpage

\appendix
\subsection{Omitted proofs of Section \ref{sec:two-sys}}


\noindent
\textbf{Proposition~\ref{p:trl}}
There are recursively defined, truth-preserving translations
\[\begin{array}{lll}
   \trParNorm{-} \colon & \langFull \to \langNorm
\\ \trNormPar{-} \colon & \langFull \to \langPar
\end{array}\]

Below we give an explicit definition of the translations $\trParNorm{-}$ and 
$\trNormPar{-}$, leaving it for the reader to check that the translations land 
in the proper fragments and are truth-preserving.

Translating from $\langPar$ to $\langNorm$ is simply taking dual and negation
normal form.

\begin{definition}\label{def:trParNorm}
We define the translation $\trParNorm{-} \colon \langFull \to \langNorm$ as 
follows:
\[\begin{array}[t]{lclcl}
  \dnnf{p} &=& p &  
  \\ \dnnf{\lnot p} &=& \lnot p & 
  \\ \dnnf{\lnot\lnot \phi} &=& \dnnf{\phi} & 
  \\ \dnnf{\lnot (\phi \lor \psi)} &=& \dnnf{\lnot\phi} \land \dnnf{\lnot\psi} 
  \\ \dnnf{\lnot (\phi \land \psi)} &=& \dnnf{\lnot\phi} \lor \dnnf{\lnot\psi} 
  \\ \dnnf{\lnot \diam{\gamma}\phi} &=& \diam{\dnnf{\gamma^d}}\dnnf{\lnot\phi} 
  \\[5mm]
    \dnnf{g} &=&  g
  \\ \dnnf{g^d} &=& g^d 
  \\ \dnnf{(\gamma^d)^d} &=& \dnnf{\gamma}
  \\ \dnnf{(\gamma \gcha \delta)^d} &=& \dnnf{\gamma^d} \gchd \dnnf{\delta^d} 
  \\ \dnnf{(\gamma \gchd \delta)^d} &=& \dnnf{\gamma^d} \gcha \dnnf{\delta^d} 
  \\ \dnnf{(\gamma ; \delta)^d} &=& \dnnf{\gamma^d} ; \dnnf{\delta^d} 
  \\ \dnnf{(\iter{\gamma})^d} &=& \diter{\dnnf{\gamma^d}} 
  \\ \dnnf{(\diter{\gamma})^d} &=& \iter{\dnnf{\gamma^d}} 
  \\ \dnnf{(\test{\phi})^d} &=& \dtest{\dnnf{\lnot\phi}} 
  \\ \dnnf{(\dtest{\phi})^d} &=& \test{\dnnf{\lnot\phi}} 
\end{array}
  \]
\end{definition}

We translate from $\langNorm$ to $\langPar$ by expanding demonic operations as
dual angelic ones.

\begin{definition}\label{def:trNormPar}
We define the translation $\trNormPar{-} \colon \langFull \to \langPar$ as
follows:
\[\begin{array}[t]{lclcl}
  \pari{p} &=& p &  
  \\ \pari{\lnot p} &=& \lnot p & 
  \\ \pari{\phi \land \psi} &=& \lnot(\lnot\pari{\phi} \lor \lnot\pari{\psi}), 
  \\ \pari{\phi \lor \psi} &=& \pari{\phi} \lor \pari{\psi} 
  \\ \pari{\diam{\gamma}\phi} &=& \diam{\pari{\gamma}}\pari{\phi}  
  \\[5mm]
    \pari{g} &=&  g
  \\ \pari{\gamma^d} &=& \gamma^d 
  \\ \pari{\gamma \gcha \delta} &=& \pari{\gamma} \gcha \pari{\delta} 
  \\ \pari{\gamma \gchd \delta} &=& (\pari{\gamma}^d \gcha \pari{\delta}^d)^d, 
  \\ \pari{\gamma ; \delta} &=& \pari{\gamma} ; \pari{\delta} 
  \\ \pari{\iter{\gamma}} &=& \iter{\pari{\gamma}} 
  \\ \pari{\diter{\gamma}} &=& (\iter{(\pari{\gamma}^d)})^d 
  \\ \pari{\test{\phi}} &=& {\test{\pari{\phi}}} 
  \\ \pari{\dtest{\phi}} &=& (\test{(\lnot\pari{\phi})})^d 
\end{array}
  \]
\end{definition}
\smallskip


\subsection{Omitted proofs of Section \ref{sec:annot-sys}}
\subsubsection{Lemmas~\ref{lem:bulltr-of-cross}~and~\ref{lem:bulltr-of-star}}

We will work towards proving Lemmas~\ref{lem:bulltr-of-cross}~and~\ref{lem:bulltr-of-star}.
To this end, we introduce some auxiliary notions. For $\nam{a} \in N^*$
and $\Gamma \subseteq F^\cross$, let $\restr{\nam{a}}{\Gamma}$ denote
the subsequence of $\nam{a}$ of all names $\nam{x}$ in $\nam{a}$ such
that $\nam{x} \in \names[\phi]$ for $\phi \in \Gamma$. Similarly, for $X
\subseteq \names$ we write $\restr{\nam{a}}{X}$ to denote the
subsequence of $\nam{a}$ consisting of names from $X$. By minor abuse of
notation, we may write $\nam{a} \subseteq X$ to indicate that all names
occurring in $\nam{a}$ are in $X$.

We define the set $\Scross(\phi) \subseteq F^\cross(\phi)$
of \emph{surface level greatest fixpoints} of $\phi \in \langNorm$ as follows.
\[
\begin{aligned}[t]
 \Scross(p) &=  \emptyset \\
 \Scross(\neg p) &=  \emptyset \\
 \Scross(\phi \land \psi) &=  \Scross(\phi) \cup \Scross(\psi) \\
 \Scross(\phi \lor \psi) &=  \Scross(\phi) \cup \Scross(\psi) \\
 \Scross(\diam{\gamma} \phi) &=  \Scross(\gamma,\phi) \cup \Scross(\phi)
\end{aligned}
\;\;\;
\begin{aligned}[t]
 \Scross(g,\phi) &=  \emptyset \\
 \Scross(g^d,\phi) &=  \emptyset \\
 \Scross(\gamma^*,\phi) &=  \emptyset \\
 \Scross(\gamma^\cross,\phi) &=  \{\diam{\gamma^\cross} \phi\}  \\
 \Scross(\psi?,\phi) &=  \Scross(\psi) \\
 \Scross(\psi!,\phi) &=  \Scross(\psi)
\end{aligned}
\]
\[\begin{aligned}[t]
 \Scross(\gamma \gcha \delta,\phi) &=  \Scross(\gamma, \phi) \cup
\Scross(\delta,\phi) \\
 \Scross(\gamma \gchd \delta,\phi) &=  \Scross(\gamma, \phi) \cup
\Scross(\delta,\phi) \\
 \Scross(\gamma \comp \delta,\phi) &=  \Scross(\gamma, \diam{\delta} \phi) \cup
\Scross(\delta,\phi) \\
\end{aligned}
\]
In other words, $\Scross(\phi)$ are those $\diam{\gamma^\cross}\psi$ in $\Fcross(\phi)$ such that $\gamma^\cross$ 
is not a direct subterms of another (angelic or demonic) iteration operator.
For example,
$\Scross(\diam{g^\cross \comp h^\cross}p) = \{ \diam{g^\cross}\diam{h^\cross}p, \diam{h^\cross}p\}$,
$\Scross(\diam{g^\cross \gcha h^\cross}p) = \{ \diam{g^\cross}p, \diam{h^\cross}p\}$,
$\Scross(\diam{(g^\cross})^*p) = \emptyset$, and
$\Scross(\diam{(\diam{g^\cross}p?)^\cross}q) = \{\diam{ (\diam{g^\cross}p?)^\cross}q \} $.
Note that $\diam{\gamma^\cross}\psi \in \Scross(\phi)$ does not imply that 
$\diam{\gamma^\cross}\psi \subterm \phi$.

\begin{lemma}
\label{lem:restriction}
For all annotations $\nam{a} \in \names^*$, formulas $\phi\in \langNorm$, game terms $\gamma \in \gameNorm$ and $X \subseteq \names$:
\begin{enumerate}
\item If $\restr{\nam{a}}{\Scross(\phi)} \subseteq X $, then
$\startr{\nam{a}}{\phi} = \startr{\restr{\nam{a}}{X}}{\phi}$.
\item 
  If $\restr{\nam{a}}{\Scross(\gamma,\phi)} \subseteq X $, then
$\startrg{\nam{a}}{\gamma}{\phi} = \startrg{\restr{\nam{a}}{X}}{\gamma}{\phi}$.
\end{enumerate}
\end{lemma}

\begin{proof}
We prove the two claims by a mutual 
induction on the subterm-relation, 
i.e.,  for any  $t \in \formNorm \cup \gameNorm$ 
the I.H. stipulates that 
\begin{itemize}
 \item claim (1) holds for all formulas  $\psi$ such that $\psi \ssubterm t$ 
 and
 \item claim (2) holds for all game terms $\delta$ such that $\delta \ssubterm t$.
\end{itemize}
For the first claim, the only interesting induction step is for modal formulas 
$\diam{\gamma}\phi \in \formNorm$. 
Suppose  $\restr{\nam{a}}{\Scross(\diam{\gamma}\phi)} \subseteq X $. We have $\Scross(\gamma,\phi) \subseteq \Scross(\diam{\gamma}\phi)$ 
and thus $\restr{\nam{a}}{ \Scross(\gamma,\phi)} \subseteq   \restr{\nam{a}}{ \Scross(\diam{\gamma}\phi)} \subseteq X$.
By the induction hypothesis of (2) on $\gamma$ we  get:
$$\startrg{\nam{a}}{\gamma}{\phi} = \startrg{\restr{\nam{a}}{X}}{\gamma}{\phi}$$
Furthermore  $\Scross(\phi) \subseteq \Scross(\diam{\gamma}\phi)$, so we have $\restr{\nam{a}}{\Scross(\phi)} \subseteq X $. 
 So by the induction hypothesis of (1) on $\phi$ we get: $$\startr{\nam{a}}{\phi} = \startr{\restr{\nam{a}}{X}}{\phi} $$
Putting these observations together we get:
\[
\begin{array}{rcl}
\startr{\nam{a}}{(\diam{\gamma}\phi)} 
   & = & \dd{\startrg{\nam{a}}{\gamma}{\phi}} \startr{\nam{a}}{\phi} 
\\ & = & \dd{\startrg{\restr{\nam{a}}{X}}{\gamma}{\phi}} \startr{\restr{\nam{a}}{X}}{\phi} \\
& = & \startr{\restr{\nam{a}}{X}}{\diam{\gamma}\phi} \\
\end{array}
\]
as required. 

We now turn to the induction on game terms, for item (2): the atomic cases for $ \gamma =  g $ or  $\gamma = g^d$ are trivial, and the induction steps for $\gcha,\gchd$ are straightforward.  For composition suppose that $\Scross(\gamma\comp \delta,\phi) \subseteq X $. Since   $\Scross(\gamma \comp \delta,\phi)  =  \Scross(\gamma, \diam{\delta} \phi) \cup
\Scross(\delta,\phi)$ we have $\restr{\nam{a}}{\Scross(\gamma, \diam{\delta}\phi)} \subseteq X $ and $\restr{\nam{a}}{\Scross( \delta,\phi)} \subseteq X $. So the induction hypothesis on $\gamma$ and $\delta$ gives:
 \begin{eqnarray*}
\startrg{\nam{a}}{\gamma \comp \delta}{\phi} 
& = &  \startrg{\nam{a}}{\gamma}{\diam{\delta} \phi} \comp
\startrg{\nam{a}}{\delta}{\phi} \\
& = & \startrg{\restr{\nam{a}}{X}}{\gamma}{\diam{\delta} \phi} \comp
\startrg{\restr{\nam{a}}{X}}{\delta}{\phi}  \\
& = & \startrg{\restr{\nam{a}}{X}}{\gamma \comp \delta}{\phi}
\end{eqnarray*}
 For angelic tests, suppose that $\restr{\nam{a}}{\Scross(\psi?,\phi)} \subseteq X $. Since $\Scross(\psi?,\phi) = \Scross(\psi) \cup \Scross(\phi)$ we get $\restr{\nam{a}}{\Scross(\psi)} \subseteq X $. Using  the induction hypothesis of (1) on the formula $\psi$, we now get:
\begin{eqnarray*}
\startrg{\nam{a}}{\psi?}{\phi} & = & \startr{\nam{a}}{\psi}? \\
& = & \startr{\restr{\nam{a}}{X}}{\psi}? \\
& = & \startrg{\restr{\nam{a}}{X}}{\psi?}{\phi}
\end{eqnarray*}
The reasoning for $!$ is the same. 

The induction step for $*$ is trivial since  $\startrg{\nam{a}}{\gamma^*}{\phi}  = \gamma^*$ for any annotation $\nam{a}$. 
Finally, the induction step for $\cross$ is handled as follows: suppose that $\restr{\nam{a}}{\Scross(\gamma^\cross,\phi)} \subseteq X $. Let $\nam{x}_1,\ldots,\nam{x}_n$ be all the names in $\nam{a}$ for the fixpoint $\diam{\gamma^\cross}\phi$. Since $\Scross(\gamma^\cross,\phi) = \{\diam{\gamma^\cross} \phi\}$, this means that $X$ contains $\nam{x}_1,\ldots,\nam{x}_n$, and so these are also all the names for $\diam{\gamma^\cross}\phi$ in $\restr{\nam{a}}{X}$. From this it is immediate from the definition of $\startrg{-}{-}{-}$
that:
\[ \startrg{\nam{a}}{\gamma^\cross}{\phi} = \startrg{\restr{\nam{a}}{X}}{\gamma^\cross}{\phi}
\]
as required.
\end{proof}

\begin{lemma} \label{lem:Scross-subterm}
For all $\delta,\gamma \in \gameNorm$ and $\phi, \psi \in \formNorm$:
\begin{enumerate}
\item If $\diam{\delta^\cross}\psi \in \Scross(\phi)$ then $\delta^\cross \ssubterm \phi$.
\item If $\diam{\delta^\cross}\psi \in \Scross(\gamma,\phi)$ then $\delta^\cross \subterm \gamma$.
\end{enumerate}
\end{lemma}
\begin{proof}
By a mutual induction on $\phi$ and $\gamma$
The easy argument is left to the reader.
\end{proof}

\begin{lemma} \label{lem:bulltr-equal}
  For all annotations $\nam{a} \in \names^*$,
  $\gamma \in \gameNorm$, $\phi \in \formNorm$ and $\circ \in \{ \cross, *\}$,
  if $\nam{a} \leqfp \diam{\gamma^\circ}\phi$ then $\startrg{\nam{a}}{\gamma}{\diam{\gamma^\circ}\phi} = \gamma$.
\end{lemma}
\begin{proof}
By Lemma~\ref{lem:restriction} (taking $X = \Scross(\gamma,\diam{\gamma^\circ}\phi)$), we have that
$\startrg{\nam{a}}{\gamma}{\diam{\gamma^\circ}\phi} = \startrg{\restr{\nam{a}}{\Scross(\gamma,\diam{\gamma^\circ}\phi)}}{\gamma}{\diam{\gamma^\circ}\phi}$.
Hence, it suffices to show that
$\restr{\nam{a}}{\Scross(\gamma,\diam{\gamma^\circ}\phi)} = \eps$, because
$\startrg{\eps}{\gamma}{\diam{\gamma^\circ}\phi} = \gamma$.

So let $\nam{x} \in \names[\diam{\delta^\cross}\psi]$ be a name occurring in $\nam{a}$.  
By the assumption that $\nam{a} \leqfp \diam{\gamma^\circ} \varphi$,
it follows that $\diam{\delta^\cross} \psi \leqfp \diam{\gamma^\circ} \varphi$,
i.e., $\gamma^\circ \subterm \delta^\cross$
and hence $\gamma \ssubterm \delta^\cross$,
so it is \emph{not} the case that $\delta^\cross \subterm \gamma$.
By Lemma~\ref{lem:Scross-subterm}(2) this entails that
$\diam{\delta^\cross} \psi \notin \Scross(\gamma,\diam{\gamma^\circ}\phi)$.
We have therefore shown that
$\restr{\nam{a}}{\Scross(\gamma,\diam{\gamma^\circ}\phi)} = \eps$,
which concludes the proof.
\end{proof}


We denote by $C(\phi)$, the number of occurrences of the
demonic iteration symbol $^\cross$ in the formula $\phi \in \langNorm$.
A precise, inductive definition is left to the reader.

\begin{lemma}
\label{lem:isnotins}
For all $\phi,\psi \in \formNorm$ and $\gamma \in \gameNorm$:
\begin{enumerate}
\item If $\psi \in \Scross(\phi)$ then $C(\psi) \leq C(\phi)$. 
\item If $\psi \in \Scross(\gamma,\phi)$ then $C(\psi) \leq C(\gamma) + C(\phi)$ \end{enumerate}
\end{lemma}
\begin{proof}
The two items can be proved by a straightforward mutual induction on
the complexity of game terms $\gamma$ and formulas $\phi$. 
\end{proof}

\begin{lemma}
\label{lem:abequal}
 Consider a sequence $\nam{a} = \nam{b} \nam{x}_1 \dots \nam{x}_n$ where $\nam{x}_1,\dots,\nam{x}_n$
are all the names of a fixpoint $\diam{\gamma^\cross} \phi$.
Then we have that $\startr{b}{\phi} = \startr{\nam{a}}{\phi}$
\end{lemma}
\begin{proof}
First, we observe that
$\startr{\nam{a}}{\phi} = \startr{\restr{\nam{a}}{\Scross(\phi)}}{\phi}$.
This is an instance of 
Lemma~\ref{lem:restriction}, if we take $X$ to be the set of all names
associated with a fixpoint in $\Scross(\phi)$.
It therefore suffices to prove that $\nam{x}_1,\ldots,\nam{x}_n$ do not appear in  $\restr{\nam{a}}{\Scross(\phi)}$.
Since $\nam{x}_1,\ldots,\nam{x}_n$ are all the names for $\diam{\gamma^\cross} \phi$, it suffices to prove $\diam{\gamma^\cross}\phi \notin \Scross(\phi)$.
By Lemma~\ref{lem:isnotins}(1), $\diam{\gamma^\cross}\phi \in \Scross(\phi)$
would imply that $C(\diam{\gamma^\cross}\phi) \leq C(\phi)$, which is clearly
impossible.
\end{proof}

We are now ready to prove Lemma~\ref{lem:bulltr-of-cross}.


\begin{proofof}{Lemma~\ref{lem:bulltr-of-cross}}
For item \eqref{bulltr of cross}, we apply the definition of the bullet translation to obtain: 
\begin{eqnarray*}
\startr{\nam{a}}{(\diam{\gamma^\cross} \phi)} 
  & = &\dd{\startrg{\nam{a}}{\gamma^\cross}{\phi}} \startr{\nam{a}}{\phi}
\\ & = & \diam{(\ulag{\chi}{\gamma})^\cross}\dd{\ula{\chi}}\startr{\nam{a}}{\phi}
\end{eqnarray*}
The result now follows since $\startr{\nam{a}}{\phi} = \startr{b}{\phi}$, by 
Lemma~\ref{lem:abequal}.

For item \eqref{bulltr of unfolded cross},
we apply the bullet translation again to get: 
\begin{eqnarray*}
&& \startr{\nam{a}}{(\phi \land \diam{\gamma} \diam{\gamma^\cross} \phi)} 
  \\
  &  = &  \startr{\nam{a}}{\phi} \land 
  \dd{\startrg{\nam{a}}{\gamma}{\diam{\gamma^{\cross}}\phi}}
  \dd{\startrg{\nam{a}}{\gamma^\cross}{\phi}} 
  \startr{\nam{a}}{\phi}
\\ & = &  \startr{\nam{a}}{\phi} \land
   \dd{\startrg{\nam{a}}{\gamma}{\diam{\gamma^{\cross}}\phi}}
   \diam{(\ulag{\chi}{\gamma})^\cross}
   \dd{\ula{\chi}}\startr{\nam{a}}{\phi}
\\ & = &  \startr{b}{\phi} \land
   \dd{\startrg{\nam{a}}{\gamma}{\diam{\gamma^{\cross}}\phi}}
   \diam{(\ulag{\chi}{\gamma})^\cross}
   \dd{\ula{\chi}}\startr{b}{\phi}
\end{eqnarray*}
where for the last step we used Lemma~\ref{lem:abequal} again. 
By Lemma~\ref{lem:bulltr-equal}, we get  
$\startrg{\nam{a}}{\gamma}{\diam{\gamma^{\cross}}\phi} = \gamma$,
and so we are done. 
\end{proofof}


We also need the following simpler analogue of
Lemma~\ref{lem:bulltr-of-cross} for least fixpoints.

\begin{lemma} \label{lem:bulltr-of-star} \label{l:bullet of star}
For a least fixpoint formula $\diam{\gamma^*}\phi \in F^*$ and an annotation
$\nam{a} \leqfp \diam{\gamma^*} \phi$, we have
\begin{equation}
\label{bullet of unfolded star}
\startr{\nam{a}}{(\phi \lor \diam{\gamma} \diam{\gamma^*} \phi)}
=
\startr{\nam{a}}{\phi} \lor \diam{\gamma}
\diam{\gamma^*} \startr{\nam{a}}{\phi}.
\end{equation}
\end{lemma}
\begin{proof}
We use the definition of the bullet translation to compute as follows: 
\begin{eqnarray*}
 \startr{\nam{a}}{(\phi \lor \diam{\gamma} \diam{\gamma^*} \phi)} 
  & = & \startr{\nam{a}}{\phi} \lor
\diam{\startrg{\nam{a}}{\gamma}{\diam{\gamma^{*}}\phi}}
\diam{\startrg{\nam{a}}{\gamma^*}{\phi}} \startr{\nam{a}}{\phi}
\\ & = & \startr{\nam{a}}{\phi} \lor
\diam{\startrg{\nam{a}}{\gamma}{\diam{\gamma^{*}}\phi}}
\diam{\gamma^*} \startr{\nam{a}}{\phi}
\end{eqnarray*}
By Lemma~\ref{lem:bulltr-equal} we get
$\startrg{\nam{a}}{\gamma}{\diam{\gamma^{*}}\phi} = \gamma$, and so we
are done. 
\end{proof}

\subsubsection{Missing cases in Theorem~\ref{thm:CloG-to-KozenG}}

We now give more details for the translation of some of the rules
from \CloG to \GLms in the proof of Theorem~\ref{thm:CloG-to-KozenG}.

For the weakening rule 
this is trivial. The case
of \AxOne: $\Phi= p^\eps,\nnf{p}^\eps$ is also immediate  
as $\startr{}{\Phi} = p,\nnf{p}$ which is an instance of \Ax. 
Equally straightforward to translate are the \CloG rules dealing with
the Boolean and the (basic) modal operators. Here one simply has to observe that 
these connectives commute with the bullet translation. 
Concerning the rules for tests and for angelic and demonic choice we consider only the 
demonic choice operator in detail and leave the other similar cases to the reader. 
Suppose we derive $\Phi, (\diam{\gamma \gchd \delta} \phi)^{\nam{a}}$ 
within the \CloG proof $\pi$ via an application of the \RuChoiceD-rule.
The translation of the assumption of the rule is 
$\startr{}{\Phi}, \startr{\nam{a}}{(\diam{\gamma} \phi)} \land \startr{\nam{a}}{(\diam{\delta} \phi)}$.
Spelling out the details of the bullet translation, this can be rewritten as
$\startr{}{\Phi}, \dd{\startrg{\nam{a}}{\gamma}{\varphi}} \startr{\nam{a}}{\phi} \land  \dd{\startrg{\nam{a}}{\delta}{\varphi}} \startr{\nam{a}}{\phi}$.
From here we obtain the following \GLms derivation steps:
\begin{prooftree}
 \AxiomC{$\startr{}{\Phi}, \dd{\startrg{\nam{a}}{\gamma}{\varphi}} \startr{\nam{a}}{\phi} \land  \dd{\startrg{\nam{a}}{\delta}{\varphi}} \startr{\nam{a}}{\phi}$}
 \doubleLine
 \RightLabel{(*)}
 \UnaryInfC{$\startr{}{\Phi}, \diam{\startrg{\nam{a}}{\gamma}{\varphi}} \startr{\nam{a}}{\phi} \land  \diam{\startrg{\nam{a}}{\delta}{\varphi}} \startr{\nam{a}}{\phi}$}
 \RightLabel{\RuChoiceD}
  \UnaryInfC{$\startr{}{\Phi}, \diam{\startrg{\nam{a}}{\gamma \gchd \delta}{\varphi}} \startr{\nam{a}}{\phi}$}
\end{prooftree}
where (*) possibly involves  applying the \RuCompD-rule multiple times. It is easy to see that the conclusion of the \GLms derivation
is equal to $\startr{}{\Phi}, \startr{\nam{a}}{ (\diam{\gamma \gchd \delta} \phi)}$ as required.

To see how to deal with the \RuStar-rule consider an application of the
rule in $\pi$ of the form 
\begin{prooftree}
  \AxiomC{$\Phi, (\phi  \lor \diam{\gamma} \diam{\gamma^*} \phi)^{\nam{a}}$}
  \RightLabel{\RuStar}
  \UnaryInfC{$\Phi, ( \diam{\gamma^*} \phi)^{\nam{a}}$}
\end{prooftree}
The premise of this rule translates to $\startr{}{\Phi}, \startr{\nam{a}}{\phi} \vee \diam{\gamma} \diam{\gamma^*}  \startr{\nam{a}}{\phi}$ - this can 
be seen using Lemma~\ref{l:bullet of star} and the side condition of the \RuStar-rule. An application of the \RuStar-rule in \GLms yields 
$\startr{}{\Phi}, \diam{\gamma^*}\startr{\nam{a}}{\phi}$ which in turn equals $\startr{}{\Phi}, \startr{\nam{a}}{(\diam{\gamma^*}\phi)}$ as required.

For the \RuCross rule consider a rule application
\begin{prooftree}
  \AxiomC{$\Phi, (\phi  \land \diam{\gamma} \diam{\gamma^\cross} \phi)^{\nam{a}}$}
  \RightLabel{\RuCross}
  \UnaryInfC{$\Phi, ( \diam{\gamma^\cross} \phi)^{\nam{a}}$}
\end{prooftree}
   By the side condition of the \RuCross-rule we
   can assume that $\nam{a}$ is of the form $\nam{b}  \nam{x}_1 \dots \nam{x}_n$ where $\nam{x}_1 \dots \nam{x}_n$
   are all the names of $\diam{\gamma^\cross}\phi$ occurring in $\nam{a}$.
   Applying Lemma~\ref{lem:startr-of-cross} we get
   that the premise of the rule translates  to $\startr{}{\Phi}, \startr{\nam{b}}{\phi} \land \dd{\gamma}
\diam{(\ulag{\chi}{\gamma})^\cross}
\dd{\ula{\chi}}\startr{\nam{b}}{\phi}$ where $\vec{\chi}$ consists of the context sequents as in previous cases.
   Consider now the following \GLms derivation steps:
   \begin{prooftree}
\AxiomC{$\startr{}{\Phi}, 
  \startr{\nam{b}}{\varphi} \land 
  \dd{\gamma}
  \diam{(\ulag{\chi}{\gamma})^\cross}
  \dd{\ula{\chi}}\startr{\nam{b}}{\varphi}
  $}
 \RightLabel{\RuCompD}
\UnaryInfC{$\startr{}{\Phi}, 
  \startr{\nam{b}}{\varphi} \land
  \diam{\gamma} 
  \diam{(\ulag{\chi}{\gamma})^\cross} 
  \dd{\ula{\chi}}\startr{\nam{b}}{\varphi}
  $}
  \doubleLine
 \RightLabel{\RuMonDgame, \RuMonDfmla}
\UnaryInfC{$\startr{}{\Phi}, 
  \dd{\ula{\chi}}\startr{\nam{b}}{\varphi} \land 
  \diam{\ulag{\chi}{\gamma}}
  \diam{(\ulag{\chi}{\gamma})^\cross} 
  \dd{\ula{\chi}}\startr{\nam{b}}{\varphi}
  $}
\RightLabel{\RuCross}
  \UnaryInfC{$\startr{}{\Phi}, 
  \diam{(\ulag{\chi}{\gamma})^\cross}\dd{\ula{\chi}}\startr{\nam{b}}{\phi},
  $}
\end{prooftree}
where the double line indicates multiple applications of the deep monotonicity
rules. Observe now that the conclusion $\startr{}{\Phi}, 
  \diam{(\ulag{\chi}{\gamma})^\cross}\dd{\ula{\chi}}\startr{\nam{b}}{\phi},
  $ is by Lemma~\ref{lem:startr-of-cross} equal to $\startr{}{\Phi}, \startr{\nam{a}}{(\diam{\gamma^\cross} \phi)}$
  as required.

\subsection{Omitted proofs of Section \ref{sec:mon-mu}}
This part of the appendix contains definitions and lemmas
that lead up to a detailed proof of Theorem~\ref{thm:CloM-complete}.


First, we translate monotone $\mu$-calculus into normal $\mu$-calculus 
by extending the translation from \cite{KrachtWolter99,Helle03}
of monotonic modal logic into normal bimodal logic.
More precisely, we define the language $\langBimod$ to be the set of
modal $\mu$-calculus formulas over the set of labels
$L = \{g_N \mid g \in \AtGame\} \cup \{g_\ni \mid g \in \AtGame\}$
defined in the usual way,
and interpret $\langBimod$
over Kripke models that have an accessibility relation for each label in $L$.
We define the translation $(-)^t \colon \langM \to \langBimod$
follows:
\begin{itemize}
\item $p^t=p^t$ and $(\lnot p)^t = \lnot p$.
\item For Boolean connectives: compositionally. 
\item $(\mu x. A)^t = \mu x. (A^t)$ and $(\nu x. A)^t = \nu x. (A^t)$,
\item $(\diam{g}A)^t = \diam{g_{N}}[g_{\ni}] (A^t)$,
\item $(\diam{g^d}A)^t = [g_{N}]\diam{g_{\ni}}(A^t)$.
\end{itemize} 

\begin{proofof}{Lemma~\ref{lem:t-sat-val}}
Item 1)
Given a game model $\bbS = (S,E,V)$
we construct a Kripke model $\bbS^t$ for $\langBimod$
with state space $S \cup \wp(S)$ and valuation $V'(p) = V(p)$,
by taking \\[.3em]
\begin{tabular}{llcl}
\quad & $R_{g_N} = \{ (s,U) \in S \times \wp(S) \mid U \in E_g(s)\}$,\\
& $R_{g_\ni} = \{ (U,s) \in \wp(S) \times S \mid s \in U\}$.\\
\end{tabular}\\[.3em]
It is straightforward to show by induction that for all $s \in S$ and all $C \in \langM$: $s \in \mng{C}^\bbS$ iff $s \in \mng{C^t}^\bbS$.
It follows that $(-)^t$ preserves satisfiability.

Item 2)
Given a Kripke  $\langBimod$- model $\mathcal{K}$
with state space $W$ and valuation $V$,
we can construct a game model $\mathcal{K}_m$ with the same state space and valuation, and  by defining for $g \in \AtGame$ an effectivity
function $E_g$ by
\[
 E_g(Z) = \{w \in W \mid \exists v \in W (w R_{g_N} v \mbox{ and } R_{g_\ni}[v] \subseteq Z)\}.
\]
It is again routine to show that for all Kripke models $\mathcal{K}$ and
for all $C \in \langM$, we have:
$\mng{C}^{\mathcal{K}_m} = \mng{C^t}^{\mathcal{K}}$.

From this it follows that if $C$ is valid over game models (for $\langM$)
then $C^t$ is valid over Kripke models (for $\langBimod)$.\qed
\end{proofof}


We will now show how to transform derivations from
the system $\Clo$ \cite{AfshariLeigh:LICS17} into $\CloM$
using the translation $(-)^t$.

\begin{lemma}\label{lem:from-Clo-to-CloM}
For all $C \in \langM$, if $\Clo \vdash C^t$ then $\CloM[C] \vdash C$. 
\end{lemma}
\begin{proof}
Given a $\Clo$-proof $\pi$ for $C^t$, we shall construct step by step a
tree rooted at a node labelled $C$ in which every edge is labelled by a
proof rule of $\CloM[C]$, every node is labelled by a sequent or an
expression $[\Gamma]^{\mathsf{x}}$ marking the sequent $\Gamma$ as a
discharged assumption. We also construct a map $h$ sending each node in
the tree to some node in the proof tree $\pi$, such that the label of
$h(u)$ is $\Gamma^t$ if the label of $u$ is $\Gamma$. Here, we are
extending the translation $(-)^t$ to annotated sequents in the obvious
way. Rather than performing an induction on the size of proofs, the
construction will simply proceed from the root up, and will be carried
out in such a way that the end result will clearly be a $\CloM[C]$-proof.

We start by creating a root node labelled with $C^\varepsilon$ and letting $h$ map this node to the root of the proof tree $\pi$. Whenever we create a node in the construction that is mapped via $h$ to a discharged assumption $[\Gamma^t]^\nam{x}$, we make sure to label that node by the discharged assumption $[\Gamma]^\nam{x}$. Note that we can always assume that for each $A$ in the label of $h(l)$, there is at most one $B$ in the label of $l$ with $B^t = A$. (Otherwise, we just apply the weakening rule.) 

Now, given that we have constructed the tree up to some point, but we do
not yet have a proper $\CloM[C]$-proof, we pick a leaf $l$ of the tree that is not a discharged assumption and not an axiom of $\CloM[C]$, and continue the construction by a case distinction. Most of the cases are handled in a trivial manner, and we only give two examples of the easy cases: if the leaf $l$ is mapped to a node $h(l)$ labelled $\Gamma,\nu x.A^{\nam{a}}$,  and is the conclusion to an instance of the $\RuNuClo_{\nam{x}}$-rule with premise  $\Gamma,A(\nu x.A)^{\nam{a}}$, then there must be some $\Psi,B$ such that $\Psi^t = \Gamma$, $B = A$ and $l$ is labelled $\Psi,\nu x. B$. We add a new child of $l$ which we label $\Psi,B(\nu x. B)^{\nam{ax}}$. We let $h$ map this child to the premise of $h(l)$, and label the edge by the rule $\RuNuClo_{\nam{x}}$. For another example, if  the leaf $l$ is mapped to $h(l)$ labelled $\Gamma, A \wedge B$, and is the conclusion to an instance of the $\RuAnd$-rule with premises  $\Gamma,A$ and $\Gamma,B$, then $l$ must be labelled $\Psi,C \wedge D$ where $\Psi^t = \Gamma$, $C^t = A$ and $D^t = B$.  We add two new children of $l$ which we label $\Psi,C$ and $\Psi,D$ respectively. We extend the map $h$ by sending each child to its corresponding premise, and labelling the edges by the $\RuAnd$-rule. 

The only non-trivial case is when the node $h(l)$ is a conclusion to an instance of the modal rule. This is only possible if $h(l)$ has a label of the form $\diam{g_{N}} [g_\ni] A_1 ,..., \diam{g_{N}} [g_\ni] A_n, [g_{N}]\diam{g_{\ni}} B$, where $l$ is labelled with the sequent $\diam{g} C_1,...,\diam{g} C_n, \diam{g^d} D$ and $C_1^t = A_1,...,C_n^t = A_n$, $D^t = B$. By inspection of the rules of $\Clo$, and assuming that $n\geq 2$ (since the other case is easier), we can see that the subtree of the  proof $\pi$ rooted at the node $h(l)$ must have the following shape, since these are the only rules that are applicable:


\begin{prooftree}
 \AxiomC{}
 \noLine
 \UnaryInfC{$\vdots$}
 \noLine
\UnaryInfC{$A_i,B$}
\RightLabel{$\mathsf{mod}$}
\UnaryInfC{$[g_\ni] A_i,\diam{g_{\ni}} B$}
\RightLabel{$\RuWeak$}
 \UnaryInfC{$ [g_\ni] A_1 ,...,  [g_\ni] A_n,\diam{g_{\ni}} B$}
 \RightLabel{$\mathsf{mod}$}
 \UnaryInfC{$\diam{g_{N}} [g_\ni] A_1 ,..., \diam{g_{N}} [g_\ni] A_n, [g_{N}]\diam{g_{\ni}} B$}
\end{prooftree}
for some $i \in \{1,...,n\}$. We thus continue the construction as follows: 


\begin{prooftree}
\AxiomC{$C_i,D$}
\RightLabel{$\mathsf{mod}$}
 \UnaryInfC{$\diam{g} C_i, \diam{g^d} D$}
 \RightLabel{$\RuWeak$}
 \UnaryInfC{$\diam{g} C_1,...,\diam{g} C_n, \diam{g^d} D$}
\end{prooftree}
Finally, we extend the map $h$ by sending the new leaf labelled $C_i,D$ to the node in the previous figure labelled $A_i,B$. 
\end{proof}

\begin{lemma}\label{lem:from-CloM-to-Clo}
For all $C \in \langM$, if $\CloM \vdash C$ then $\Clo \vdash C^t$. 
\end{lemma}
\begin{proof}
The proof is very similar to the proof of Lemma~\ref{lem:from-Clo-to-CloM} above. 
In this case --again starting from the root-- we turn a $\CloM_C$-proof $\pi$ of a given formula $C \in \langM$ into a 
corresponding $\Clo$-proof of $C^t$. Unlike in the proof of Lemma~\ref{lem:from-Clo-to-CloM} we only sketch
the construction. First we note that all cases of non-modal rules are even easier than in the proof of
Lemma ~\ref{lem:from-Clo-to-CloM}
as the non-modal $\CloM_C$-rules are instances of $\Clo$-rules.

For the modal case suppose that we have constructed a partial $\Clo$-proof of $C^t$ with
a leaf labelled with $(\diam{g} D_1)^t, (\diam{g^d} D_2)^t$ and that the corresponding node in 
the $\CloM[C]$-proof $\pi$ labelled
with $\diam{g} D_1, \diam{g^d} D_2$ has been derived via an application of the modal rule:
\begin{prooftree} \AxiomC{}
 \noLine
 \UnaryInfC{$\vdots$}
 \noLine
\UnaryInfC{$D_1,D_2$}
\RightLabel{$\mathsf{mod}$}
 \UnaryInfC{$\diam{g} D_1, \diam{g^d} D_2$}
\end{prooftree}
By definition we have $(\diam{g} D_1)^t = \diam{g_{N}} [g_\ni] D_1^t$ and $(\diam{g^d} D_2)^t = [g_{N}] \diam{g_\ni} D_2^t$.
Therefore we can extend the $\Clo$-proof of $C^t$ as follows:
\begin{prooftree}
 \AxiomC{$D_1^t,D_2^t$}
 \RightLabel{$\mathsf{mod}$}
 \UnaryInfC{$[g_\ni] D_1^t, \diam{g_\ni} D_2^t$}
 \RightLabel{$\mathsf{mod}$}
 \UnaryInfC{$\diam{g_{N}} [g_\ni] D_1^t,[g_{N}] \diam{g_\ni} D_2^t$}
\end{prooftree}
This finishes the proof sketch.
\end{proof}

\begin{proofof}{Theorem~\ref{thm:CloM-complete}}
\emph{Soundness:}
Let $C \in \langM$,
such that ${\CloM} \vdash C$.
Then by Lemma~\ref{lem:from-CloM-to-Clo},
$\Clo \vdash C^t$ and hence by the soundness of $\Clo$ 
it follows that $C^t$ is valid on all Kripke models for $\langBimod$.
Suppose now that $C$ would not be valid in all game models,
i.e., there is a game model $\mathcal{M}$ and a state $w$ in $\mathcal{M}$
such that $\mathcal{M}, w \vDash \nnf{C}$.
Since $(-)^t$ preserves satisfiability (Lemma~\ref{lem:t-sat-val}(1)),
it follows that $\nnf{C}^t$ is satisfiable in a  Kripke model for $\langBimod$.
Finally, since $\nnf{C}^t$ is equivalent with $\nnf{C^t}$ 
this would imply that $C^t$ is not valid, a contradiction.

\emph{Completeness:}
Assume that $C$ is valid over game models. 
Then by Lemma~\ref{lem:t-sat-val}(2), $C^t$ is valid over Kripke models for $\langBimod$.
From the completeness of $\Clo$ \cite{AfshariLeigh:LICS17}, it follows
that $\Clo \vdash C^t$, and hence by Lemma~\ref{lem:from-Clo-to-CloM} that $\CloM \vdash C$.
\qed
\end{proofof}

\subsection{Omitted proofs of Section \ref{sec:complete}}
\subsubsection{Proof of Proposition~\ref{p:translation}}

This part of the appendix contains a detailed proof for
Proposition~\ref{p:translation}, which summarises the most important properties 
of our translation from game logic into the monotone $\mu$-calculus.

We address the claims in the proposition with separate lemmas. The first
of this lemmas entails that $x^\varphi \in \Var(\tr{\xi})$ for all
$\varphi \in F(\xi)$:
\begin{lemma}\label{lem:corr-fixpt-to-vars-items}
  \hfill
  \begin{enumerate}
  \item For all game logic formulas $\phi, \psi$: If $\psi\in F(\phi)$ then $\var{x}^\psi \in \Var(\tr{\phi})$.\label{itm:phi}
  \item For all game logic formulas $\chi, \psi$ and games $\gamma$: If $\psi \in F(\gamma,\chi)$ then $\var{x}^\psi \in \Var(\itrg{\chi}{\gamma}(\tr{\chi}))$.\label{itm:gamma}
  \end{enumerate}
\end{lemma}

\begin{proof}
Both items are are proven with a mutual induction over the complexity of
the formula $\phi$ and of the game term $\gamma$.

  \emph{Base case (\ref{itm:phi}):}
  If $\phi$ is of the form $p$ or $\lnot p$, then $F(\phi)=\emptyset$, hence (\ref{itm:phi}) holds trivially.
  \emph{Base case (\ref{itm:gamma}):}
  Similarly, if $\gamma$ is is of the form $g$ or $g^d$, then $F(\gamma,\phi)=\emptyset$, hence (\ref{itm:gamma}) holds trivially.

  \emph{Induction hypotheses:}
  Assume (\ref{itm:phi}) holds for all formulas that are proper subterms of $\phi$ or $\gamma$.
  Assume (\ref{itm:gamma}) holds for all games that are proper subterms of $\phi$ or $\gamma$.

  \emph{Induction step (\ref{itm:phi}):}
  Suppose $\phi = \phi_1 \lor \phi_2$. We then have that $F(\phi) = F(\phi_1) \cup F(\phi_2)$.
  Let $i \in \{1,2\}$ and assume $\psi \in F(\phi_i)$ then by IH for (\ref{itm:phi}), we have $\var{x}^\psi \in \Var(\tr{\phi_i})$ and hence $\var{x}^\psi \in \Var(\tr{\phi}) = \Var(\tr{\phi_1})\cup \Var(\tr{\phi_2})$.

  The argument for $\phi=\phi_1\land\phi_2$ is similar.

  Suppose now $\phi = \diam{\delta}\phi_0$. We then have $F(\phi) = F(\delta,\phi_0) \cup F(\phi_0)$.
  For all $\psi \in F(\phi_0)$, we have by the IH for (\ref{itm:phi}) that $\var{x}^\psi \in \Var(\tr{\phi_0})$,
  and since $\tr{\phi_0}$ is a subterm of $\tr{(\diam{\delta}\phi_0)} = \itrg{\phi_0}{\delta}(\tr{\phi_0})$,
  it follows that $\Var(\tr{\phi_0}) \subseteq \Var(\tr{\phi})$.
  For all $\psi \in F(\delta,\phi_0)$, we have by the IH for (\ref{itm:gamma}) that
  $\var{x}^\psi \in \Var(\itrg{\phi_0}{\delta}(\tr{\phi_0})) = \Var(\tr{\phi})$.

  \emph{Induction step (\ref{itm:gamma}):}
  Let $\chi$ be an arbitrary game logic formula.
  Suppose that $\gamma = \gamma_1 \gcha \gamma_2$.
  We then have $F(\gamma,\chi) = F(\gamma_1,\chi) \cup F(\gamma_2,\chi)$.
  Let $i \in \{1,2\}$ and assume that $\psi \in F(\gamma_i,\chi)$.
  Then by IH for (\ref{itm:gamma}), we have that
  $\var{x}^\psi \in \Var(\itrg{\chi}{\gamma_i}(\tr{\chi})) \subseteq \Var(\itrg{\chi}{\gamma_1}(\tr{\chi}) \lor \itrg{\chi}{\gamma_2}(\tr{\chi})) = \Var(\itrg{\chi}{\gamma}(\tr{\chi}))$.
  The case for $\gamma = \gamma_1 \gchd \gamma_2$ is similar.

  Suppose $\gamma = \gamma_1 \comp \gamma_2$.
  We then have  $F(\gamma,\chi) = F(\gamma_1,\diam{\gamma_2}\chi) \cup F(\gamma_2,\chi)$.
  If $\psi \in F(\gamma_1,\diam{\gamma_2}\chi)$, then by IH for (\ref{itm:gamma}), we have that
  $\var{x}^\psi \in \Var(\itrg{\diam{\gamma_2}\chi}{\gamma_1}(\tr{(\diam{\gamma_2}\chi)})) = \Var(\itrg{\chi}{\gamma_1;\gamma_2}(\tr{\chi}))$.
  If $\psi \in F(\gamma_2,\chi)$, then by IH for (\ref{itm:gamma}), we have that
  $\var{x}^\psi \in \Var(\itrg{\chi}{\gamma_2}(\tr{\chi})) \subseteq \Var(\itrg{\diam{\gamma_2}\chi}{\gamma_1}(\itrg{\chi}{\gamma_2}(\tr{\chi}))) = \Var(\itrg{\chi}{\gamma_1;\gamma_2}(\tr{\chi}))$.

  Suppose $\gamma = \delta^*$.
  We then have $F(\gamma,\chi) = \{\diam{\delta^*}\chi\} \cup F(\delta,\diam{\delta^*}\chi)$.
  If $\psi = \diam{\delta^*}\chi$, we use that
  $\itrg{\chi}{\delta^*}(\tr{\chi}) = \mu x^{\diam{\delta^*}\chi} . \tr{\chi} \lor \itrg{\diam{\delta^*}\chi}{\delta}(x^{\diam{\delta^*}\chi})$ and hence $x^\psi = x^{\diam{\delta^*}\chi} \in \Var(\itrg{\chi}{\delta^*}(\tr{\chi}))$.
  If $\psi \in F(\delta,\diam{\delta^*}\chi)$, then by  IH for (\ref{itm:gamma}), we have that
  $\var{x}^\psi \in \Var(\itrg{\diam{\delta^*}\chi}{\delta}(\tr{(\diam{\delta^*}\chi)}))$.
  It therefore suffices to show that
  $\Var(\itrg{\diam{\delta^*}\chi}{\delta}(\tr{(\diam{\delta^*}\chi)})) \subseteq \Var(\itrg{\chi}{\delta^*}(\tr{\chi}))=\Var(\mu x^{\diam{\delta^*}\chi} . \tr{\chi} \lor \itrg{\diam{\delta^*}\chi}{\delta}(x^{\diam{\delta^*}\chi}))$.
  Since the context $\itrg{\diam{\delta^*}\chi}{\delta}(-)$ occurs on both sides of the inclusion,
  it suffices to show that
  $\Var(\tr{(\diam{\delta^*}\chi)}) \subseteq \Var(\itrg{\chi}{\delta^*}(\tr{\chi}))$,
  and this holds since $\tr{(\diam{\delta^*}\chi)} = \itrg{\chi}{\delta^*}(\tr{\chi})$.

    Suppose $\gamma = \phi ?$. By IH for (\ref{itm:phi}) it follows that $x^\psi \in \Var(\tr{\phi})$.
  Since $\itrg{\chi}{\phi ?}(\tr{\chi}) = \tr{\phi} \lor \tr{\chi}$, we have that
  $x^\psi \in \Var(\itrg{\chi}{\phi ?}(\tr{\chi}))$.
  
  The case for $\gamma = \phi !$ is similar.
\end{proof}

To show that $\tr{\xi}$ is locally well-named and that $x^\varphi
\leqAL_{\tr{\xi}} x^\psi$ implies $\varphi \leqfp \psi$ we prove a
series of lemmas leading to the central property from
Lemma~\ref{l:compatible orders} below.

\begin{lemma} \label{l:no free in tau}
For all game logic formulas $\phi$, all games $\gamma$ and all
$\mu$-calculus formulas $A$, we have that
$\FVar(\itrg{\varphi}{\gamma}(A)) \subseteq \FVar(A)$, i.e., all free
variables in $\itrg{\varphi}{\gamma}(A)$, occur free in $A$.
\end{lemma}
\begin{proof}
 The proof is a simple induction on the complexity of $\gamma$ in the
recursive clauses defining $\itrg{\varphi}{\gamma}(A)$.
\end{proof}


\begin{definition}
A \emph{call triple} is a triple $(\gamma,\phi,A)$ where $\gamma$ is a
game term, $\phi$ is a game logic formula and $A$ is a formula of the
monotone $\mu$-calculus. The set of call triples is denoted by
$\mathsf{CTr}$.
\end{definition}

\begin{definition}
Given a game logic formula $\xi$, the set of \emph{recursive calls of $\tau$ on $\xi$}, denoted $C_\xi$, is the least fixpoint  of  the monotone map $c_\xi : \mathsf{P}(\mathsf{CTr}) \to \mathsf{P}(\mathsf{CTr})$ defined by setting $ t \in c_\xi(X)$, for $X \subseteq \mathsf{CTr}$, iff: 
\begin{itemize}
\item $t = (\gamma,\varphi,\tr{\phi})$ for some subformula $\diam{\gamma}\varphi$ of $\xi$, or
\item $t$ is of the form $(\gamma,\phi,A)$ or $(\gamma',\phi,A)$ where $ (\gamma \gcha \gamma', \phi, A) \in X$,  or
\item $t$ is of the form $(\gamma,\phi,A)$ or $(\gamma',\phi,A)$ where $ (\gamma \gchd \gamma', \phi, A) \in X$, or
\item $t$ is of the form $(\gamma, \diam{\gamma^\cross} \phi,x^{ \diam{\gamma^\cross} \phi})$ and there exists some $A$ for which $(\gamma^\cross,\phi,A) \in X$,
\item $t$ is of the form $(\gamma, \diam{\gamma^*} \phi,x^{ \diam{\gamma^*} \phi})$ and there exists some $A$ for which $(\gamma^*,\phi,A) \in X$,
\item $t$ is of the form $(\gamma',\phi,A)$ or  $(\gamma,\diam{\gamma'}\phi,\itrg{\gamma'
}{\phi}(A))$  for some $\gamma,\gamma',\phi$ and $A$ such that  $((\gamma \comp \gamma'),\phi,A) \in X$.
\end{itemize}
\end{definition}

\begin{lemma} \label{l:variables and recursive calls}
Let $\xi$ be a game logic formula and let $\eta x^{\diam{\delta^\circ}
\varphi} . B$ be a subformula of $\tr{\xi}$. Then  there exists some $A$
such that 
\[
 \eta x^{\diam{\delta^\circ} \varphi} . B = \itrg{\varphi}{\delta^\circ}(A)
\]
and, furthermore, $(\delta^\circ,\varphi,A) \in C_\xi$.
\end{lemma}
\begin{proof}
We show with mutual induction on the complexity on subformula $\psi$ of
$\xi$ and game terms $\gamma$ occurring in $\xi$ that
\begin{enumerate}
 \item If $\eta x^{\diam{\delta^\circ} \varphi} . B$ is a subformula of
$\tr{\psi}$ then there exists some $A$ such that 
$\eta x^{\diam{\delta^\circ} \varphi} . B = \itrg{\delta^\circ}{\varphi}(A)$
and, furthermore, $(\delta^\circ,\varphi,A) \in C_\xi$.
\label{itm:ih for formulas}

 \item If $\eta x^{\diam{\delta^\circ} \varphi} . B$ is a subformula of
$\itrg{\chi}{\gamma}(D)$ and $(\gamma,\chi,D) \in C_\xi$ then it is
either a subformula of $D$ or there exists some $A$ such that $\eta
x^{\diam{\delta^\circ} \varphi} . B = \itrg{\delta^\circ}{\varphi}(A)$
and, furthermore, $(\delta^\circ,\varphi,A) \in C_\xi$.
\label{itm:ih for game terms}
\end{enumerate}
In the inductive argument we distinguish the following cases:
\begin{itemize}
 \item In the base case we have that either $\psi = p$, $\psi = \neg p$,
$\gamma = g$ or $\gamma = g^d$. The inductive claims are trivially
satisfied for all of these cases because $\eta x^{\diam{\delta^\circ}
\psi} . B$ is not a subformula of either $\tr{\psi}$ or of
$\itrg{\chi}{\gamma}$.

 \item If $\psi = \chi_1 \land \chi_2$ or $\psi = \chi_1 \lor \chi_2$
and $\eta x^{\diam{\delta^\circ} \varphi} . B$ is a subformula of
$\tr{\psi}$ then $\eta x^{\diam{\delta^\circ} \varphi} . B$ is already a
subformula of $\tr{\chi_i}$ for some $i \in \{1,2\}$. It follows from
(\ref{itm:ih for formulas}) in the inductive assumption for $\chi_i$
that there is some $A$ such that $\eta x^{\diam{\delta^\circ} \varphi} .
B = \itrg{\delta^\circ}{\varphi}(A)$ and that $(\delta^\circ,\varphi,A)
\in C_\xi$. This is already the inductive claim we needed to show.

 \item Consider the case where $\psi = \diam{\gamma} \chi$ and assume
that $\eta x^{\diam{\delta^\circ} \varphi} . B$ is a subformula of
$\tr{\psi} = \itrg{\chi}{\gamma}(\tr{\chi})$. Because $\diam{\gamma}
\chi$ is a subformula of $\xi$ it follows from the definition of $C_\xi$
that $(\gamma,\chi,\tr{\chi}) \in C_\xi$. We can thus apply (\ref{itm:ih
for game terms}) from inductive hypotheses to $\gamma$ to obtain that
either $\eta x^{\diam{\delta^\circ} \varphi} . B$ is a subformula of
$\tr{\chi}$ or there exists some $A$ such that $\eta
x^{\diam{\delta^\circ} \varphi} . B = \itrg{\delta^\circ}{\varphi}(A)$
and $(\delta^\circ,\varphi,A) \in C_\xi$. In the latter case we are done
and in the former we can apply (\ref{itm:ih for formulas}) from the
inductive hypothesis to the subformula $\chi$ of $\psi$ and obtain the
required properties as well.

 \item For the cases where $\gamma = \gamma_1 \gchd \gamma_2$ or $\gamma
= \gamma_1 \gcha \gamma_2$ assume that $\eta x^{\diam{\delta^\circ}
\varphi} . B$ is a subformula of $\itrg{\chi}{\gamma}(D)$ and that
$(\gamma,\chi,D) \in C_\xi$. From the former it follows then that $\eta
x^{\diam{\delta^\circ} \varphi} . B$ is already a subformula of
$\itrg{\chi}{\gamma_i}(D)$ for some $i \in \{1,2\}$. From the latter it
follows with the closure conditions of $C_\xi$ that $(\gamma_i,\chi,D)
\in C_\xi$. We can thus immediately apply (\ref{itm:ih for game terms})
of the inductive hypothesis for the respective $\gamma_i$ to obtain the
required property.

 \item The interesting case is where $\gamma = \rho^\circ$ with $\circ
\in \{*,\cross\}$. We only consider the case where $\circ = {*}$. Assume
that $(\rho^*,\chi,D) \in C_\xi$ and that $\eta x^{\diam{\delta^\circ}
\varphi} . B$ is a subformula of
\[
 \itrg{\chi}{\rho^*}(D) = \mu x^{\diam{\rho^*} \chi} . D \lor
\itrg{\diam{\rho^*}\varphi}{\rho}(x^{\diam{\rho^*} \chi}).
\]
There are three possibilities how the latter might be the case.

First, $\eta x^{\diam{\delta^\circ} \varphi} . B$ might be equal to
$\itrg{\chi}{\rho^*}(D)$. In that case $\eta x^{\diam{\delta^\circ}
\varphi} . B$ is of the required shape and $(\rho^*,\chi,D) \in C_\xi$
holds by assumption.

Second, it might be that $\eta x^{\diam{\delta^\circ} \varphi} . B$ is a
subformula of $D$, in which case we are done immediately.

Third and last, $\eta x^{\diam{\delta^\circ} \varphi} . B$ might be a
subformula of $\itrg{\diam{\rho^*}\chi}{\rho}(x^{\diam{\rho^*} \chi})$.
In that case we can apply (\ref{itm:ih for game terms}) of the inductive
hypothesis to $\rho$ because $(\rho,\diam{\rho^*}\chi,x^{\diam{\rho^*}
\chi}) \in C_\xi$ follows from the assumption that $(\rho^*,\chi,D) \in
C_\xi$ together with the closure conditions from the definition of
$C_\xi$. Hence, it follows that either $\eta x^{\diam{\delta^\circ}
\varphi} . B$ is a subformula of $x^{\diam{\rho^*} \chi}$, which is
impossible, or that there is some $A$ such that $\eta
x^{\diam{\delta^\circ} \varphi} . B = \itrg{\delta^\circ}{\varphi}(A)$
and, furthermore, $(\delta^\circ,\varphi,A) \in C_\xi$. . The latter is
exactly what we have to show.

\item For the case where $\gamma = \gamma_1 \comp \gamma_2$ assume that
$\eta x^{\diam{\delta^\circ} \varphi} . B$ is a subformula of
$\itrg{\chi}{\gamma_1 \comp \gamma_2}(D) =
\itrg{\diam{\gamma_2}\chi}{\gamma_1}(\itrg{\chi}{\gamma_2}(D))$ and that
$(\gamma_1 \comp \gamma_2,\chi,D) \in C_\xi$. From the latter it follows
with the definition of $C_\xi$ that $(\gamma_1,\diam{\gamma_2} \chi,
\itrg{\chi}{\gamma_2}(D)) \in C_\xi$. Hence, we can apply (\ref{itm:ih
for game terms}) of the induction hypothesis to $\gamma_1$ and obtain
that either $\eta x^{\diam{\delta^\circ} \varphi} . B$ is already of the
required shape or it is a subterm of $\itrg{\chi}{\gamma_2}(D)$. In the
latter case we can use the fact that $(\gamma_2,\chi,D) \in C_\xi$,
which also follows from $(\gamma_1 \comp \gamma_2,\chi,D) \in C_\xi$, to
apply the induction hypothesis again, this time to $\gamma_2$, and
obtain that $\eta x^{\diam{\delta^\circ} \varphi} . B$ is either of the
required shape or that it is a subterm of $D$. This is precisely what we
need to show.

\item Of the cases where $\gamma = \alpha?$ or $\gamma = \alpha!$ we
just consider the former because the latter goes analogously. Assume
that $\eta x^{\diam{\delta^\circ} \varphi} . B$ is a subformula of
$\itrg{\chi}{\alpha ?}(D) = \tr{\alpha} \lor D$ and that $(\alpha ?,
\chi,D) \in C_\xi$. The former means that $\eta x^{\diam{\delta^\circ}
\varphi} . B$ is either a subformula of $D$, in which case we are done,
or it is a subformula of $\tr{\alpha}$. But if $\eta
x^{\diam{\delta^\circ} \varphi} . B$ is a subformula of $\tr{\alpha}$ we
can apply (\ref{itm:ih for formulas}) to the subformula $\alpha$ of
$\xi$ to conclude that $\eta x^{\diam{\delta^\circ} \varphi} . B$ must
be of the required shape.
\end{itemize}
\end{proof}

\begin{lemma} \label{tau constrained}
\rem{ Definition~\ref{d:tr} is well-defined if we define
$\itrg{\varphi}{\gamma}$ to be a function that accepts only formulas $A$
as argument that satisfy the constraint that all free variables in $A$
are of the form $x^{\diam{\delta^\circ} \psi}$ such that $\delta$ is a
strict superterm of $\gamma$.
}
Let $\xi$ be a game logic formula, let $(\gamma,\phi,A)$ be a recursive call of $\tau$ on $\xi$ and let 
 $x^{\diam{\delta^\circ} \psi}$ be a bound variable of $\tr{\xi}$. If $x^{\diam{\delta^\circ} \psi}$ is free in $A$,
then $\gamma$ is a proper subterm of $\delta^\circ$.
\end{lemma}
\begin{proof}
We prove this by least fixpoint induction: let  $X$ be the set of all call triples $(\gamma,\phi,A)$ such that, for every bound variable $x^{\diam{\delta^\circ} \psi}$ of $\tr{\xi}$ such that $x^{\diam{\delta^\circ} \psi}$ is free in $A$, we have that $\gamma$ is a strict subterm of $\delta$.  We show that $c_\xi(X) \subseteq X$, hence $C_\xi \subseteq X$ as required.

So suppose that the triple $t$ is in $c_\xi(X)$. We make a case distinction:
\begin{itemize}
\item $t = (\gamma,\varphi,\tr{\phi})$ for some subformula $\diam{\gamma}\varphi$ of $\xi$. Then $t \in X$ since no bound variable of $\xi$ appears free in $\tr{\phi}$, which means that the defining condition of the set $X$ holds trivially. 
\item $t$ is of the form $(\gamma,\phi,A)$ or $(\gamma',\phi,A)$ where $ (\gamma \gcha \gamma', \phi, A) \in X$. Given a bound variable $x^{\diam{\delta^\circ} \psi}$ of $\tr{\xi}$ such that $x^{\diam{\delta^\circ} \psi}$ is free in $A$, since $ (\gamma \gcha \gamma', \phi, A) \in X$ it follows that $\gamma \gcha \gamma'$ is a proper subterm of $\delta^\circ$. So clearly this also holds for both $\gamma$ and $\gamma'$. Hence $(\gamma,\phi,A) \in X$ and $(\gamma',\phi,A) \in X$ as required. 
\item The case where $t$ is of the form $(\gamma,\phi,A)$ or $(\gamma',\phi,A)$ for $ (\gamma \gchd \gamma', \phi, A) \in X$ is similar.
\item Suppose $t$ is of the form $(\gamma, \diam{\gamma^\cross} \phi,x^{ \diam{\gamma^\cross} \phi})$ and there exists some $A$ for which $(\gamma^\cross,\phi,A) \in X$. Given a bound variable $x^{\diam{\delta^\circ} \psi}$ of $\tr{\xi}$, the only possible way that $x^{\diam{\delta^\circ} \psi}$ can appear free in $x^{ \diam{\gamma^\cross} \phi}$ is if $x^{\diam{\delta^\circ} \psi} = x^{ \diam{\gamma^\cross} \phi}$, so ${\diam{\delta^\circ} \psi} = { \diam{\gamma^\cross} \phi}$ and therefore $\delta^\circ = \gamma^\cross$. Since $\gamma$ is a proper subterm of $\gamma^\cross$, this means that the required conclusion holds. 
\item The case where $t$ is of the form $(\gamma, \diam{\gamma^*} \phi,x^{ \diam{\gamma^*} \phi})$  is similar. 
\item Suppose $t$ is of the form $(\gamma',\phi,A)$ for some $\gamma,\gamma',\phi$ and $A$ such that  $((\gamma \comp \gamma'),\phi,A) \in X$. If  $x^{\diam{\delta^\circ} \psi}$ is a bound variable of $\tr{\xi}$ that appears free in $A$, then since $((\gamma \comp \gamma'),\phi,A) \in X$  it follows that $\gamma \comp \gamma'$ is  a proper subterm of $\delta^\circ$. Hence, so is $\gamma'$. 
\item Finally, suppose $t$ is of the form  $(\gamma,\diam{\gamma'}\phi,\itrg{\gamma'
}{\phi}(A))$  for some $\gamma,\gamma',\phi$ and $A$ such that  $((\gamma \comp \gamma'),\phi,A) \in X$. If  $x^{\diam{\delta^\circ} \psi}$ is a bound variable of $\tr{\xi}$ that appears free in $\itrg{\gamma'
}{\phi}(A)$, then by Lemma \ref{l:no free in tau}  $x^{\diam{\delta^\circ} \psi}$ appears free in $A$ as well. Since $((\gamma \comp \gamma'),\phi,A) \in X$  it follows that $\gamma \comp \gamma'$ is  a proper subterm of $\delta^\circ$, hence so is $\gamma$. 
\end{itemize}
We have shown that $t \in c_\xi (X)$ implies $t \in X$, so the proof is finished.  
\end{proof}

\begin{lemma} \label{l:compatible orders}
Let $\xi$ be a game logic formula and 
$\diam{\gamma^\circ} \varphi, \diam{\delta^\dagger} \psi \in F(\xi)$,
where $\circ,\dagger \in \{{*},\cross\}$.
If
$x^{\diam{\delta^\dagger} \psi} \lessAL[\tr{\xi}] x^{\diam{\gamma^\circ} \varphi}$
then
$\diam{\delta^\dagger} \psi \lessfp[\xi] \diam{\gamma^\circ} \varphi$.
\end{lemma}
\begin{proof}
We consider the case that $\gamma^\circ=\gamma^\cross$. The case where $\circ=*$ is similar.

Assume that the variable $x^{\diam{\delta^\dagger} \psi}$ is free
in some subformula $\nu x^{\diam{\gamma^\cross} \varphi} . B$ of
$\tr{\xi}$. By Lemma~\ref{l:variables and recursive calls},
 there exists some $A$ such that 
$$ \nu x^{\diam{\gamma^\cross} \varphi} . B = \itrg{\gamma^\cross}{\varphi}(A) $$
 and furthermore, 
$(\gamma^\cross,\varphi,A) \in C_\xi$. 
But we have:
\[
 \itrg{\varphi}{\gamma^\cross}(A) = \nu x^{\diam{\gamma^\cross} \varphi} .
A \land \itrg{\diam{\gamma^\cross}\varphi}{\gamma}(x^{\diam{\gamma^\cross}
\varphi}),
\]
so 
\[
\nu x^{\diam{\gamma^\cross} \varphi} . B = \nu x^{\diam{\gamma^\cross} \varphi} .
A \land \itrg{\diam{\gamma^\cross}\varphi}{\gamma}(x^{\diam{\gamma^\cross}
\varphi}),
\]

so the variable $x^{\diam{\delta^\dagger} \psi}$ must occur freely in either $A$ or in $\itrg{\diam{\gamma^\cross}\varphi}{\gamma}(x^{\diam{\gamma^\cross} \varphi})$.
If the variable $x^{\diam{\delta^\dagger} \psi}$ occurs freely in
$\itrg{\diam{\gamma^\cross}\varphi}{\gamma}(x^{\diam{\gamma^\cross} \varphi})$
then by Lemma~\ref{l:no free in tau},
$\diam{\gamma^\cross} \varphi = \diam{\delta^\dagger} \psi$, hence $x^{\diam{\gamma^\cross} \varphi} = x^{\diam{\delta^\dagger} \psi}$. But this contradicts our assumption that $x^{\diam{\delta^\dagger} \psi} \lessAL[\tr{\xi}] x^{\diam{\gamma^\cross} \varphi}$, since $\lessAL[\tr{\xi}]$ is an irreflexive relation. 

So the only possibility is that $x^{\diam{\delta^\dagger} \psi}$ occurs freely in $A$.
Since $(\gamma^\cross,\varphi,A) \in C_\xi$ was a recursive call of $\tau$ on $\xi$, it follows from Lemma~\ref{tau constrained}
that $\gamma^\cross$ is a proper subterm of $\delta^\dagger$. Hence $\diam{\delta^\dagger} \psi \lessfp[\xi]  \diam{\gamma^\cross} \varphi$ as required.
\end{proof}

\begin{proofof}{Proposition~\ref{p:translation}}
From Lemma~\ref{l:compatible orders} it follows that $\tr{\xi}$ is
locally well-named: If $\tr{\xi}$ was not locally well-named then the transitive
closure $\lesstrAL[\tr{\xi}]$ of $\lessAL[\tr{\xi}]$ would be reflexive
and hence by Lemma~\ref{l:compatible orders} the relation $\lessfp$
would be reflexive as well. But $\lessfp$ is irreflexive,
since the strict subterm relation is irreflexive.

That $x^\varphi \leqAL[\tr{\xi}] x^\psi$ implies $\phi \leqfp \psi$
follows also from Lemma~\ref{l:compatible orders} because
$\leqAL[\tr{\xi}]$ is the reflexive and transitive closure of
$\lessAL[\tr{\xi}]$ and $\leqfp$ is transitive and reflexive, as it is
defined via the subterm relation.
\qed
\end{proofof}

\subsection{Omitted proofs of Section \ref{sec:s-c}}


We prove completeness for $\HilbP$ from completeness of $\GLms$
by going via an intermediate Hilbert system
\HilbFull for the language $\langFull$.
Note that $\langPar \subseteq \langFull$. 
The system \HilbFull is defined as the extension of \HilbP with 
the axioms and rules listed in Figure~\ref{fig:pfFullRu} below. 

 \begin{figure}[htb]
\fbox{
\begin{minipage}[t]{.47\textwidth}
  \begin{minipage}[t]{.5\textwidth}
\quad\textbf{Additional axioms:}
  \begin{itemize}
  \item $\diam{\gamma \gchd \delta}\phi \lra \diam{\gamma}\phi \land \diam{\delta}\phi$
  \item $\diam{\gamma^\cross}\phi \lra \phi \land \diam{\gamma}\diam{\gamma^\cross}\phi$
  \item $\diam{{\psi}!}\phi \lra \psi \lor \phi$
  \end{itemize}
    \end{minipage}
  \begin{minipage}[t]{.47\textwidth}
\quad\textbf{Additional rule:}
    \begin{prooftree}
      \AxiomC{$\phi \to\diam{\gamma}\phi$}
      \RightLabel{\RuBdi}  
      \UnaryInfC{$\phi \to \diam{\gamma^\cross}\phi$}
    \end{prooftree}
  \end{minipage}
\end{minipage}
}
\caption{Additional axioms and rules of $\HilbFull$.}
\label{fig:pfFullRu}
\end{figure}



The next lemma shows that 
$\HilbFull$ is a conservative extension of $\HilbP$.

\begin{lemma}
\label{l:HilbFull-to-HilbP-pari}
\label{c:pari}
 For all $\phi\in\langFull$, if $\HilbFull \vdash \phi$ then $\HilbP \vdash \pari{\phi}$.
\end{lemma}
\begin{proof}
We need to check that the $\pari{}$-translation of every axiom of $\HilbFull$ is derivable in $\HilbP$, and that the $\pari{}$-translation of every instance of a rule of $\HilbFull$ is admissible in $\HilbP$. We shall allow defined propositional connectives like $\to, \lra, \land$ as abbreviations here. 

\textbf{Case:}  $\diam{\gamma \gchd \delta}\phi \lra \diam{\gamma}\phi \land \diam{\delta}\phi$. The translation of this axiom becomes:
$$\diam{(\pari{\gamma}^d \gcha \pari{\delta}^d)^d} \pari{\phi} \lra  \diam{\pari{\gamma}}\pari{\phi} \land \diam{\pari{\delta}}\pari{\phi}$$
We derive this through the following chain of provable equivalences in \HilbP:

\begin{eqnarray*}
 & \diam{(\pari{\gamma}^d \gcha \pari{\delta}^d)^d} \pari{\phi} \\
  \Leftrightarrow &  \neg \diam{\pari{\gamma}^d \gcha \pari{\delta}^d} \neg \pari{\phi} \\
  \Leftrightarrow   & \neg (\diam{\pari{\gamma}^d} \neg{\pari{\phi}} \vee \diam{\pari{\delta}^d} \neg \pari{\phi}) 
\\  \Leftrightarrow  &  \neg \diam{\pari{\gamma}^d} \neg{\pari{\phi}} \wedge \neg \diam{\pari{\delta}^d} \neg \pari{\phi}
\\  \Leftrightarrow  &  \neg \neg \diam{\pari{\gamma}}\neg  \neg{\pari{\phi}} \wedge \neg \neg \diam{\pari{\delta}}\neg \neg \pari{\phi})
\\  \Leftrightarrow  &  \diam{\pari{\gamma}}{\pari{\phi}} \wedge  \diam{\pari{\delta}}\pari{\phi})
\end{eqnarray*}

\textbf{Case:}
$\diam{\gamma^\cross}\phi \lra \phi \land \diam{\gamma}\diam{\gamma^\cross}\phi$. The translation of this axiom becomes:
$$\diam{((\pari{\gamma}^d)^*)^d} \pari{\phi} \lra \pari{\phi} \land \diam{\pari{\gamma}}\diam{((\pari{\gamma}^d)^*)^d} \pari{\phi}$$
We derive this through the following chain of provable equivalences in \HilbP:
\begin{eqnarray*}
& \diam{((\pari{\gamma}^d)^*)^d} \pari{\phi} \\
 \Leftrightarrow & \neg \diam{(\pari{\gamma}^d)^*} \neg \pari{\phi} \\
 \Leftrightarrow & \neg (\neg \pari{\phi} \vee  \diam{\pari{\gamma}^d}  \diam{(\pari{\gamma}^d)^*} \neg \pari{\phi}) \\
 \Leftrightarrow & \neg (\neg \pari{\phi} \vee  \neg \diam{\pari{\gamma}} \neg \diam{(\pari{\gamma}^d)^*} \neg \pari{\phi}) \\
 \Leftrightarrow & \pari{\phi} \land \diam{\pari{\gamma}} \neg \diam{(\pari{\gamma}^d)^*} \neg \pari{\phi}  \\
  \Leftrightarrow &   \pari{\phi} \land \diam{\pari{\gamma}}\diam{((\pari{\gamma}^d)^*)^d} \pari{\phi}
\end{eqnarray*}
\textbf{Case:}
 $\diam{{\psi}!}\phi \lra \psi \lor \phi$
The translation of this axiom becomes:
$$\diam{({\neg \pari{\psi}}?)^d}\pari{\phi} \lra \pari{\psi} \lor \pari{\phi}$$
We derive this through the following chain of provable equivalences in \HilbP:
\begin{eqnarray*}
& \diam{({\neg \pari{\psi}}?)^d}\pari{\phi} \\
\Leftrightarrow &  \neg \diam{\neg \pari{\psi}?}\neg \pari{\phi} \\
\Leftrightarrow &  \neg (\neg \pari{\psi} \land  \neg \pari{\phi}) \\'
\Leftrightarrow & \pari{\psi} \lor \pari{\phi}
\end{eqnarray*}
\textbf{Case:}
      $$\frac{\phi \to\diam{\gamma}\phi}{\phi \to \diam{\gamma^\cross}\phi}$$
We show that the translation of this rule is derivable in $\HilbP$.
It then follows that it is admissible.
The translation of the premise of this rule is $\pari{\phi} \to \diam{\pari{\gamma}}\pari{\phi}$, and the conclusion becomes $\pari{\phi} \to \diam{((\pari{\gamma}^d)^*)^d}\pari{\phi}$. So suppose that $\HilbP \vdash \pari{\phi} \to \diam{\pari{\gamma}}\pari{\phi}$. Then $\HilbP \vdash \neg \diam{\pari{\gamma}}\pari{\phi} \to \neg \pari{\phi}$, which gives $\HilbP \vdash\diam{\pari{\gamma}^d} \neg \pari{\phi} \to \neg \pari{\phi}$. Bar Induction gives:
$$\HilbP \vdash\diam{(\pari{\gamma}^d)^*} \neg \pari{\phi} \to \neg \pari{\phi}$$
Contraposition now gives:
$$\HilbP \vdash \pari{\phi} \to \neg \diam{(\pari{\gamma}^d)^*}\neg \pari{\phi},$$ and the desired conclusion follows by the equivalence $\neg \diam{(\pari{\gamma}^d)^*}\neg \pari{\phi} \Leftrightarrow\diam{((\pari{\gamma}^d)^*)^d}\pari{\phi}$. This concludes the proof. 
\end{proof}

The convenience of working in a Hilbert system for the full language
is that we can prove the following lemma.

\begin{lemma}
\label{l:normalform-eq}
\label{p:normalform-eq}
For all $\phi \in \langFull$, 
$\HilbFull \vdash \phi \lra \dnnf{\phi}$.
\end{lemma}
\begin{proof}
Straightforward induction on the complexity of formulas.  
\end{proof}

Since we eventually want to connect $\GLms$-provability of $\dnnf{\phi}$
with $\HilbP$-provability of $\phi$ via $\HilbFull$,
the following proposition takes care of one of the steps. 

\begin{proposition}\label{p:HilbFull-to-HilbP}
For all $\phi \in \langPar$,
if $\HilbFull \vdash \dnnf{\phi}$  then $\HilbP \vdash \phi$. 
\medskip
\end{proposition}
\begin{proof}
  Due to Lemma~\ref{l:normalform-eq}, it suffices to show that
  for all $\phi \in \langPar$,
  if $\HilbFull \vdash \phi$  then $\HilbP \vdash \phi$.
  This follows from Lemma~\ref{l:HilbFull-to-HilbP-pari}.
  since, if $\phi \in \langPar$ then $\pari{\phi} = \phi$.
\end{proof}


We now show that we can translate $\GLms$-derivations into
$\HilbFull$-derivations.

\begin{proposition}
\label{p:GLms-to-HilbNorm}
\label{glms-to-hilbnorm}
For all sequents $\Phi \subseteq\langNorm$,
\[
\GLms \vdash \Phi \quad\text{ implies }\quad \HilbFull \vdash \bigvee \Phi.
\]
Consequently, for all $\xi\in\langNorm$,
if $\GLms \vdash \xi$ then $\HilbFull \vdash \xi$.  
\end{proposition}
\begin{proof}
To prove this proposition, we need to prove that the disjunction of any axiom of
\GLms is derivable in \HilbFull, and that every rule $\mathsf{R}$ of \GLms is 
admissible in \HilbFull in the following sense: if the disjunction of the 
premise of an instance of $\mathsf{R}$ is derivable in \HilbFull, then so is 
the disjunction of the corresponding conclusion. 
Axioms of \GLms are taken care of by Lemma~\ref{l:normalform-eq}, since 
they are of the form $\Phi,\ol{\Phi}$, and the disjunction of such a sequent is 
the $\dnnf{}$-translation of a propositional tautology. 
For admissibility of the $\GLms$-rules $\RuIndS$ and $\RuMonDgame$
is shown in Lemmas~\ref{l:RuMonDgame-admiss}~and~\ref{l:adminds} below.
The other cases are straightforward, and we leave the details to the reader.
\end{proof}

\begin{lemma}
\label{l:adminds}
The rule $\RuIndS$ is admissible in the system $\HilbFull$: if 
$\HilbFull \vdash \bigvee \Gamma \vee (\phi \wedge \diam{\gamma}\diam{ (\nnf{\Gamma}!\comp \gamma)^\cross}\diam{\nnf{\Gamma}!}\phi)$
then $\HilbFull \vdash \bigvee \Gamma \vee \diam{\gamma^\cross} \phi$ as well.
\end{lemma}
\begin{proof}
Suppose that:  $$\HilbFull \vdash \bigvee \Gamma \vee (\phi \wedge
\diam{\gamma}\diam{ (\overline{\Gamma} ! \comp
\gamma)^\cross}\diam{\overline{\Gamma}!}\phi)$$
We show that $\HilbFull \vdash \bigvee \Gamma \vee \diam{\gamma^\cross} \phi$. Note that we can rewrite the assumption as: 
$$(*) \quad \HilbFull \vdash \overline{ \Gamma} \to (\phi \wedge
\diam{\gamma}\diam{ (\overline{\Gamma}!\comp
\gamma)^\cross}\diam{\overline{\Gamma}!}\phi)$$
and the desired conclusion as:
$\HilbFull \vdash \overline{ \Gamma} \to \diam{\gamma^\cross} \phi$.
\begin{claim}
\label{claim:nummerett}

$\HilbFull \vdash \overline{\Gamma} \to \phi \wedge \diam{\gamma}\diam{
(\overline{\Gamma}!\comp \gamma)^\cross}\phi$
\end{claim}
\begin{pfclaim}
By our assumption $(*)$ we get $\HilbFull \vdash \overline{\Gamma} \to
\phi$, and since $\HilbFull \vdash \diam{\overline{\Gamma}!}\phi
\leftrightarrow \overline{\Gamma} \vee \phi$ we get $\HilbFull \vdash
\diam{\overline{\Gamma}!}\phi \to \phi$. Together with Monotonicity applied to the consequent in $(*)$, we get
$$\HilbFull \vdash \overline{\Gamma} \to \phi \wedge \diam{\gamma}\diam{
(\overline{\Gamma}!\comp \gamma)^\cross}\phi$$
as required. 
\end{pfclaim}
\begin{claim}
\label{claim:nummertva}
$\HilbFull \vdash \diam{(\overline{\Gamma} !\comp \gamma)^\cross} \to
\diam{\gamma^\cross}\diam{(\overline{\Gamma}! \comp \gamma)^\cross} $
\end{claim}
\begin{pfclaim}
First, by simply unfolding the fixpoint and applying propositional reasoning we get:
$$\HilbFull \vdash \diam{(\overline{\Gamma}! \comp \gamma)^\cross} \to
\diam{\overline{\Gamma}! \comp \gamma}\diam{(\overline{\Gamma}! \comp \gamma)^\cross}$$
But the consequent of this implication is equivalent to:
$$\overline{\Gamma} \vee \diam{\gamma} \diam{(\overline{\Gamma}! \comp \gamma)^\cross}$$
By Claim \ref{claim:nummerett} applied to the left disjunct of this formula we get:
$$\HilbFull \vdash \diam{(\overline{\Gamma}! \comp \gamma)^\cross} \to
\diam{\gamma}\diam{(\overline{\Gamma}! \comp \gamma)^\cross}$$
By the Bar Induction rule we get:
$$\HilbFull \vdash \diam{(\overline{\Gamma}! \comp \gamma)^\cross} \to
\diam{\gamma^\cross}\diam{(\overline{\Gamma}! \comp \gamma)^\cross} $$
as required. 
\end{pfclaim}

\begin{claim}
\label{claim:nummertre}
$\HilbFull \vdash \overline{\Gamma} \to \diam{(\overline{\Gamma}!\comp \gamma)^\cross} \phi$
\end{claim}
\begin{pfclaim}
By Claim \ref{claim:nummerett} we get $ \vdash \overline{\Gamma} \to
\phi \wedge \diam{\gamma}\diam{ (\overline{\Gamma}!\comp \gamma)^\cross}\phi$. By propositional reasoning we have:
$$\HilbFull \vdash \diam{\gamma}\diam{ (\overline{\Gamma}!\comp
\gamma)^\cross}\phi \to \overline{\Gamma} \vee \diam{\gamma}\diam{
(\overline{\Gamma}!\comp \gamma)^\cross}\phi$$
But the consequent of this implication is equivalent to
$\diam{\overline{\Gamma}! \comp \gamma}\diam{ (\overline{\Gamma}!\comp \gamma)^\cross}\phi$, so we get:
$$\vdash \overline{\Gamma} \to \phi \wedge \diam{\overline{\Gamma}!
\comp \gamma}\diam{ (\overline{\Gamma}!\comp \gamma)^\cross}\phi$$
But the consequent of this implication is just the unfolding of the
fixpoint $\diam{ (\overline{\Gamma}!\comp \gamma)^\cross}\phi$, so we get:
$$\HilbFull \vdash \overline{\Gamma} \to \diam{(\overline{\Gamma}!\comp \gamma)^\cross} \phi$$
as required. 
\end{pfclaim}
\begin{claim}
\label{claim:nummerfyra}
$\HilbFull \vdash \diam{(\overline{\Gamma}!\comp \gamma)^\cross} \phi \to \phi$
\end{claim}
\begin{pfclaim}
Just unfold the fixpoint  $\diam{(\overline{\Gamma}!\comp
\gamma)^\cross} \phi$ to $\phi \wedge  \diam{\overline{\Gamma}\comp
\gamma}  \diam{(\overline{\Gamma}!\comp \gamma)^\cross} \phi$.
\end{pfclaim}
We can now prove the lemma: by Claim \ref{claim:nummertre} we have:
$$(\text{i}) \quad \HilbFull \vdash \overline{\Gamma} \to
\diam{(\overline{\Gamma}!\comp \gamma)^\cross} \phi$$
Combining $(\text{i})$ with Claim \ref{claim:nummertva} we get:
$$(\text{ii}) \quad \HilbFull \vdash \overline{\Gamma} \to
\diam{\gamma^\cross}\diam{(\overline{\Gamma}!\comp \gamma)^\cross} \phi$$
By Claim \ref{claim:nummerfyra} and Monotonicity we get:
$$(\text{iii})\quad \HilbFull \vdash
\diam{\gamma^\cross}\diam{(\overline{\Gamma}!\comp \gamma)^\cross} \phi \to \diam{\gamma^\cross}\phi$$
Putting $(\text{ii})$ and $(\text{iii})$ together we get:
$$\HilbFull \vdash \overline{ \Gamma} \to \diam{\gamma^\cross} \phi$$
as required.
\end{proof}

\begin{lemma}\label{l:RuMonDgame-admiss}
The rule $\RuMonDgame$ is derivable in the system \HilbFull.
\end{lemma}
\begin{proof}
We shall prove this by induction on the complexity of formulas in dual normal form. To keep notation simple, in this proof we abbreviate $\HilbFull \vdash \phi$ by $\vdash \phi$, and $\HilbFull \vdash \varphi \to \psi$ by $\varphi \vdash \psi$. The induction hypothesis on a formula $\varphi(\delta)$ is that $ \varphi(\delta) \vdash \varphi({\chi}!\comp \delta)$. 

\begin{claim}
\label{c:gameterminduction}
For every game term $\gamma(\delta)$ in dual normal form and every  formula $\phi$, we can prove the following implication in \HilbFull:
$$\diam{\gamma(\delta)} \phi \to \diam{\gamma({\chi}!\comp \delta)}\phi $$
provided that the main induction hypothesis holds for every formula $\theta$ corresponding to a  subterm ${\theta}!$ or ${\theta}?$ of $\gamma(\delta)$.  
\end{claim}
\begin{pfclaim}
We prove that the implication holds for all $\phi$ by induction on the complexity of the game term $\gamma(\delta)$, treating $\delta$ as an atomic case.
 
Atomic case, $\gamma(\delta) =  \delta$. For all $\phi$ we have:
$$\diam{\delta} \phi \vdash \chi \vee \diam{\delta}\phi \vdash \diam{{\chi}!}\diam{\delta}\phi \vdash \diam{{\chi}!\comp \delta}\phi$$
as required. 
\\\\
Atomic case for game terms $g$ or $g^d$: trivial.
\\\\
Case for $\theta(\delta)?$ or $\theta(\delta)!$: follows immediately from the induction hypothesis on $\theta$, since $\vdash \diam{\theta(\delta)?} \phi \lra \theta(\delta) \land \phi$ and $\vdash \diam{\theta(\delta)!} \phi \lra \theta(\delta) \lor \phi$.
\\\\
Case for $\gcha:$ the induction hypothesis on the subterms $\gamma_1(\delta)$ and $\gamma_2(\delta)$ of $(\gamma_1 \gcha \gamma_2)(\delta)$ gives $\vdash \diam{\gamma_1(\delta)}\phi \to  \diam{\gamma_1({\chi }! \comp \delta)} \phi $ and $\vdash \diam{\gamma_2(\delta)}\phi \to  \diam{\gamma_2({\chi }! \comp \delta)} \phi $. We get:
\begin{eqnarray*}
\diam{\gamma_1(\delta) \gcha \gamma_2(\delta)}\phi & \vdash &
\diam{\gamma_1(\delta)} \phi \vee \diam{ \gamma_2(\delta)}\phi \\
& \vdash & \diam{\gamma_1({\chi}! \comp \delta)} \phi \vee \diam{{\chi}! \comp  \gamma_2(\delta)}\phi \\
&  \vdash & \diam{\gamma_1({\chi}! \comp \delta) \gcha \gamma_2({\chi}! \comp \delta)}\phi 
\end{eqnarray*}
as required.
\\\\
Case for $\gchd$: similar. 
\\\\
Case for $\comp$: consider the formula $\diam{\gamma_1(\delta) \comp \gamma_2(\delta)}\phi$. The induction hypothesis on $\gamma_2(\delta)$ instantiated for the formula $\phi$ gives 
$$\vdash \diam{\gamma_2(\delta)}\phi \to \diam{\gamma_2({\chi}! \comp \delta)} \phi$$
By monotonicity we get:
$$\vdash \diam{\gamma_1(\delta)}\diam{\gamma_2(\delta)}\phi \to  \diam{\gamma_1(\delta)}\diam{\gamma_2({\chi}! \comp \delta)} \phi$$
But the induction hypothesis on $\gamma_1(\delta)$ instantiated for the formula $\diam{\gamma_2({\chi}! \comp \delta)} \phi$ gives:
 $$ \diam{\gamma_1(\delta)}\diam{\gamma_2({\chi}! \comp \delta)} \phi \to   \diam{\gamma_1({\chi}! \comp\delta)}\diam{\gamma_2({\chi}! \comp \delta)} \phi$$
Putting these implications together we get:
$$\vdash \diam{\gamma_1(\delta)}\diam{\gamma_2(\delta)}\phi \to   \diam{\gamma_1({\chi}! \comp\delta)}\diam{\gamma_2({\chi}! \comp \delta)} \phi$$
The required result now follows from the reduction axioms for $\comp$ applied to both the antecedent and the consequent in this implication.
\\\\
Case for $*$: by the induction hypothesis we have, for every formula $\psi$, $\vdash \diam{\gamma(\delta)}\psi \to \diam{\gamma({\chi}!\comp \delta)}\psi $. We wish to show that for all $\phi$, we have $\vdash \diam{\gamma(\delta)^*}\phi \to \diam{\gamma({\chi}!\comp \delta)^*}\phi $. By the induction hypothesis instantiated with the formula $\psi = \diam{\gamma({\chi}!\comp \delta)^*}\phi $ we have 
$$\vdash \diam{\gamma(\delta)}\diam{\gamma({\chi}!\comp \delta)^*}\phi  \to \diam{\gamma({\chi}!\comp \delta)}\diam{\gamma({\chi}!\comp \delta)^*}\phi  $$
But by the unfolding axiom for angelic iteration and propositional reasoning we get:  $$\vdash \diam{\gamma({\chi}!\comp \delta)}\diam{\gamma({\chi}!\comp \delta)^*}\phi \to  \diam{\gamma({\chi}!\comp \delta)^*}\phi $$
Hence, putting these two implications together, we get:
$$ \vdash \diam{\gamma(\delta)}\diam{\gamma({\chi}!\comp \delta)^*}\phi  \to \diam{\gamma({\chi}!\comp \delta)^*}\phi $$
By the 
 Bar Induction rule for angelic iteration, we now get:
$$(\dagger)\quad \vdash \diam{\gamma(\delta)^*}\diam{\gamma({\chi}!\comp \delta)^*}\phi  \to \diam{\gamma({\chi}!\comp \delta)^*}\phi $$
Applying unfolding and propositional reasoning again, we get:
$$\vdash \phi \to \diam{\gamma({\chi}!\comp \delta)^*}\phi$$
By the monotonicity rule we get:
$$(\ddagger) \quad \vdash \diam{\gamma(\delta)^*}\phi  \to \diam{\gamma(\delta)^*}\diam{\gamma({\chi}!\comp \delta)^*}\phi$$
Putting together the implications $(\dagger)$ and $(\ddagger)$, we get: 
$$\vdash \diam{\gamma(\delta)^*}\phi  \to \diam{\gamma({\chi}!\comp \delta)^*}\phi$$
as required. 
\\\\
Case for $\cross$: 
 by the induction hypothesis we have, for every formula $\psi$, $\vdash \diam{\gamma(\delta)}\psi \to \diam{\gamma({\chi}!\comp \delta)}\psi $. We wish to show that for all $\phi$, we have $\vdash \diam{\gamma(\delta)^\cross}\phi \to \diam{\gamma({\chi}!\comp \delta)^\cross}\phi $. By the induction hypothesis instantiated with the formula $\psi = \diam{\gamma(\delta)^\cross}\phi $ we have 
$$\vdash \diam{\gamma(\delta)}\diam{\gamma(\delta)^\cross}\phi  \to \diam{\gamma({\chi}!\comp \delta )}\diam{\gamma( \delta)^\cross}\phi  $$
But by unfolding $\diam{\gamma(\delta)^\cross}\phi $ to $\phi \wedge \diam{\gamma(\delta)}\diam{\gamma(\delta)^\cross}\phi$, we see that:
$$\vdash  \diam{\gamma(\delta)^\cross}\phi  \to \diam{\gamma(\delta)}\diam{\gamma(\delta)^\cross}\phi$$
Putting together the implications we have established so far, we get:
$$\vdash  \diam{\gamma(\delta)^\cross}\phi  \to \diam{\gamma({\chi}!\comp \delta )}\diam{\gamma( \delta)^\cross}\phi $$
By the Bar Induction rule for $\cross$ we now get:
$$(\dagger) \quad \vdash  \diam{\gamma(\delta)^\cross}\phi  \to \diam{\gamma({\chi}!\comp \delta )^\cross}\diam{\gamma( \delta)^\cross}\phi $$
But we have $\vdash \diam{\gamma( \delta)^\cross}\phi \to \phi$, so by monotonicity we get: 
$$(\ddagger)\quad \vdash \diam{\gamma({\chi}!\comp \delta )^\cross}\diam{\gamma( \delta)^\cross}\phi \to \diam{\gamma({\chi}!\comp \delta )^\cross}\phi$$
Putting together $(\dagger)$ and $(\ddagger)$ we get:
$$\vdash \diam{\gamma( \delta)^\cross}\phi \to \diam{\gamma({\chi}!\comp \delta )^\cross}\phi$$
as required. 

\end{pfclaim}

We can now complete the main induction: the atomic cases for literals and induction steps for $\lor, \land$ are easy. The only interesting step is for a formula of the form $(\diam{\gamma}\phi)(\delta) = \diam{\gamma(\delta)}\phi(\delta)$. By the induction hypothesis on $\phi(\delta)$ we get $\vdash \phi(\delta) \to \phi({\chi}! \comp \delta)$, so by monotonicity we get $$\vdash \diam{\gamma({\chi}!\comp \delta)}\phi(\delta) \to  \diam{\gamma({\chi}!\comp \delta)}\phi({\chi}!\comp \delta)$$ By the induction hypothesis on all subformulas $\theta$ occurring in subterms ${\theta}!$ or ${\theta}?$ of $\gamma(\delta)$, we can apply Claim \ref{c:gameterminduction} and get
$$\vdash  \diam{\gamma( \delta)}\phi(\delta)  \to \diam{\gamma({\chi}!\comp \delta)}\phi(\delta)$$
Putting together these implications we get:
$$\vdash  \diam{\gamma( \delta)}\phi(\delta)  \to \diam{\gamma({\chi}!\comp \delta)}\phi({\chi}! \comp \delta)$$
as required. 
\end{proof}

We are now ready to prove the transformations between
$\GLms$ and $\HilbP$.

\begin{proofof}{Theorem~\ref{thm:GLms-to-HilbP}}
  For item 1, let $\phi \in \langPar$ such that 
  $\GLms \vdash \dnnf{\phi}$.
  By Proposition~\ref{p:GLms-to-HilbNorm},
  $\HilbFull \vdash \dnnf{\phi}$, and by
  Proposition~\ref{p:HilbFull-to-HilbP}, we obtain
  $\HilbP \vdash \phi$.  

  For item 2, let $\xi \in\langNorm$ such that $\GLms \vdash \xi$.
  By Proposition~\ref{p:GLms-to-HilbNorm},
  $\HilbFull \vdash \xi$, and since $\langNorm \subseteq \langFull$,
  we obtain $\HilbP \vdash \pari{\xi}$
  from Lemma~\ref{l:HilbFull-to-HilbP-pari}. 

\end{proofof}



%
%
%

\end{document}